\documentclass[11pt,a4paper]{amsart}
\usepackage{amsthm}
\usepackage{amssymb}
\usepackage{mathtools}
\usepackage{graphicx}
\usepackage{float}  

\usepackage{tikz}
\usetikzlibrary{calc}
\usetikzlibrary {decorations.pathmorphing, decorations.pathreplacing, decorations.shapes,decorations.markings}
\tikzset{%
    add/.style args={#1 and #2}{
        to path={%
 ($(\tikztostart)!-#1!(\tikztotarget)$)--($(\tikztotarget)!-#2!(\tikztostart)$)%
  \tikztonodes},add/.default={.2 and .2}}
}  

\tikzstyle{point}=[draw,circle,fill=black,scale=.4] 

\tikzstyle{square}=[minimum size=.706cm,draw]

\tikzset{ 
  pics/H/.style 2 args= {
    code={
       \begin{pgfonlayer}{bg}    
      \draw[thick,blue] (#1:#2) -- ($(180 + #1:#2)$);
      \end{pgfonlayer}
    }
  }
}

\pgfdeclarelayer{bg}    
\pgfsetlayers{bg,main}  

\usepackage{subcaption}
\usepackage[normalem]{ulem} 

\usepackage{enumerate}
\usepackage{stmaryrd} 
\usepackage{xcolor}
\usepackage{verbatim}

\newcommand{\CC}{\mathbb{C}}
\newcommand{\CP}{\mathbb{P}}
\newcommand{\RR}{\mathbb{R}}
\newcommand{\mr}{\mathrm}
\newcommand{\ra}{\rightarrow}
\newcommand{\OO}{\mathcal{O}}

\newcommand{\dd}{\partial}
\newcommand{\Map}{\mathrm{Map}}
\newcommand{\QMap}{\mathrm{QMap}}
\newcommand{\ZZ}{\mathbb{Z}}
\newcommand{\n}{\mathsf{n}}
\newcommand{\mc}{\mathcal}
\newcommand{\ul}{\underline}
\newcommand{\pr}{\mathrm{pr}}

\newcommand{\PSL}{\mathrm{PSL}}
\newcommand{\FR}{\mathrm{PQL}}
\newcommand{\ev}{\mathrm{ev}}

\newcommand{\Conf}{\mathrm{Conf}}
\newcommand{\codim}{\operatorname{codim}}
\newcommand{\del}{\partial}
\newcommand{\delbar}{\bar\partial}
\newcommand{\qra}{\not\to}
\newcommand{\M}{\mathcal{M}}
\newcommand{\JX}{J^X} 
\newcommand{\JY}{J^Y}
\newcommand{\prX}{\Pi^X}
\newcommand{\prY}{\Pi^Y}
\newcommand{\barprX}{\bar{\Pi}^X}

\newcommand{\KM}{\mathrm{KM}}
\newcommand{\GLSM}{\mathrm{QM}}
\newcommand{\QM}{\mathrm{QM}}
\newcommand{\PQM}{\mathrm{PQM}}
\newcommand{\PQL}{\mathrm{PQL}}
\newcommand{\bt}{\bullet}
\newcommand{\QP}{\mathsf{CQP}}
\newcommand{\GGamma}{\pmb{\Gamma}}
\newcommand{\rk}{\mathrm{rk}}
\newcommand{\Var}{\mathrm{Var}}
\makeatletter
\newcommand*{\defeq}{\mathrel{\rlap{%
                     \raisebox{0.3ex}{$\m@th\cdot$}}%
                     \raisebox{-0.3ex}{$\m@th\cdot$}}%
                     =}
\makeatother
\makeatletter
\newcommand*{\qefed}{=\mathrel{\rlap{%
                     \raisebox{0.3ex}{$\m@th\cdot$}}%
                     \raisebox{-0.3ex}{$\m@th\cdot$}}%
                     }
\makeatother

\newcommand{\bl}{\color{blue}}

\usepackage[hyperindex,breaklinks]{hyperref}

\newtheorem{enumprob}{Enumerative Problem}

\theoremstyle{remark}
\newtheorem{remark}{Remark}[section]
\theoremstyle{plain}
\newtheorem{lemma}[remark]{Lemma}
\newtheorem{proposition}[remark]{Proposition}

\newtheorem{corollary}[remark]{Corollary}
\theoremstyle{definition}
\newtheorem{definition}[remark]{Definition}
\newtheorem{example}[remark]{Example}

\newtheorem{conjecture}[remark]{Conjecture}

\usepackage{fancyhdr}
\fancypagestyle{fancyams}{
    \fancyhead[RE,LO]{}
    \fancyhead[RO,LE]{\scriptsize\thepage}
    \fancyhead[CO]{\scriptsize\MakeUppercase{Freckles and scars}} 
    \fancyhead[CE]{\scriptsize\MakeUppercase{O.Chekeres, S.Kandel, A.Losev, P.Mnev, K.Wernli, and D.R.Youmans}}
    
    \fancyfoot{}
}
\usepackage[automake,toc = false]{glossaries-extra}
\glsxtrRevertMarks
\newglossarystyle{myNoHeaderStyle}{
    \setglossarystyle{longheader} 
}

\makeglossaries

\newglossaryentry{space}
{
  name={},
  sort={},
  description={}
}

\newglossaryentry{X}
{
  name={$X$},
  sort={a},
  description={source manifold, usually $\CP^k$}
}

\newglossaryentry{Y}
{
  name={$Y$},
  sort={b},
  description={target manifold, usually $\CP^n$}
}

\newglossaryentry{QMap}
{
  name={$\QMap_d(X,Y)$},
  sort={c},
  description={space of degree $d$ quasimaps $X \not\to Y$}
}

\newglossaryentry{Map}
{
  name={$\Map_d(X,Y)$},
  sort={d},
  description={space of degree $d$ holomorphic maps $X \to Y$}
}

\newglossaryentry{PQL}
{
name={$\PQL(f)$},
sort={e},
description={proper quasimap locus of a quasimap $f$}
}

\newglossaryentry{D}{
name={$\mc{D}$},
sort={},
description={enumerative data for maps/quasimaps: the collection of source/target cycles $c_1^X,\ldots,c_l^X,c_1^Y,\ldots,c_l^Y$ and the topological type of the map}
}

\newglossaryentry{KM}{
name={$\KM(X,Y,\mc{D})$},
sort={f},
description={number of holomorphic maps $X\ra Y$ subject to conditions $\mc{D}$}
}

\newglossaryentry{QM}{
name={$\QM(X,Y,\mc{D})$},
sort={g},
description={number of quasimaps $X\qra Y$ subject to conditions $\mc{D}$}
}

\newglossaryentry{PQM}{
name={$\PQM(X,Y,\mc{D})$},
sort={h},
description={number of proper quasimaps $X\qra Y$ subject to conditions $\mc{D}$}
}

\newglossaryentry{Var}{
name={$\Var$},
sort={i},
description={``space of variables'' -- the product of the space of quasimaps and the source cycles, see (\ref{eq: def var})}
}

\newglossaryentry{h}
{
  name={$h = c_1(O(1)_{\CP^k})$},
  sort={k},
  description={generator of of $H^2(\CP^k)$}
}

\newglossaryentry{H}
{
  name={$H = c_1(O(1)_{\QMap_d(\CP^k,\CP^n)})$},
  sort={l},
  description={generator of $H^2(\QMap_d(\CP^k,\CP^n))$; note that $\QMap_d(\CP^k,\CP^n) = \CP^{n_{k,n,d}}$}
}

\newglossaryentry{zeta}{
name={$\zeta=c_1(O(1)_Z)$},
sort={m},
description={generator of $H^2(Z)$ for a component $Z\subset \QMap_d(\CP^k,\CP^n)$ of the zero locus of the section $\sigma$ of the equation bundle $E$, assuming $Z$ is a projective space}
}

\newglossaryentry{Euler}{
name={$\tilde{e}(E)$},
sort={n},
description={Euler class of a vector bundle $E$}
}

\newglossaryentry{cE}
{
  name={$c(E)$},
  sort={o},
  description={total Chern class (or a de Rham representative of it) of a vector bundle $E \to B$}
}

\begin{document}

\glsaddall

\title[Freckles and scars
]
{
On enumerative problems for maps and quasimaps: 
freckles and scars
}

\author[O.Chekeres]{Olga Chekeres}

\address{University of Connecticut, Department of Mathematics, Storrs, CT 06269, USA}
\email{olga.chekeres@uconn.edu}

\author[S.Kandel]{Santosh Kandel}
\address{California State University, Sacramento, Department of Mathematics and Statistics, CA 95819, USA}
\email{kandel@csus.edu}

\author[A.Losev]{Andrey Losev}

\address{Wu Wen-Tsun Key Lab of Mathematics, Chinese Academy of Sciences, 
USTC, No.96, JinZhai Road Baohe District, Hefei, Anhui, 230026, P.R.China}
\address{National Research University Higher School of Economics \\
Laboratory of Mirror Symmetry, NRU HSE, 6 Usacheva str., Moscow,  Russia, 119048
}

\email{
aslosev2@yandex.ru
}

\author[P.Mnev]{Pavel Mnev}

\address{University of Notre Dame, Notre Dame, IN 46556, USA}
\address{St. Petersburg Department of V. A. Steklov Institute of Mathematics of the Russian Academy of Sciences, 
27 Fontanka, St. Petersburg, Russia, 191023}
\email{pmnev@nd.edu}

\author[K.Wernli]{Konstantin Wernli}
\address{Centre for Quantum Mathematics, IMADA, University of Southern Denmark, Campusvej 55, 5230 Odense M, Denmark}
\email{kwernli@imada.sdu.dk}

\author[D.R.Youmans]{Donald R.\ Youmans}

\address{Institut für Mathematik, Ruprecht-Karls-Universit\"at Heidelberg, 69221 Heidelberg, Germany}

\email{dyoumans@mathi.uni-heidelberg.de}

\thanks{The work of A.\ S.\ Losev is partially supported by Laboratory of Mirror Symmetry NRU HSE, RF Government grant, ag. 
N\textsuperscript{\underline{o}}  14.641.31.0001.
The work of D.\ R.\ Youmans was supported by the Deutsche Forschungsgemeinschaft (DFG, German Research Foundation) under Germany’s Excellence Strategy EXC 2181/1 - 390900948 (the Heidelberg STRUCTURES Excellence Cluster). The work of K. Wernli was supported by the ERCSyG project, Recursive and Exact New Quantum Theory (ReNewQuantum) which received
funding from the European Research Council (ERC) under the European Union’s Horizon 2020 research and innovation programme
under grant agreement No. 810573}

\begin{abstract}
We address the question of counting maps between projective spaces such that images of cycles on the source intersect cycles on the target. In  this paper we do it by embedding maps into quasimaps that form a projective space of their own. When a quasimap is not a map, it contains freckles (studied earlier) and/or scars,  
appearing when the complex dimension of the source is greater than one. We consider a lot of examples showing that freckle/scar calculus (using excess intersection theory) works. We also propose the ``smooth conjecture'' that may lead to computation of the number of maps by an integral over the space of quasimaps.
\end{abstract}
\maketitle

\setcounter{tocdepth}{3}
\tableofcontents

\section{Introduction}


The goal of this paper is to attract attention to an enumerative problem that generalizes a well-known problem of counting rational curves in a toric manifold. The generalization is to higher-dimensional sources. We are doing the first steps in this project. We propose to study this problem starting with quasimaps, where the count is straightforward.
By subtracting the contribution of non-map (``proper quasimap'') solutions we obtain the number of maps. For one-dimensional source, this approach was discussed in \cite{LNS,LNS2}. 

There are two ways to extract extra contributions from the quasimap count: 
\begin{enumerate}[I.]
\item
Study the proper quasimap configurations themselves. We present here a lot of examples. 
The proper treatment of these configurations requires Fulton's excess intersection theory \cite{Fulton}.
The main formula is:
\begin{equation}
\QM=\KM+\PQM 
\end{equation}
where $\KM$ is the number of holomorphic maps (that we call Kontsevich-Manin number\footnote{
For source of complex dimension one, these numbers are known as Gromov-Witten invariants that were effectively computed by Kontsevich-Manin \cite{KM}.
}) we are interested in. $\QM$ is the (easily computable by a B\'ezout-like formula) total number of quasimaps and $\PQM$ is the count of proper quasimap configurations. 
Proper quasimaps include 
\begin{itemize}
    \item  ``freckle'' configurations (where the evaluation is not defined at a collection of isolated points -- ``freckles'' -- in the source) and
\item ``scar'' configurations (evaluation fails on a cycle of positive dimension on the source) appearing for  $\dim_\CC(\mr{source})>1$.
\end{itemize}
We define scars in Section \ref{sec:freckle_stratification} (see also Example \ref{ex: qmap stratification}) and study examples with them in Sections \ref{sss: unstable example with a scar} and \ref{ss: example: semi-moving points}.

PQM numbers range in complexity. It is easy to treat isolated freckle configurations. A bit more work is required to treat what we call \emph{quasi-stable} examples. However, we found non-quasi-stable examples where the full machinery of excess intersection theory (in particular, Segre classes) is needed. We plan to come to this issue in a subsequent publication.
    \item Study integrals of differential forms over the space of quasimaps that are smooth on the locus of actual maps. Here we experimentally observe two things: 
    \begin{enumerate}[(i)]
    \item 
    such integrals are convergent,
    \item they give correct answers in the simplest cases. 
    \end{enumerate}
    Therefore, we propose the Smooth Conjecture that states that having enough computational power, these numbers could be computed numerically.
\end{enumerate}

As a byproduct, we were looking for a higher analog of quantum multiplication -- a generating function for KM numbers with only 0-dimensional cycles on the source. Surprisingly enough, for source dimension greater than one, there is a quantum ring -- a Frobenius algebra with free commutative product but nontrivial counit -- which does not descend to a deformation of the cohomology ring of the target. 

This paper is intended as a self-contained mathematical text motivated by the problem of  gauged holomorphic models \cite{Gerasimov, CLASH} 
for complex  dimension of the source 1 and 2. The relation of numbers that we find here to physics will be explained elsewhere; here we focus on the mathematical side of the problem. 

\subsection*{Acknowledgments}
D.Y. would like to thank Felipe Espreafico for interesting discussions. P.M. and K.W. would like to thank Galileo Galilei Institute where part of the work was completed for hospitality.

\newpage
\pagestyle{fancyams}
\printunsrtglossary[type=main,title=Notations,nonumberlist,style=myNoHeaderStyle]
\newpage

\section{Projective space of quasimaps between projective spaces}
\subsection{Definition of a quasimap} 

Let $X$ be a compact complex manifold and let
$$Y=\CC^N\sslash G=(\CC^N\backslash\GGamma)/G$$ 
be the GIT quotient of $\CC^N$ by an action of a reductive  algebraic group $G$ by linear transformations.\footnote{\label{footnote: GIT}
 We are thinking here of the projective GIT quotient $\mr{Proj}\left(\Gamma(\CC^N,\oplus_{i\geq 0}\mc{L}^{\otimes i})^G\right)$, with $\mc{L}$ the trivial line bundle over $\CC^N$ equipped with ``linearization'' -- an extension of the $G$-action on the base. Recall that a point $x\in\CC^N$ is semistable if the closure of the $G$-orbit of  $(x,l)\in \mc{L}^\vee$ (with $\mc{L}^\vee$ the dual line bundle) is disjoint from the zero-section for any nonzero element $l\in \mc{L}^\vee_x$. A non-semistable point of $\CC^N$ is ``unstable.''
We refer the reader e.g. to \cite{Proudfoot} and \cite{Thomas} for details on GIT quotients.
} Here $\GGamma\subset \CC^N$ is the unstable locus of the $G$-action. 

\begin{definition}\label{def: quasimaps}
A \emph{quasimap} $f$ from $X$ to $Y$ (we will write $f\colon X\qra Y$) is a pair $(\mc{P},\ul{f})$, consisting of a holomorphic $G$-bundle $\mc{P}$ over $X$ and 
a holomorphic section $\ul{f}$ of the associated vector bundle $\mc{P}\times_G \CC^N $ over $X$,  considered modulo 
automorphisms\footnote{I.e., for $\gamma$ a holomorphic section of the automorphism bundle $\mr{Aut}(\mc{P})$, pairs $(\mc{P},\ul{f})$ and $(\mc{P},\gamma(\ul{f}))$ are understood as equivalent.} of $\mc{P}$ acting on the pair $(\mc{P},\ul{f})$.\footnote{Such pairs are closely related to Brill-Noether pairs, see e.g. \cite{KingNewstead}. The notion of a quasimap is due to V. Drinfeld \cite{Drinfeld}. Also: the definition of a quasimap we give here is the same as a map to the stack quotient $[\CC^N/G]$, see e.g. \cite[Section 2.1]{CFK}.}   Additionally, we assume that $\ul{f}$ is not globally a section of the subbundle $\mc{P}\times_G \GGamma\subset \mc{P}\times_G \CC^N$.
\end{definition}
If the section $\ul{f}$ satisfies 
\begin{equation}\label{f not in Gamma}
\ul{f}(x)\not\in \GGamma \quad \mbox{ for all}\;\; x\in X,
\end{equation} then it defines a holomorphic map $X\ra Y$ (by abuse of notations we will also denote this map $f$). 

If (\ref{f not in Gamma}) fails, we call $f$ a ``proper'' quasimap. Then we call the set
$$ \FR(f)\defeq \{x\in X \;|\; \ul{f}(x)\in \GGamma\}\quad \subset X$$
the 
``proper quasimap locus''
of $f$. If $\FR(f)$ is collection of isolated points, we call these points ``freckles.''\footnote{The term ``freckle'' in this context was introduced in \cite{LNS}.}

We will denote the space of quasimaps $\QMap(X,Y)$, while the space of holomophic maps $X\ra Y$ will be denoted by $\Map(X,Y)$. 
By the discussion above, maps are quasimaps with $\FR(f)=\varnothing$: $\Map(X,Y)\hookrightarrow \QMap(X,Y)$.

\begin{remark} We will mostly discuss the case $G=\CC^*$ or $(\CC^*)^l$ in this paper, with the quotient $Y$ a toric manifold. However, another very interesting example is $G=GL(n,\CC)$ (and $Y$ can be e.g. the Grassmannian), which has a connection to gauge theory with nonabelian gauge group and Nekrasov theory. We will discuss this connection in some detail in a separate paper. 
\end{remark}

\begin{remark}\label{rem: qmaps as tuples of sections}
    If $\CC^N=(\CC^n)^{\times p}$, with $G$-action being the diagonal extension of a $G$-action on $\CC^n$, then a quasimap $X\qra Y$ is the same as a choice of a $G$-bundle $\mc{P}$ over $X$, plus a $p$-tuple of sections of the bundle $\mc{P}\times_G \CC^n$ considered modulo automorphisms of $\mc{P}$ (acting diagonally on the $p$-tuple). E.g., if $G=GL(n,\CC)$, this data is equivalent to a choice of a rank $n$ vector bundle $V$ over $X$ and a $p$-tuple of its sections, considered modulo diagonal action of $\mr{Aut}(V)$.
\end{remark}

\begin{remark}
   Instead of using the language of GIT quotient for the target $Y$, one can alternatively use the language of symplectic (Marsden-Weinstein) reduction of $\CC^N$ by a compact subgroup $G_\mr{cpt}$ of $G$. 
\end{remark}

\subsection{Evaluation map}
A quasimap $f\colon X\qra Y$ 
can be evaluated at a point $x\in X\setminus \FR(f)$  
i.e., one has an evaluation map
\begin{equation}\label{ev qmap}
    \ev\colon (\QMap(X,Y)\times X) \backslash \{(f,x)\;|\; x\not\in \FR(f)\} \;\longrightarrow \; Y.
\end{equation}
In particular, evaluation of a quasimap at a point of $\FR(f)$ (e.g. a freckle) is not defined.
Restricted to holomorphic maps (quasimaps with $\FR(f)=\varnothing$), this is the usual evaluation map
\begin{equation} \label{ev map}
    \ev\colon \Map(X,Y)\times X \ra Y.
\end{equation}

\begin{remark} The evaluation maps (\ref{ev qmap}) and (\ref{ev map}) are invariant under the group $\mr{Aut}(X)$ of holomorphic automorphisms of $X$ acting by $g\cdot(f,x)={(f\circ g^{-1},g(x))}$.  In particular, for $X=\CP^k$, 
the group of automorphisms is the group of ``higher M\"obius transformations'', $\mr{Aut}(X)={\PSL(k+1,\CC)}$.
\end{remark}

Below we give a concrete example (Example \ref{example: evaluation at a freckle}) of how the evaluation at a freckle fails.

\subsection{The main example: quasimaps \texorpdfstring{$\CP^k\qra \CP^n$}{CPk -> CPn}. Explicit formula for \texorpdfstring{$\QMap(\CP^k,\CP^n)$}{QMap(CPk,CPn)}.}
The main example of Definition \ref{def: quasimaps} relevant for this paper will be the following:
\begin{itemize}
    \item $Y$ is the complex projective space $\CP^n$ realized as the GIT quotient $\CC^{n+1}\sslash \CC^*$, with the group $\CC^*$ acting on $\CC^{n+1}$ by $\lambda\cdot (y^0,\ldots,y^n)= (\lambda y^0,\ldots, \lambda y^n)$.  The extension of the $\CC^*$-action to the trivial line bundle (``linearization,'' see footonote \ref{footnote: GIT}) is:
    $\lambda\cdot (y^0,\ldots,y^n,\mu)= (\lambda y^0,\ldots, \lambda y^n,\lambda^{-1}\mu)$. 
    The unstable locus is $\GGamma=\{0\}\subset \CC^{n+1}$ -- the origin. We refer to \cite[Example 1.5]{Proudfoot} for details.
    \item $X=\CP^k$ is also a complex projective space (possibly of different dimension).
    \item $\mc{P}$ is  the line bundle $O(d)$ on $\CP^k$,\footnote{
     More precisely: $\mc{P}$ the principal $\CC^*$-bundle obtained by cutting out the zero-section from the line bundle $O(d)$.
    } for some $d\in \ZZ_{\geq 0}$ -- the ``degree.'' In this case we will be speaking of a quasimap of degree $d$ and denote the space of such quasimaps $\QMap_d(\CP^k,\CP^n)$.
\end{itemize}

Using homogeneous coordinates $(x^0:\cdots:x^k)$ on the source $\CP^k$, a general degree $d$ quasimap $\CP^k\qra \CP^n$ is a collection 
$$P^a(x^0,\ldots,x^k),\quad a=0,\ldots,n$$
of $n+1$ homogeneous polynomials of degree $d$ in $k+1$ variables, considered up to multiplying all polynomials simultaneously by $\lambda \in \CC^*$.

We will require that the polynomials $P^a$ are not all identically zero. (This corresponds to the assumption in Definition \ref{def: quasimaps} that the section $\ul{f}$ does not land in $\GGamma$ everywhere, i.e., is not the zero-section.)

As implied by the description above, degree $d$ quasimaps $\CP^k\qra \CP^n$ are parametrized by the collection of coefficients of the polynomials $P^a$ modulo multiplying all coefficients by $\lambda\in \CC^*$. Thus we obtained the following. 
\begin{proposition}
\begin{equation}\label{qmap(P^k,P^n) formula}
    \QMap_d(\CP^k,\CP^n)=\CP^{\n_{k,n,d}}
\end{equation}
-- the projective space of dimension
\begin{equation}\label{n_knd}
    \n_{k,n,d} = (n+1) \left(\begin{array}{c}
         d+k \\
          k
    \end{array} \right) -1.
\end{equation}
\end{proposition}
Here $-1$ corresponds to the quotient by $\CC^*$ and the binomial coefficient $\left(\begin{array}{c}  d+k \\ k\end{array} \right)$ is the number of coefficients in a single homogeneous polynomial $P^a$.

\begin{remark}
    In the spirit of Remark \ref{rem: qmaps as tuples of sections}, one can identify degree $d$ quasimaps $\CP^k\qra \CP^n$ with $(n+1)$-tuples of sections of $O(d)$ (not all identically zero) modulo $\CC^*$ acting diagonally: 
    $$\QMap_d(\CP^k,\CP^n)= \Gamma(\CP^k,O(d)\otimes \CC^{n+1})\sslash \CC^*. $$
\end{remark}

\begin{remark}\label{rem: qmaps from a product of CPs} One has the following generalization of the result (\ref{qmap(P^k,P^n) formula}).
    Quasimaps 
    $\CP^{k_1}\times \CP^{k_2}\qra \CP^n$ are characterized by a bi-degree $(d_1,d_2)$ (i.e., the corresponding line bundle $\mc{P}$ over $\CP^{k_1}\times \CP^{k_2}$ is $O(d_1)\boxtimes O(d_2)$). One can interpret a quasimap $\CP^{k_1}\times \CP^{k_2}\qra \CP^n$ as a degree $d_1$ quasimap from $\CP^{k_1}$ to $\QMap_{d_2}(\CP^{k_2},\CP^n)=\CP^{\n_{k_2,n,d_2}}$. Thus, one has
    \begin{equation}
    \begin{split}
    \QMap_{d_1,d_2}(\CP^{k_1}\times \CP^{k_2},\CP^n)&=
    \Gamma(\CP^{k_1}\times \CP^{k_2},(O(d_1)\boxtimes O(d_2))\otimes \CC^{n+1})\sslash \CC^*\\
    &=\Gamma(\CP^{k_1},O(d_1)\otimes\Gamma(\CP^{k_2},O(d_2)\otimes \CC^{n+1}))\sslash \CC^*\\
    &=\QMap_{d_1}(\CP^{k_1},\underbrace{\QMap_{d_2}(\CP^{k_2},\CP^n)}_{\CP^{\n_{k_2,n,d_2}}})\\
    &= \CP^{\n_{k_1,\n_{k_2,n,d_2},d_1}}\\
    &=\CP^{\n_{(k_1,k_2),n,(d_1,d_2)}}
    \end{split}
    \end{equation}
    -- the projective space of dimension
    $$ \n_{(k_1,k_2),n,(d_1,d_2)}=(n+1) \left(\begin{array}{c}
         d_1+k_1 \\
          k_1
    \end{array} \right)  \left(\begin{array}{c}
         d_2+k_2 \\
          k_2
    \end{array} \right) -1.
    $$
    
    Likewise, the space of quasimaps from any product of projective spaces to $\CP^n$ of given multi-degree is itself a projective space:
    \begin{equation}
         \QMap_{d_1,\ldots,d_m}(\CP^{k_1}\times\cdots \times \CP^{k_m},\CP^n) = \CP^{\n_{(k_1,\ldots,k_m),n,(d_1,\ldots,d_m)}},
    \end{equation}
    where
    \begin{equation}\label{n for X=product of CPs}
    \n_{(k_1,\ldots,k_m),n,(d_1,\ldots,d_m)}=
    (n+1) \prod_{i=1}^m \left(\begin{array}{c}
         d_i+k_i \\
          k_i
    \end{array} \right) -1.
    \end{equation}
\end{remark}

\subsection{
Stratification of 
\texorpdfstring{$\QMap$}{QMap} (by the type of \texorpdfstring{$\FR$}{FR})
}\label{sec:freckle_stratification}
The space of quasimaps has a natural stratification by the type of the proper quasimap locus  
(the class $\alpha$ of $\FR(f)$ in  homology of $X$\footnote{Or, more appropriately, the class of the cycle $\PQL(f)\subset \CP^k$ 
in the Chow ring, $\alpha\in A_*(X)$.}): 
\begin{equation}\label{QMap stratification}
\QMap_d(\CP^k,\CP^n)=\bigsqcup_\alpha \QMap^\alpha_d(\CP^k,\CP^n).
\end{equation}
In particular, holomorphic maps $\CP^k\ra \CP^n$ are embedded into $\QMap$ as the stratum with, $\alpha=\varnothing$. The union of all the other strata of $\QMap$ corresponds to ``proper'' quasimaps; we will denote it $\QMap^\pr$. We will denote $\QMap^m$ the stratum corresponding to quasimaps with $m=1,2,\ldots$ freckles (counted with multiplicities). Thus, stratification (\ref{QMap stratification}) has the form
\begin{multline} \label{QMap stratification 2}
    \QMap_d(\CP^k,\CP^n) = \\ =\underbrace{\QMap_d^\varnothing(\CP^k,\CP^n)}_{\Map_d(\CP^k,\CP^n)}\sqcup \underbrace{\bigsqcup_{m\geq 1} \QMap_d^m(\CP^k,\CP^n) \sqcup \bigsqcup_{\alpha,\dim\alpha\geq 1} \QMap^\alpha_d(\CP^k,\CP^n)}_{\QMap^\pr_d(\CP^k,\CP^n)}.
\end{multline}
The last term here corresponds to quasimaps $f$ for which $\FR(f)$ is not just a collection of isolated points in $X$, but a cycle of positive dimension. We will call such positive-dimensional $\FR$ loci ``scars.''  

\subsection{Examples}

\begin{example}
    \item Quasimaps of degree $d=0$ are constant maps:
    $$\QMap_0(\CP^k,\CP^n)=\Map_0(\CP^k,\CP^n)=\CP^n.$$
\end{example}
\begin{example} Quasimaps $\CP^1\qra \CP^1$ of degree $d$ are given by a pair of homogeneous polynomials
     \begin{equation}\label{qmaps P1 to P1 eq1}
     \begin{split}
    y^0 &=P^0(x^0,x^1)=\sum_{i=0}^d A^{0i}(x^0)^i (x^1)^{d-i}, \\
    y^1 &=P^1(x^0,x^1)=\sum_{i=0}^d A^{1i}(x^0)^i (x^1)^{d-i} ,
    \end{split}
    \end{equation}
    up to multiplying both of them by $\lambda\in \CC^*$. In terms of nonhomogeneous coordinates $z=x^1/x^0$ and $y=y^1/y^0$ on the source and target, one has\footnote{For simplicity here we assume $A^{0,0},A^{1,0}\neq 0$.}
    \begin{equation}\label{qmaps P1 to P1 eq2}
    y=\frac{P^1(1,z)}{P^0(1,z)}=C\,\frac{\prod_{i=1}^d(z-z_i^0)}{\prod_{j=1}^d(z-z_j^\infty)}.
    \end{equation}
    The constant $C$ and positions of zeros/poles $z_i^0,z_i^\infty$ are the parameters of a quasimap.
    This is an actual map of degree $d$ if $z^0_i\neq z^\infty_j$ for all $i,j$.   
    
    If $z^0_i=z^\infty_j=w$ (i.e. polynomials (\ref{qmaps P1 to P1 eq1}) have a common linear factor), then (\ref{qmaps P1 to P1 eq1}), (\ref{qmaps P1 to P1 eq2}) is a proper quasimap with a freckle at $w$. Evaluation of the quasimap at points $z\neq w$ then corresponds to a holomorphic map one degree lower.
    
    If $m$ pairs of $z^0_i$'s and $z^\infty_j$'s coincide, we have an $m$-freckle quasimap; on the complement of freckles its evaluation corresponds to a holomorphic map of degree $d-m$.
\end{example}

\begin{example}
\label{example: evaluation at a freckle}
    As an illustration of how the evaluation map (\ref{ev qmap}) can fail to exist at a freckle, 
    consider the following sub-example of (\ref{qmaps P1 to P1 eq1}). Consider a 1-parametric family of degree 1 quasimaps $f_a\colon\CP^1\qra \CP^1$ given by 
    $$y^0=x^1, y^1= a x^0,$$ 
    with $a$ a parameter. In nonhomogeneous coordinates, the family is 
    $$y=\frac{a}{z}.$$ 
    The family consists of maps for $a\neq 0$ and a proper quasimap for $a=0$, with freckle at $z=0$. Let $\omega=\frac{i}{2\pi}\frac{dy d\bar{y}}{(1+|y|^2)^2}$ be the Fubini-Study form on the target. Then the limit of the pullback of $\omega$ by the evaluation map, while simultaneously taking  $a$ to zero and  the evaluation point to the freckle
    \begin{equation}
        \lim_{a,z\ra 0} (\mr{ev}|_{f_a})^*(\omega)
        = \lim_{a,z\ra 0}\frac{i}{2\pi} \frac{(zda -adz)(\bar{z} d\bar{a}-\bar{a}d\bar{z})}{(|z|^2+|a|^2)^2}
    \end{equation}
    fails to exist. 
\end{example}

Let us introduce the notation
\begin{equation}\label{Conf}
\Conf_m(X)\defeq X^{\times m}/\mr{Sym}_m
\end{equation}
for the configuration space of $m$ unordered points on $X$ (where the points are allowed to collide).

\begin{example}
A quasimap $\CP^1\qra \CP^n$ of degree $d$ is determined by the (nonzero) matrix of coefficients $(A^{ai})_{0\leq a\leq n, 0\leq i \leq d}$ of polynomials
\begin{equation}
    y^a=P^a(x^0,x^1)=\sum_{i=0}^d A^{ai} (x^0)^i (x^1)^{d-i},\quad a=0,\ldots,n,
\end{equation}
considered modulo scaling $A^{ai}\ra \lambda A^{ai}$ for $\lambda\in \CC^*$.

A quasimap has a freckle if all $P^a$'s have a linear polynomial as a common divisor (the root of this polynomial is the freckle position). Likewise, if all $P^a$'s have a common divisor $Q$ of degree $m$, the quasimap has $m$ freckles located at the roots of $Q$.
This discussion implies that the stratification (\ref{QMap stratification 2})
has the form
\begin{equation}
    \QMap_d(\CP^1,\CP^n)=\bigsqcup_{m=0}^d \QMap_d^m (\CP^1,\CP^n).
\end{equation}
The closure of the $m$-freckle stratum can be identified with
\begin{equation}\label{m-freckle stratum P^1 to P^n}
    \overline{\QMap_d^m(\CP^1,\CP^n)} = \QMap_{d-m}(\CP^1,\CP^n)\times 
    \CP^m.
\end{equation}
Here the first factor corresponds to the quasimap defined by the polynomials $P^a/Q$ and the 
second factor parametrizes the possible polynomials $Q$ modulo $\CC^*$.\footnote{One can also write the last factor in (\ref{m-freckle stratum P^1 to P^n}) as $
\Conf_m(\CP^1)
$, parametrizing the polynomial $Q$ by its $m$ roots (as an unordered set) -- the positions of freckles. The denominator $\mr{Sym}_m$ in (\ref{Conf}) is the Galois group.}
\end{example}

\begin{proposition}\mbox{}
\begin{enumerate}[(i)]
\item \label{Prop qmap strata (i)} 
If $n\geq k$, then a quasimap $\CP^k\qra \CP^n$ in general position is a map and, for $d\geq 1$, the  $1$-freckle stratum $\QMap^1_d(\CP^k,\CP^n)$ has complex codimension $n+1-k$ in $\QMap_d(\CP^k,\CP^n)$.

\noindent Moreover, for $m\geq 1$ and $d$ large enough, the $m$-freckle stratum $\QMap^m_d(\CP^k,\CP^n)$ has codimension $m(n+1-k)$ in $\QMap_d(\CP^k,\CP^n)$.
\item \label{Prop qmap strata (ii)} If $n<k$, then a quasimap $\CP^k\qra \CP^n$ of degree $d\geq 1$ is always proper. For a quasimap $f$ in general position, 
$\FR(f)$
is a cycle of complex dimension $k-n-1$ in $\CP^k$.
\end{enumerate}
\end{proposition}

\begin{proof}
(\ref{Prop qmap strata (i)}): The condition that 
a quasimap 
has a freckle at 
a given point $x\in \CP^k$ 
is a system of $n+1$ homogeneous linear equations $P^a(x)=0,\; a=0,\ldots,n$ on the coefficients of the polynomials $P^a$. Its solution locus is an intersection of $n+1$ transversal hyperplanes in $\QMap$, i.e., a projective space $\CP_x\subset \QMap$ of codimension $n+1$. 
Taking a union over possible positions of the freckle, we obtain a cycle $\alpha=\cup_{c\in \CP^k}\CP_x$ of codimension $n+1-k$ in $\QMap$. 
For $n\geq k$ this codimension is positive, so a generic quasimap is a map. 
A subtlety is that a quasimap might have more than one freckle, so the stratum of quasimaps with exactly one freckle is 
$\alpha$ with some higher codimension strata removed.

The $m$-freckle case is similar.

(\ref{Prop qmap strata (ii)}): Given a quasimap $\CP^k\qra \CP^n$, its  $\FR$ locus is the intersection of the zero-loci of polynomials $P^a$, i.e., an intersection of $n+1$ hypersurfaces in $\CP^k$ of degree $d$. The intersection is, generically (when the hypersurfaces intersect transversally), a cycle of codimension $n+1$.\footnote{Another reason why one cannot have a map $f\colon\CP^k\ra\CP^n$ of degree $d>0$ if $n<k$: consider the generator of the second cohomology of the target $[\omega_Y]\in H^2(\CP^n)$ and $[\omega_X]$  the generator of $H^2(\CP^k)$.
One has the relation $[\omega_Y]^{n+1}=0$. Since $f$ has degree $d$, one also has $f^*[\omega_Y]=d[\omega_X]$, which implies $d^{n+1} [\omega_X]^{n+1}=0$, which is false in the cohomology of the source. 

A related remark: a quasimap $f$ is a section of the bundle $\underbrace{O(d)\oplus\cdots \oplus O(d)}_{n+1}$ over $\CP^k$. Its Euler class is $(d [\omega_X])^{n+1}\in H^{2(n+1)}(\CP^k)$ and it is the Poincar\'e dual of the homology class of 
$\FR(f)$ for a generic quasimap $f$.
}
\end{proof}

Note that the proof above also shows that the closure 
of the $m$-freckle stratum is (for $d$ large enough) 
birationally equivalent\footnote{ I.e., isomorphic outside higher-codimension strata.} to
a bundle over $\Conf_m(\CP^k)$ with fiber $\CP^{\n_{k,n,d}-m(n+1)}$.\footnote{
For $d$ too small, the $m$-freckle stratum might vanish. E.g., for $m\geq 2$, $m$-freckle strata in $\QMap(\CP^2,\CP^2)$ vanish for $d=1$; for $m\geq 5$, $m$-freckle strata vanish for $d=1,2$, cf. Example \ref{ex: qmap stratification}.
}

\begin{example} For $k\geq 1$, all quasimaps from $\CP^k$ to $\CP^0$ of degree $d\geq 1$ are proper. Note that $\CP^0=\CC\sslash \CC^*$ is a point with a particular presentation as a GIT quotient. A map to a point is necessarily constant, while the space of quasimaps (\ref{qmap(P^k,P^n) formula}) is nontrivial.
\end{example}

\begin{example}\label{ex: qmap stratification}
    A degree $d$ quasimap $\CP^2\qra \CP^n$ has a one-dimensional proper quasimap 
    locus (``scar'') if all polynomials $P^a$ are divisible by a nonzero polynomial $Q$ of some degree $1\leq \Delta\leq d$. Then the scar is the zero-locus of $Q$ -- a degree $\Delta$ curve in $\CP^2$. It is easy to see that the closure of the corresponding stratum in $\QMap$ is
    \begin{equation}
        \overline{\QMap_d^{\Delta\mr{-scar}}(\CP^2,\CP^n)} = \CP^{\n_{2,n,d-\Delta}}\times \CP^{\n_{2,0,\Delta}}.
    \end{equation}
    Here the first factor parametrizes the polynomials $P^a/Q$ modulo $\CC^*$ and the second factor parametrizes the polynomial $Q$ itself modulo $\CC^*$.

\begin{itemize}
    \item 
   For instance, degree $1$ quasimaps $\CP^2\qra \CP^2$ form $\CP^8$, with  the following strata:


\begin{center}
    \begin{tabular}{c|c|c}
          stratum $\sigma$ & $\mr{codim}(\sigma)$ & shape of $\overline{\sigma}$ \\ \hline
        maps & 0 & $\CP^{8}$ \\
         1-freckle & 1 & $\CP^5$-bundle over $\CP^2$ \\
         $\Delta=1$ scar & 4 & $\CP^2\times \CP^2$
    \end{tabular}
\end{center}

    Here we note that if for a degree 1 quasimap, 
    $\FR$
    contains two points, then it also contains a line through those points. This shows that there are no $m>1$ freckle cases in the table above (they are subsumed by the scar case). 

The ``shape'' of each stratum is given up to  birational equivalence.

\item For degree 1 quasimaps $\CP^2\qra \CP^1$, 
the similar table is:

\begin{center}
    \begin{tabular}{c|c|c}
          stratum $\sigma$ & $\mr{codim}(\sigma)$ & shape of $\overline{\sigma}$ \\ \hline
        1-freckle & 0 & $\CP^{5}$ \\
         $\Delta=1$ scar & 2 & $\CP^1\times \CP^2$
    \end{tabular}
\end{center}

\item Degree $1$ quasimaps $\CP^2\qra \CP^0$ all belong to $\Delta=1$ scar stratum.

\item Another example: degree $2$ quasimaps $\CP^2\qra \CP^2$. The space $\QMap_2(\CP^2,\CP^2)=\CP^{17}$ has the following strata:


\begin{center}
    \begin{tabular}{c|c|c}
        stratum $\sigma$ & $\mr{codim}(\sigma)$ & shape of $\overline{\sigma}$ \\ \hline
        maps & 0 & $\CP^{17}$ \\
        1-freckle & 1 & $\CP^{14}$-bundle over $\Conf_1(\CP^2)$ \\
        2-freckle & 2 & $\CP^{11}$-bundle over $\Conf_2(\CP^2)$ \\
        3-freckle & 3 & $\CP^{8}$-bundle over $\Conf_3(\CP^2)$ \\
        4-freckle & 4 & $\CP^{5}$-bundle over $\Conf_4(\CP^2)$ \\
        $\Delta=1$ scar & 7 & $\CP^8\times \CP^2$ \\
         $\Delta=2$ scar & 10 & $\CP^2\times \CP^5$
    \end{tabular}
\end{center}
In the last two lines, the second factor describes the position of the scar.

Note that if one has 
more than $5$ 
points in $\FR$ (common zeros of the polynomials $P^a$) in general position, then $\FR$ is 
a scar, since a conic in $\CP^2$ is uniquely determined by $5$ points and so the zero-loci of polynomials $P^a$ coincide. 
\end{itemize}
\end{example}

\section{Counting holomorphic maps (``\texorpdfstring{$\KM$}{KM} numbers'')}
\subsection{Formulation of the problem} \label{ss holom maps: formulation of the problem}

\begin{enumprob}\label{enum prob: counting hol maps}
Let $X$ and $Y$ be two compact complex manifolds (source and target), of dimensions $k$ and $n$ respectively. 
Fix a collection of cycles\footnote{By default, a ``cycle'' in this paper stands for a holomorphic (or, equivalently, algebraic) cycle.
} $c_1^X,\ldots,c_l^X$ in $X$ and a collection of cycles $c_1^Y,\ldots,c_l^Y$ in $Y$ ($l$ is the same in both). Also, fix an element $\delta$ in the set $[X,Y]$ of homotopy classes of maps $X\ra Y$.

We are interested in the number\footnote{The notation stands for ``Kontsevich-Manin number'' and refers to \cite{KM}.}
\begin{equation}\label{N^enum KM}
\KM(X,Y;\{c_i^X,c_i^Y\}_{i=1,\ldots,l}|\delta)
\end{equation}
of holomorphic maps $f\colon X\ra Y$ in homotopy class $\delta$ such that the image of $c_i^X$ in $Y$ intersects $c_i^Y$ for each $i=1,\ldots,l$.
\end{enumprob}

By convention, if holomorphic maps subject to the condition above have continuous moduli, we set $\KM=0$. 

We are thinking of the problem above in terms of an $l$-tuple $x_1,\ldots,x_l$ of points in $X$ which we want to be mapped to the target cycles. Some of the points $x_i$ may be ``fixed'' (i.e. the respective source cycle $c_i^X$ is a point), some may be ``moving freely'' (the respective source cycle is $c_i^X$ is the entire $X$); in general the points $x_i$ are constrained to their respective source cycles $c_i^X$.


While it is possible to study this enumerative problem in full generality, in this paper we will restrict ourselves to the case of maps between projective spaces $\CP^k\ra\CP^n$, since most of the 
phenomena show up already in this case. 

\begin{remark}
    The usual setup of genus zero Gromov-Witten invariants (Gromov-Witten classes integrated over the full moduli space of curves with $l$ marked points) is a special case of the Enumerative Problem \ref{enum prob: counting hol maps} 
    where $X=\CP^1$, 
    with the following choice of source cycles: 
    three cycles among $c_i^X$ are fixed points (e.g. $c_1^X=\{0\}, c_2^X=\{1\}, c_3^X=\{\infty\}$ in $X=\CP^1$ -- needed in order to fix the group of automorphisms of $\CP^1$) and the rest are copies of the fundamental cycle $X$.  (I.e., we have $3$ fixed points and $l-3$ moving points.)
    
    Instead of a homotopy class of a map $\delta$, one usually specifies an element $\beta\in H_2(Y,\ZZ)$, requiring that the holomorphic maps are such that the homology class of the image of $X$ in $Y$ is $\beta$.
\end{remark}

\begin{remark}
The numbers that we are considering look similar to Donaldson invariants \cite{Donaldson,Witten94}. 
\end{remark}

\subsection{Important remark: meaningfulness of the problem. Syzygies of holomorphicity equations.}  
It is widely believed that the problem of holomorphic maps from a higher-dimensional source has virtual dimension equal to $-\infty$ and that is why it should not be studied on the same footing as holomorphic maps from 1-dimensional source. However, this argument is not quite correct due to syzygies, as we will outline below.


Let $(X,\JX)$ and $(Y,\JY)$ be two almost complex manifolds. 
A map $f\colon X \to Y$ is called \emph{pseudo-holomorphic} if its differential intertwines the two almost complex structures, i.e.\ if 
\begin{equation}
    \JY \circ f_* = f_*\circ \JX.
    \label{eq:condition_pseudo_hol}
\end{equation}
Consequently, a pseudo-holomorphic map $f$ intertwines the Nijenhuis tensors of $X$ and $Y$:
\begin{equation}\label{Nijenhuis}
  \del_{\mu}f^a\ (N_{\JX})^\mu_{\bar \rho \bar \sigma} =  (N_{\JY}\circ f)^a_{\bar b \bar c}\  \overline{\del_{\rho} f^b}\ \overline{\del_{\sigma} f^{\vphantom{b}c}}.
\end{equation}
A sufficient condition to satisfy the constraints (\ref{Nijenhuis}) is the integrability of the source and target complex structures $\JX$ and $\JY$, which we will consider henceforth.

The space of holomorphic maps $X \to Y$ between two complex manifolds naively has negative infinite virtual dimension if $\dim_\CC X \geq 2$. 
Indeed, if $\dim_{\CC} X = k$ and $\dim_{\CC} Y = n$, then if $z^i$ are local complex coordinates of $X$, the holomorphicity equation 
\begin{equation}
    \delbar f^a(z)  = 0 
    \label{eq:holomorphicity}
\end{equation}
yield $nk$ pointwise conditions for the $n$ variables $f^a$.
Therefore
\begin{equation*}
    \dim_{\rm vir} \Map(X,Y) = (\# {\rm variables} - \# {\rm equations})\cdot \#{\rm points} = n(1 - k)\cdot\infty
\end{equation*}
which equals $-\infty$ when $k \geq 2$.
However, in higher dimensions ($k\geq 2$) there exist syzygies (linear relations among the conditions \eqref{eq:holomorphicity}) which render the dimension of the space of holomorphic maps finite.

Since the holomorphicity equation \eqref{eq:holomorphicity} give point-wise conditions, syzygies can be expressed in terms of integrals.
For $\dim_{\CC}X = k \geq 2$, we call $\sigma_a \in \Omega^{(k,k-1)}$ a \emph{syzygy} if
\begin{equation}
    \int_X \sigma_a \wedge \delbar f^a = 0,
\end{equation}
independently whether or not $\delbar f^a = 0$.
Syzygies in the above sense are thus given by $\delbar$-closed forms.
If $k$ is large enough, syzygies are determined only up to $\delbar$-exact forms.
This redundancy usually indicates the presence of syzygies among syzygies.

\begin{remark}
    In principle it is possible to construct the full tower of syzygies by constructing the Koszul-Tate resolution of an appropriate bundle:
    Let $\mathcal E$ be the vector bundle over the space of smooth maps $\Map^{\rm sm}(X,Y)$, whose fiber above a map $f$ is given by 
    \[
    \mathcal E_{f} = \Omega^{0,1}(X,f^*T^{1,0} Y).
    \]
    There exists a natural section $s \in \Gamma(\mathcal E)$, which takes a function $f$ to its Dolbeault differential $\delbar f$. 
    Equation \eqref{eq:holomorphicity} is thus expressed as the zero set of the section $s$. 
    Syzygies (and syzygies for syzygies and so forth) can now be described in terms of Tate generators of the Koszul-Tate resolution of  the sheaf of function on the zero-locus of $s$. 
\end{remark}

\subsection{The case of maps \texorpdfstring{$\CP^k\ra \CP^n$}{CPk -> CPn}}\label{ss: counting hol maps CP^k to CP^n}

Consider the Enumerative Problem \ref{enum prob: counting hol maps} in the case $X=\CP^k$, $Y=\CP^n$, with $n\geq k$. Instead of specifying the homotopy class $\delta$ of maps between the two projective spaces we will be specifying their degree $d$.

We will also restrict ourselves to cycles $c_i^X, c_i^Y$ given by intersection of several hyperplanes in general position.

The space\footnote{More appropriately, $\mc{M}$ is a cycle in $\Map_d(\CP^k,\CP^n)$ with $\ZZ$-coefficients. From this viewpoint, $p$ in (\ref{maps as intersection}) should be replaced with the pushforward $p_*$ of cycles.}
\begin{equation}\label{M_d}
\mc{M}(\CP^k,\CP^n;\{c_i^X,c_i^Y\}|d)
\end{equation}
of degree $d$ holomorphic maps $\CP^k\ra \CP^n$, such that the images of $c_i^X$ intersect $c_i^Y$ for $i=1,\ldots,l$, can be 
represented as follows:
\begin{equation}\label{maps as intersection}
\mc{M}=p\left(\bigcap_{i=1}^l \mr{ev}^{-1}_i c_i^Y\right) \quad \subset \Map_d(\CP^k,\CP^n).  
\end{equation}
Here $p\colon\mr{Map}_d(\CP^k,\CP^n)\times \prod_{i=1}^l c_i^X \ra \mr{Map}_d(\CP^k,\CP^n)$ is the projection onto the first factor and
\begin{equation}\label{ev s.3.1}
\mr{ev}_i\colon \mr{Map}_d(\CP^k,\CP^n)\times \prod_{i=1}^l c_i^X \ra \CP^n
\end{equation} 
is the evaluation of a map on a point of the $i$-th cycle.
It follows that the expected (or ``virtual'') complex dimension of the space (\ref{M_d}) is
\begin{equation}\label{dim_vir(M_d)}
  \dim_\mr{vir} \mc{M}(\CP^k,\CP^n;\{c_i^X,c_i^Y\}|d) =  \n_{k,n,d}+\sum_{i=1}^l\dim c_i^X - \sum_{i=1}^l \mr{codim}\, c_i^Y.
\end{equation}

We assume that if the virtual dimension vanishes, $\mc{M}$ is actually a finite set (for source/target cycles in general position). 
It is true in all examples that we encounter in this paper. 
In this case, the $\KM$  number (\ref{N^enum KM}) counts the points of $\mc{M}$. 

Note that the evaluation map (\ref{ev s.3.1}) is invariant under the automorphism group of the source, $\mr{Aut}(\CP^k)=\PSL(k+1,\CC)$. 
Hence, in order to have $\dim \M=0$, we need to assume that the configuration of source cycles $\{c_i^X\}_{i=1,\ldots,l}$ has discrete 
stabilizer in $\mr{Aut}(\CP^k)$.
In particular, if the source cycles are $l_1$ points in general position and $l_2$ copies of $\CP^k$ (i.e. we have $l_1$ fixed points and $l_2$ freely moving points), then we need to assume $l_1\geq k+2$.\footnote{ \label{footnote: transitivity of PSL(k+1,C) acting on P^k}
Note that $\PSL(k+1,\CC)$ acts $3$-transitively on $\CP^k$ for $k=1$ and ``generically $(k+2)$-transitively'' if $k\geq2$. More precisely: for two configurations of $k+2$ points \emph{in general position} on $\CP^k$ there is a unique  element of $\PSL(k+1,\CC)$ mapping the first configuration to the second.}

\subsection{Easy enumerative problem}\label{ss: easy enum prob}

For maps between projective spaces, consider the case when all source cycles $c_i^X$ are points $x_1,\ldots,x_l\in \CP^k$ in general position 
(i.e., in the terminology of Section \ref{ss holom maps: formulation of the problem}, we have only ``fixed points'' on the source). 
Assume that for each $i=1,\ldots,l$, the target cycle $c_i^Y$ is the intersection of hyperplanes 
\begin{equation}\label{target cycles as intersections of hyperplanes}
c_i^Y=\{y\in \CC^{n+1}\backslash\{0\} \;|\; H_{i,\alpha_i}(y)=0\}/\CC^*
\end{equation} 
specified by nonzero covectors $H_{i,\alpha_i}\in (\CC^{n+1})^\vee$, with $\alpha_i=1,\ldots,\codim c_i^Y$.
Then solutions of the Enumerative Problem \ref{enum prob: counting hol maps} are the nonzero solutions (modulo $\CC^*$) of the system of homogeneous linear equations 
$$\sum_{a=0}^n H_{i,\alpha_i,a} P^a(x^0_i,\ldots,x^k_i)=0,\quad i=1,\ldots,l,\;\; \alpha_i=1,\ldots,\codim c_i^Y$$
on coefficients of the polynomials $P^a$ determining the map. Thus we have exactly one solution modulo $\CC^*$ if the dimension (\ref{dim_vir(M_d)}) vanishes, and no solutions or a continuous family of solutions otherwise. In summary, we have the following
\begin{proposition}
    Let the source cycles $\{c_i^X\}$ be points in general position in $\CP^k$ and let the target cycles be intersections of hyperplanes in $\CP^n$. Then the answer to the Enumerative Problem \ref{enum prob: counting hol maps} is:
    \begin{equation}\label{KM for easy problem}
        \KM(\CP^k,\CP^n;\{c_i^X,c_i^Y\}_{i=1,\ldots,l}|d)=
        \left\{
        \begin{array}{ll}
             1& \mr{if\;} \n_{k,n,d}=\sum_{i=1}^l \codim c_i^Y, \\
             0& \mr{otherwise}
        \end{array}
        \right.
    \end{equation}
\end{proposition}

\subsection{Toward higher quantum cohomology and a mysterious theta-function}

In the case $X=\CP^1$, the numbers (\ref{KM for easy problem}) summed over $d$ with weight $q^d$ organize into the quantum cohomology of the target $\CP^n$, i.e.\ into the commutative associative product $*$ on $H^\bt(\CP^n)[[q]]$ such that 
\begin{equation}\label{qcoh eq1}
\langle \alpha_1*\alpha_2*\cdots*\alpha_{l-1}, \alpha_l \rangle =
\sum_{d\geq 0} q^d \KM(\CP^1,\CP^n;\{\mr{point}_i,\alpha_i^\vee\}_{i=1,\ldots,l}|d) ,
\end{equation}
 where $\alpha_i\in H^\bt(\CP^n)$ are target cohomology classes and $c_i^Y=\alpha_i^\vee$ are cycles in the target representing the Poincar\'e duals of $\alpha_i$;\footnote{We are implicitly extending the numbers (\ref{KM for easy problem}) by linearity to general target cycles.} $\langle,\rangle$ is the standard Poincar\'e pairing on cohomology. In particular, the ``easy $\KM$ numbers'' (\ref{KM for easy problem}) for $l=3$ give the structure constants of the quantum product $*$ and the case $l\geq 4$ is recovered as the iterated quantum product.\footnote{In this discussion one can replace the target $\CP^n$ with any compact complex  manifold $Y$; the $\KM$ numbers with source cycles being points will be more complicated, but will still arrange into a commutative associative ring structure on $H^\bt(Y)[[q_1,\ldots,q_m]]$, where $m=\mr{rk}\,H^2(Y)$.} 
 
 In particular, the quantum cohomology ring of $\CP^n$ can be identified with 
 \begin{equation}\label{qcoh of Pn as a quotient}
 \CC[[q]] [\omega]/(\omega^{n+1}-q)
 \end{equation} -- the quotient of the ring of polynomials in the variable $\omega$ (the class of the Fubini-Study 2-form) by the ideal generated by $\omega^{n+1}-q$.
 For instance, in this ring one has $\omega^n*\omega= q$ (while the classical cup product is 
$\omega^n\!\smile\! \omega=0$).  
 
One can reformulate this structure in terms of a weak\footnote{By ``weak'' we mean that the inner product induced by the counit is degenerate.} Frobenius algebra $\CC[[q]][x]$ 
-- the algebra of polynomials in $x$ (a formal variable identified with the class $\omega$) 
equipped with the standard product  and the trace (counit) $\eta\colon \CC[[q]][x]\ra \CC[[q]]$ defined by
$$ \eta(p(x))\defeq\frac{1}{2\pi i}\oint_\gamma \frac{p(x)dx}{x^{n+1}-q}  = \sum_{d\geq 0} q^d \frac{1}{2\pi i}\oint_\gamma \frac{p(x)dx}{x^{(n+1)(d+1)}},$$
with $p(x)$ a polynomial and $\gamma$ a closed contour going around the origin.
This counit induces a (degenerate) pairing $(p_1,p_2) \defeq \eta(p_1 p_2)$.  
One can then identify the quantum cohomology ring (\ref{qcoh of Pn as a quotient}) as the quotient of $\CC[[q]][x]$ by the kernel of the pairing $(,)$. 
The cyclic quantum product 
 (\ref{qcoh eq1}) can then be written as 
$$\eta(p_1 p_2\cdots p_l)=\frac{1}{2\pi}\oint_{\gamma} \frac{p_1(x)\cdots p_l(x)dx}{x^{n+1}-q},$$
with $p_i(x)$ the polynomials corresponding to the cohomology classes $\alpha_i$.

Next, consider the case $X=\CP^k$ with $k\geq 2$. 
Consider the following multilinear operation (``cyclic higher quantum product'') on target cohomology
\begin{equation}
   \begin{array}{cccc}
         \QP\colon& \mr{Sym}^l H^\bt(\CP^n) [[q]]&\ra & \CC[[q]]  \\ &
         \alpha_1\otimes \cdots\otimes \alpha_l& \mapsto & \sum_{d\geq 0} q^d  \KM(\CP^k,\CP^n;\{\mr{point}_i,\alpha_i^\vee\}_{i=1,\ldots,l}|d) 
   \end{array}
\end{equation}
Here 
$\alpha_i\in H^\bt(\CP^n)$ are target cohomology classes and $c_i^Y=\alpha_i^\vee$ are representatives of the Poincar\'e dual homology classes.

 We define the \emph{higher quantum cohomology} of $Y=\CP^n$ parametrized by $X=\CP^k$ as the Frobenius algebra 
 $\mc{A}_{\CP^k,\CP^n}\defeq \CC[[q]][x]$ with the standard product and the counit
 \begin{equation}\label{counit eta}
 \eta(p(x))\defeq\sum_{d\geq 0} q^d\frac{1}{2\pi i}\oint_\gamma \frac{p(x)dx}{x^{\n_{k,n,d}+1}}.
 \end{equation}
 
 Note that $\eta$ acts on monomials as
 $$
 \eta(x^j) = \left\{ 
 \begin{array}{cc}
     q^d & \mr{if}\; j=\n_{k,n,d},   \\
     0 & \mr{otherwise}
 \end{array}
 \right.
 $$
 \begin{example}
 For instance, for $k=n=2$, $\eta$ maps monomials as
 $$
 x^2\mapsto 1,\; x^8\mapsto q,\; x^{17}\mapsto q^2, \;  x^{29} \mapsto q^3,\; x^{44}
\mapsto q^4, \ldots
 $$
 while monomials of other degrees $j\not\in \{2,8,17,29,44,\ldots\}$ are mapped to zero.
 \end{example}

 \begin{remark}
     For $k=2$ (i.e. for maps $\CP^2\ra \CP^n$), one can write the counit (\ref{counit eta}) in terms of the Jacobi theta function $\theta_{10}$ (in Mumford's notation) as follows:
     $$\eta(p(x))=\frac{1}{2\pi i}\oint_\gamma p(x) F(q,x),$$
     where 
     \begin{align*}
     F(q,x)&=\sum_{d\geq 0} \frac{q^d}{x^{(n+1)\frac{(d+2)(d+1)}{2}}} dx\\
     &=  \mr{reg}_{q=0}\left(q^{-1}\sum_{d'\in \ZZ}\frac{q^{d'}}{x^{(n+1)\frac{d'(d'+1)}{2}}}dx\right)\\
     &= \mr{reg}_{q=0} \Big(x^{\frac{n+1}{8}}q^{-\frac32}\underbrace{\sum_{d'\in\ZZ}q^{d'+\frac12}x^{-\frac{n+1}{2}(d'+\frac12)^2}}_{\theta_{10}(z,\tau)}dx\Big).
     \end{align*}
     Here $\mr{reg}_{q=0}$ stands for the operation of subtracting the negative part of the Laurent expansion in $q$; the variables $z,\tau$ are related to $q,x$ by
     $$x^{-\frac{n+1}{2}}=e^{\pi i\tau},\;\; q=e^{2\pi i z}.$$
     Thus, one has
     \begin{equation}\label{counit via theta}
     \eta(p(x))
     =-\frac{1}{n+1}\mathsf{P}\oint_{\tilde\gamma} d\tau\, p(e^{-\frac{2\pi i \tau}{n+1}}) e^{-2\pi i\tau (\frac{1}{n+1}+\frac18)}e^{-3\pi iz} \theta_{10}(z,\tau) ,
     \end{equation}
     where $\mathsf{P}$ is the projection to nonnegative Fourier modes in the variable $z$; $\tilde\gamma$ is the image of the contour $\gamma$ in the plane of the modular parameter $\tau$.
 \end{remark}

 The following is an obvious consequence of (\ref{KM for easy problem}).
 \begin{corollary}
One has
\begin{equation}
    \QP(\alpha_1,\ldots,\alpha_l)= \eta(p_1(x)\cdots p_l(x)),
\end{equation}
where $\alpha_i$ are cohomology classes of $\CP^n$ and $p_i(x)$ are the corresponding polynomials.
 \end{corollary}
 
 \begin{lemma}\label{lem:non_zero_kernel}
 For $k\geq 2$, the pairing on $\mc{A}_{\CP^k,\CP^n}$ induced by the counit, $(p_1,p_2)=\eta(p_1p_2)$, has zero kernel.
 \end{lemma}
\begin{proof}
    Assume that a polynomial 
    $$p(x)=\sum_{i=0}^\delta p_i(x) q^i$$ 
    (where $p_i$ are polynomials in $x$ independent of $q$) is in the kernel of $(,)$, i.e., that $\eta(p(x) x^j)=0$ for all $j\geq 0$.  Assume that $p_\delta$ is nonzero. Choose $d$ large enough such that one has
\begin{itemize}
    \item $\n_{k,n,d}\geq \deg p_\delta$,
    \item $\n_{k,n,d}+\deg p_i < \n_{k,n,d+1}$ for $i<\delta$.
\end{itemize}
The fact that it is possible to satisfy the second condition relies on the assumption $k\geq 2$ -- then the gaps between consecutive numbers in the sequence $\{\n_{k,n,d}\}_{d=0,1,2,\ldots}$ are growing.

The expression
\begin{equation}\label{proof of zero kernel eq}
\eta(p(x) x^{\n_{k,n,d}-\deg p_\delta})
\end{equation}
contains the term $q^{\delta+d}\cdot (\mr{top\; coefficient\; of}\;p_\delta )$ coming from $p_\delta$ which cannot be canceled by anything from $p_{<\delta}(x)$. Thus, (\ref{proof of zero kernel eq}) is nonzero. Hence, a nonzero $p(x)$ cannot be in the kernel of $(,)$.
\end{proof}
 
 In particular, Lemma \ref{lem:non_zero_kernel} implies that unlike in the case of the usual quantum product ($k=1$), $\mc{A}_{\CP^k,\CP^n}$ does not have a finite-dimensional quotient isomorphic to the  cohomology of the target  with a $q$-deformed product.
 
 Considering the higher quantum product as a $2\ra 1$ operation $*$, if we identify $x^j$ with cohomology classes $\omega^j$ of $\CP^n$ for $j=0,\ldots,n$, we should say that the quantum product of several cohomology classes $\xi=\omega^{j_1}*\cdots*\omega^{j_l}$ is $\omega^{j_1+\cdots+j_l}$ if $j_1+\cdots+j_l\leq n$, otherwise $\xi$ does not correspond to an element of $H^\bt(\CP^n)$ (but corresponds to an element of $\mc{A}_{\CP^k,\CP^n}$). In other words, $H^\bt(\CP^n)$ is identified with a subspace of $\mc{A}_{\CP^k,\CP^n}$ which is not closed under the quantum product.

In the language of QFT (a higher-dimensional A-model localizing to holomorphic maps $X\ra Y$), we expect the higher quantum cohomology ring to describe the OPE algebra of a class of observables: elements of $H^\bt(Y)$ correspond to the usual evaluation observables, while other elements of $\mc{A}_{\CP^k,\CP^n}$  correspond to a new type of observables.\footnote{Perhaps, a higher-dimensional counterpart of the ``tangency observables'' known in the A-model.} In particular, OPEs of evaluation observables can contain the ``new'' observables.
 
\begin{remark}
    The discussion above generalizes straightforwardly to maps $\CP^{k_1}\times\cdots \times \CP^{k_m}\ra \CP^n$ (cf. Remark \ref{rem: qmaps from a product of CPs}). In this case, one introduces formal parameters $q_1,\ldots, q_m$ so that a map of multi-degree $(d_1,\ldots d_m)$ is counted with weight $q_1^{d_1}\cdots q_m^{d_m}$.   The higher quantum cohomology in this case is the Frobenius algebra $\mc{A}_{\CP^{k_1}\times\cdots\times \CP^{k_m},\CP^n}=\CC[[q_1,\ldots,q_m]][x]$ with the standard product and the counit
    $$
    \eta(p(x))=\sum_{d_1,\ldots,d_m\geq 0}q_1^{d_1}\cdots q_m^{d_m} \frac{1}{2\pi i}\oint_\gamma \frac{p(x) dx}{x^{\n_{(k_1,\ldots,k_m),n,(d_1,\ldots,d_m)}+1}}
    $$
    with the exponents $\n_{(k_1,\ldots,k_m),n,(d_1,\ldots,d_m)}$ as in (\ref{n for X=product of CPs}).

    For instance, for maps $\CP^1\times \CP^1\ra \CP^n$ one has $\mc{A}_{\CP^1\times\CP^1,\CP^n}=\CC[[q_1,q_2]][x]$ with
    $$ \eta(p(x)) = \sum_{d_1,d_2 \geq 0} q_1^{d_1} q_2^{d_2} \frac{1}{2\pi i} \oint_\gamma \frac{p(x)dx}{x^{(n+1)(d_1+1)(d_2+1)}}. $$
\end{remark}

\section{Counting quasimaps (``\texorpdfstring{$\GLSM$}{QM} numbers'')}
\subsection{Formulation of the problem}\label{sec:setup} 
We now describe the \emph{quasimap counting problem}. 
Let $X$ be a  compact complex manifold and \mbox{$Y = \mathbb{C}^N \sslash G$}. We denote $\pi\colon \mathbb{C}^N \setminus \GGamma \to Y$ the quotient map. We want to count quasimaps $f = (\mathcal{P},\underline{f})$ subject to certain conditions. Namely, fix a natural number $l$ and closed submanifolds $c_1^X,\ldots,c_l^X \subset X$ and $c_1^Y,\ldots,c_l^Y \subset Y$. 
We have that $\overline{\pi^{-1}(c_i^Y)}$ is a $G$-space and we can consider the inclusion of fiber bundles 
\begin{equation}
    \iota\colon \mc{P} \times_G  \overline{\pi^{-1}(c_i^Y)} \to \mc{P} \times_G \mathbb{C}^N.
\end{equation}
We denote  $\widetilde{c_i^Y} =  \mc{P} \times_G  \overline{\pi^{-1}(c_i^Y)}$. 

To condense the notation, we denote  
\begin{equation}\label{qmap counting data}
\mathcal{D}= ( \{c_i^X\}_{i=1}^l, \{c_i^Y\}_{i=1}^l|\mathcal{P})
\end{equation} 
and call it the  \emph{quasimap counting data}.

\begin{enumprob}
\label{enumprob QMap}
 Given a quasimap counting data (\ref{qmap counting data}), we consider the set 
    \begin{multline} \label{enum prob B eq}
        \QMap(X,Y,\mathcal{D})\! \defeq \! \left\lbrace f \in \QMap(X,Y) \colon f = (\mathcal{P},\underline{f}),\; 
        \underline{f}(c_i^X)
        \cap \widetilde{c_i^Y} \neq \varnothing\right\rbrace.
    \end{multline}
    If this set is finite we call the quasimap data \emph{stable} and we define the $\GLSM$ number 
    \begin{equation}
        \GLSM(X,Y,\mathcal{D}) \defeq \#  \QMap(X,Y,\mathcal{D}).
    \end{equation}
\end{enumprob}
\begin{remark}
    Suppose $f \in \QMap(X,Y,\mathcal{P},(c_i^X)_{i=1}^n,(c_i^Y)_{i=1}^n)$ and $\FR(f) = \varnothing$. Then $f$ defines a holomorphic map $f\colon X \to Y$ such that $f(c_i^X) \cap c_i^Y \neq \varnothing$, whose homotopy class $\delta$ is determined by the principal bundle $\mathcal{P}$. 
\end{remark}
\begin{remark}
    We can extend the $\GLSM$ number to collections of \emph{cycles} $c_i^X$, $c_i^Y$ by multilinearity.
\end{remark}
\subsubsection{Main example}
Let us consider the main example of quasimaps $\CP^k \qra \CP^n$. In this case $\mathcal{P} = O(d)$ for some $d \geq 1$ and we will denote the quasimap data simply by $\mathcal{D}= (\{c_i^X\}_{i=1}^l,\{c_i^Y\}_{i=1}^l|d).$ Quasimaps $f = (O(d), P^\alpha(x))$ are given by a collection of  $n+1$ homogeneous polynomials  in $k+1$ variables of degree $d$. 
We consider target cycles 
$c^Y_i \subset \CP^n$ such that the closure of the preimage  $\overline{\pi^{-1}(c^Y_i)}$ is a linear subspace of the same codimension and hence can be written as an intersection of $n_i^Y = \codim c_i^Y$ hyperplanes $H_i^j$. 
Then, quasimaps $f\in \QMap(\CP^k,\CP^n,\{c_i^X\}_{i=1}^l,\{c_i^Y\}_{i=1}^l)|d)$ are given by solutions to the system of equations 
\begin{equation}\label{eq: QM equation}
    H_i^j(P^\alpha(x_i)) = 0
\end{equation}
subject to the condition $x_i \in c_i^X$. Here the coefficients of these equations are the coefficients of $H_i^j$, the unknowns are the coefficients of the homogeneous polynomials $P^\alpha$ and the points $x_i$.  In particular, the virtual dimension of $\QMap(\CP^k,\CP^n,\{c_i^X\}_{i=1}^l,\{c_i^Y\}_{i=1}^l)|d)$ is given by 
\begin{multline}
    \label{QMap dim vir}
    \dim_\mr{vir}
    \QMap(\CP^k,\CP^n,(c_i^X)_{i=1}^l,(c_i^Y)_{i=1}^l|d) = \\
    =\n_{k,n,d} + \sum_{i=1}^l \dim c_i^X - \codim c_i^Y,
\end{multline}
cf. \eqref{dim_vir(M_d)}.
\begin{example}\label{expl: k1n1d, I}
    Consider the example $k=1, n=2$ and $d \geq 1$. For $i=1,2,3$ we let $c_i^X$ be points in $\CP^1$ and $c_i^Y$ be lines in $\CP^2$.  Let $l'$ be a natural number. For $i = 4, \ldots, l'+3\qefed l$, we let $c_i^X = \CP^1$ and $c_i^Y$ be points in $\CP^2$. Then the virtual dimension of the quasimap space $\QMap(\CP^k,\CP^n,(c_i^X)_{i=1}^l,(c_i^Y)_{i=1}^l|d)$ is 
    \begin{equation}
        \underbrace{3(d+1) - 1}_{n_{1,2,d}} + \underbrace{3 \cdot 0 + l' \cdot 1}_{\sum \dim c_i^X} - \underbrace{(2l' + 3)}_{\sum\codim c_i^Y} = 3d + 2 - l' - 3 = 3d -1 -l'.
    \end{equation} 
    We interpret this example as follows: By fixing the images of three points in $\CP^1$ to lie on lines we ``gauge fix'' the  $\PSL(2, \CC)$ action on $\CP^1$ (notice that this reduces the dimension of the quasimap space by $3 = \dim \PSL(2,\CC)$). 
    To have a valid quasimap count we need to demand that the quasi-map pass through $3d-1$ points in the target $Y$. In this way we recover the 
    zero virtual dimension
    condition for counting holomorphic maps. 
\end{example}

\subsection{Explicit answer from intersection theory}\label{sec:answer_from_intersection_theory}
In the main example, we can easily compute the $\GLSM$ number from intersection theory. Namely, we observe that given $H \in (\CC^{n+1})^\vee$, the map $(P^\alpha,x) \mapsto H(P^\alpha(x))$ can be interpreted as a section of the line bundle 
\begin{equation}
   L_d \defeq \OO(1) \boxtimes \OO(d) \to \QMap_d(\CP^k,\CP^n) \times \CP^k. 
\end{equation} 
Equation \eqref{eq: QM equation} can then be interpreted as the statement that such a section has a zero. 
We will use this observation to reformulate the quasimap count as a the count of zeros of a 
section. 
Namely, given a collection of submanifolds $c_i^X \subset \CP^k, i = 1, \ldots, k$, we define the compact complex manifold 
\begin{equation}
    \Var = \QMap_d(\CP^k,\CP^n) \times c_1^X \times \ldots \times c_l^X. \label{eq: def var}
\end{equation}
We have the obvious maps $p_i \colon \Var \to \QMap_d(\CP^k,\CP^n) \times \CP^k$ given by composing the inclusion $c_i^X \hookrightarrow \CP^k$ and the projection to the $i$th factor. We then have the following 
\begin{proposition}
 Given the quasimap counting data $\mathcal{D} = \left(\{c_i^X\}_{i=1}^l,\{c_i^Y\})_{i=1}^l|d\right)$ we define the following vector bundle $E$ over $\Var$ 
    \begin{equation}
       E \equiv E(\mathcal{D}) \defeq \bigoplus_{i=1}^l \oplus_{j=1}^{\codim c_i^Y} p_i^*L_d.
    \end{equation}
    Writing $\overline{\pi^{-1}(c_i^Y)}= \cap_{j=1}^{\codim c_i^Y} H_i^j$, 
    one has a section $\sigma$ of $E$
    given by 
    $$
    \sigma\colon \underbrace{(\underline{f},x_1,\ldots,x_l)}_{\in \Var}\mapsto \underbrace{\Big(\Big(H_i^j(\underline{f}(x_i))\Big)_{i=1}^l\Big)_{j=1}^{\codim c_i^Y}}_{\in E_{\ul{f},x_1,\ldots,x_l}}.$$
    Then
    \begin{equation}
        p(\sigma(\Var) \cap E_0) = \QMap(\CP^k,\CP^n,\mathcal{D}),
    \end{equation}
    where $E_0$ denotes the zero section of $E$ and $p$ is the projection to the first factor in the r.h.s. of (\ref{eq: def var}).
    
\end{proposition}
\begin{remark}
    The base space $\Var$ has dimension 
    \[
    \dim \Var = \n_{k,n,d} + \sum_i \dim c_i^X,
    \]
    while the vector bundle $E$ has rank 
    \[
    \rk(E) = \sum_i \sum_{j=1}^{\codim c_i^Y} \rk(p_i^* L_d) = \sum_i \codim c_i^Y.
    \]
    The virtual dimension can thus be expressed as
    \[
    \dim_{\mr{vir}} \QMap(\CP^k,\CP^n,\mathcal D) = \dim \Var - \rk(E).
    \]
\end{remark}
Observing that the number of zeroes of a generic section of a vector bundle $E$ is given by the Euler number $e(E)$ of that vector bundle, we obtain the following 
\begin{corollary}
For generic stable quasimap data $\mathcal{D}$, we have 
\begin{equation}
    \GLSM(\CP^k,\CP^n, \mathcal{D}) = e(E(\mathcal{D})).
\end{equation}
\end{corollary}
The Euler number can be readily computed from the fact that the Euler class $[\tilde{e}]$ of a complex vector bundle is its  top Chern class and that the Euler class 
is multiplicative under the Whitney sum of vector bundles 
$$\tilde{e}(E \oplus E') = \tilde{e}(E) \wedge \tilde{e}(E').$$ 

Let us denote by $H\in \Omega^2_{cl}(\QMap_d(\CP^k,\CP^n))$ 
a representative of the class Poincar\'e dual to a hyperplane in $\QMap_d(\CP^k,\CP^n)=\CP^{\n_{k,n,d}}$ and by $h\in \Omega^2_{cl}(\CP^k)$ -- a representative of the class Poincar\'e dual to a hyperplane in $\CP^k$. 
Then the first Chern class of $L_d$ is 
\begin{equation}
    c_1(L_d) = 
    H+h.
\end{equation}
(We suppress in the notation the pullbacks along the projections from $\QMap_d(\CP^k,\CP^n)\times \CP^k$ to the first or second summand.)
In particular, we have the following result. 
\begin{corollary}
For stable quasimap counting data $\mathcal{D}$, the $\GLSM$ number is given by the following integral
    \begin{equation}\label{eq: GLSM via Euler class}
        \GLSM(\CP^k,\CP^n,\mathcal{D}) = \int_{\Var} \tilde{e}(E) = \int_{\Var}\bigwedge_{i=1}^l(
        H+d\, h_i
        )^{\codim c_i^Y}.
    \end{equation}
    Here $h_i=p_i^* h$.
\end{corollary}
\begin{remark}
    Sometimes we will work with the space $\widetilde{\Var} = \QMap(X,Y) \times X^{\times l}$ instead. Then, for arbitrary representatives $\delta_1,\ldots,\delta_l$ of classes Poincar\'e dual to $c_1^X,\ldots,c_l^X$, we have 
    \begin{equation}
        \begin{split}
            \GLSM(\CP^k,\CP^n,\mathcal{D}) &= \int_{\widetilde{\Var}}\pi_1^*\delta_1\wedge\cdots \wedge \pi_l^*\delta_l  \wedge \tilde{e}(E) \\ 
            &= \int_{\widetilde{\Var}}\pi_1^*\delta_1\wedge\cdots \wedge \pi_l^*\delta_l\wedge \bigwedge_{i=1}^l(
                H+d\, \pi_i^* h
                )^{\codim c_i^Y}.
        \end{split}
    \end{equation}
    where $\pi_i\colon \widetilde{\Var}\ra \CP^k$ is the projection to the $i$th copy of $\CP^k$.
\end{remark}

\subsection{\texorpdfstring{$\GLSM$}{QM} numbers as sums of map counts and proper quasimap numbers}
Let us consider again the general case of quasimaps $f\colon X \not\to Y$. Imposing stable counting data $\mathcal{D}$ on the decomposition 
$$\QMap(X,Y) = \Map(X,Y) \sqcup \QMap^{\pr}(X,Y),$$
we obtain 
\begin{equation}
    \QMap(X,Y,\mathcal{D}) = \Map(X,Y,\mathcal{D}) \sqcup \QMap^{\pr}(X,Y,\mathcal{D}).
\end{equation} 
This immediately implies the following statement: 
\begin{equation}\label{eq: decomp1}
    \GLSM(X,Y,\mathcal{D}) = \KM(X,Y,\mathcal{D}) + \#\QMap^{\pr}(X,Y,\mathcal{D}),
\end{equation}
where $\KM(X,Y,\mathcal{D})$ is the Kontsevich-Manin number defined in (\ref{N^enum KM}).
\begin{definition}
    Given stable quasimap counting data $\mathcal{D}$ we denote the number of proper quasimaps in $\QMap(X,Y,\mathcal{D})$ by
    \begin{equation}
        \PQM(X,Y,\mathcal{D}) \defeq \#\QMap^{\pr}(X,Y,\mathcal{D}).
    \end{equation}
\end{definition}
We restate \eqref{eq: decomp1} in this terminology: 
\begin{proposition}\label{prop: QMap count stable}
    For $X,Y$ as in Section \ref{sec:setup} and stable quasimap counting data $\mathcal{D}$, the $\GLSM$ number is the sum of the Kontsevich-Manin and the proper quasimap numbers, 
    \begin{equation}
        \GLSM(X,Y,\mathcal{D}) = \KM(X,Y,\mathcal{D}) + \PQM(X,Y,\mathcal{D}).
    \end{equation}
\end{proposition}

\section{Counting quasimaps in non-stable cases via excess intersection theory
}
\subsection{Quasi-stable case}
As we will see in the next section, 
in many examples the counting data turns out to be \emph{unstable}, i.e.\ the zero locus of the section $\sigma$ has components of positive dimension. 
 In this case 
the $\GLSM$ number can be defined by the integral \eqref{eq: GLSM via Euler class}, i.e.\ 
\begin{equation}
    \GLSM(X,Y,\mathcal{D}) \defeq \int_{\Var} \tilde{e}(E). 
\end{equation}
The right hand side no longer has an interpretation as counting the zeros of the section $\sigma$. 
However, one can still make sense of the formula 
$$\GLSM(X,Y,\mathcal{D}) = \KM(X,Y,\mathcal{D}) + \PQM(X,Y,\mathcal{D})$$ 
via \emph{excess intersection theory} \cite{Fulton} (see also the nice introduction in \cite[Chapter 8]{Katz}). We denote by ${\rm Z}(\sigma)$ the zero set  of the section $\sigma$ restricted to proper quasimaps. Then, suppose that 
$$Z  
\subseteq {\rm Z}(\sigma)$$  
is a connected component,
possibly of positive dimension.
If $Z$ is a smooth submanifold of $\Var$,  we can define the \emph{excess bundle} on $Z$ by
\begin{equation}
    B_Z = \frac{E\rvert_Z}{\sigma_*(N_{Z})},
\end{equation} 
where $N_Z$ is the normal bundle of $Z$ in $\Var$ that we identify as a subbundle of $E$ via the section $\sigma$. We give a special name to configurations where all components of the zero set are of this form: 
\begin{definition}
    In the situation of the Enumerative Problem \ref{enumprob QMap}, suppose that $\KM(X,Y,\mathcal{D})$ is finite and ${\rm Z}(\sigma)$ 
    is a union of connected components that are submanifolds defined by transverse intersections, 
$$
{\rm Z}(\sigma)
= \bigsqcup_i Z_i.$$ 
Then, we call the counting data $\mathcal{D}$ \emph{quasi-stable}.
\end{definition}
For quasi-stable counting data, we have the following result.
\begin{proposition}
   Suppose the counting data is quasi-stable and denote $Z_i$ the connected components of ${\rm Z}(\sigma)$. 
   Let $c(B_{Z_i})$ denote the total Chern class of the bundle $B_{Z_i}$. Then
   $$\GLSM(X,Y,\mathcal{D}) = \int_{\Var}\tilde{e}(E) =  \KM(X,Y,\mathcal{D}) + \sum_i \int_{Z_i} c(B_{Z_i}). $$ 
\end{proposition}
\begin{proof}
    For a holomorphic vector bundle $E \to \Var$, the self-intersection class of the zero section is the the Euler class of $\Var$, 
    
$$[E_0] \cdot [E_0] = c_{\rk\, E } \frown [\Var] \in A_0(\Var),$$ see e.g. \cite[Example 2.9]{Aluffi}. Here $c_{\rk\, E }$ denotes the $\rk\, E$ Chern class of $E$. Since the image $\sigma(\Var)$ of a holomorphic section $\sigma\colon \Var \to E$ is rationally equivalent to the zero section, the intersection product $[\sigma(\Var)] \cdot [E_0] = c_{\rk\, E} \frown [\Var]$ as well. Then, we have that 
\begin{equation}\label{excess proof eq}
\sigma(\Var) \cap E_0 = \widehat\Map(X,Y,\mathcal{D}) \sqcup \bigsqcup_{i} Z_i.
\end{equation}
In formula (\ref{excess proof eq}), $\widehat\Map(X,Y,\mathcal{D})$ is the intersection of the zero set of $\sigma$ with $\Map_d(X,Y)\times \prod_i c_i^X$.
 By a B\'ezout-like theorem, on the one hand the degree of $[\sigma(\Var)] \cdot [E_0]$ is the quasimap number $\QM(X,Y,\mathcal{D})$. On the other hand, the intersection product $[\sigma(\Var)] \cdot [E_0]$ decomposes (in Fulton's terminology, this is ``canonical decomposition'') as the sum of pieces supported on connected components in the r.h.s. of (\ref{excess proof eq}). From \cite[Proposition 9.1.1]{Fulton}, the contribution to $[\sigma(\Var)] \cdot [E_0]$ supported on any smooth component $Z_i$ is precisely $$c(E|_{Z_i})c(N_{Z_i})^{-1}\frown [Z_i] = c(B_{Z_i})\frown [Z_i].$$ 
 Passing to degrees of $0$-cycles and summing over connected components, we obtain the statement.
 \end{proof}
Then we define the proper quasimap number of $Z_i$ by
\begin{equation}
   \PQM(X,Y,\mathcal{D};Z_i)\defeq  \int_{Z_i} c(B_{Z_i})\label{eq: def freckle number stratum}
\end{equation}  and the total proper quasimap number as 
\begin{equation} \PQM(X,Y,\mathcal{D}) \defeq \sum_i \PQM(X,Y,\mathcal{D};Z_i).\label{eq: def freckle number quasi stable}
\end{equation}
In this way, we obtain the generalization of Proposition \ref{prop: QMap count stable}:
\begin{proposition}
    For X,Y as in Section \ref{sec:setup} and quasi-stable counting data $\mathcal{D}$, we have 
    \begin{equation}
        \GLSM(X,Y,\mathcal{D}) = \KM(X,Y,\mathcal{D}) + \PQM(X,Y,\mathcal{D}). 
    \end{equation}
\end{proposition}
We remark that because of the splitting principle, for any component $Z \subset \QMap(X,Y,\mathcal{D})$ we can compute the total Chern class $c(B_{Z})$ appearing in  \eqref{eq: def freckle number stratum} as 
\begin{equation}
    c(B_Z) = \frac{c(E\rvert_Z)}{c(N_{Z})} = \frac{c(E\rvert_Z)c(Z)}{c(\Var)\rvert_Z}.\label{eq: compute chern class B}
\end{equation}

 For later convenience, let us write out $c(B_Z)$ in detail: 
The external product $E_1 \boxtimes E_2 \to X_1 \times X_2$ of two vector bundles  $E_i\to X_i$, $i = 1,2$, is defined by the tensor product $\pr_1^*E_1 \otimes \pr_2^*E_2$, where $\pr_i \colon X_1 \times X_2 \to X_i$ is the projection to the first, respectively the second factor.
Now recall from Section \ref{sec:answer_from_intersection_theory} that for a given quasimap counting data $\mathcal D$, we defined the space of variables 
\[
\Var = \QMap_d(\CP^k,\CP^n) \times c_1^X \times \dots \times c_l^X
\]
and the vector bundle of equations $E \to \Var$  
\begin{equation}
    E = \bigoplus_{i=1}^l \oplus_{j=1}^{\codim c_i^Y} p_i^*L_d.
\end{equation}
By the Whitney sum formula and naturality of the total Chern class, one has 
\begin{equation}
     c(E\rvert_Z) = \prod_{i=1}^l \prod_{j=1}^{\codim c_i^Y} p_i^*c(L_d)\rvert_Z,
\end{equation}
where
\begin{equation}
    p_i^* c(L_d)|_Z=(1+H+d\, h_i)|_Z.
\end{equation}
Moreover, since $\QMap_d(\CP^k,\CP^n) = \CP^{\n_{k,n,d}}$,
\begin{equation}
    T\Var = \pi^*_{\CP^{\n_{k,n,d}}} T\CP^{\n_{k,n,d}}\oplus  \bigoplus_{i=1}^l \pi^*_{c_i^X} Tc_i^X,
\end{equation} 
where $\pi_{\CP^{n_{k,n,d}}} \colon \Var \to \CP^{\n_{k,n,d}}$ and $\pi_{c_i^X} \colon \Var \to c_i^X$ denote the projections to the components of $\Var$.
Hence,
\begin{equation}
    \begin{split}
    c(B_Z) &= \frac{\prod_{i=1}^l \prod_{j=1}^{\codim c_i^Y} 
    (1+H+d\, h_i)|_Z
    \cdot c(Z) }{\pi^*_{\CP^{\n_{k,n,d}}}c(\CP^{\n_{k,n,d}} ) \big\rvert_Z\ \prod_{i=1}^l \pi^*_{c_i^X} c(c_i^X)\big\rvert_Z}
     \\
    &=\frac{\prod_{i=1}^l \prod_{j=1}^{\codim c_i^Y} 
    (1+H+d\, h_i)|_Z
    \cdot c(Z) }{(1+H)^{\n_{k,n,d}+1} \big\rvert_Z\ \prod_{i=1}^l \pi^*_{c_i^X} c(c_i^X)\big\rvert_Z}.
    \label{eq:BZ}
\end{split}
\end{equation}

\subsection{General unstable case. First appearance of Segre classes.} \label{sec: unstable}
In the general case, connected components $Z \subset {\rm Z}(\sigma)$ need not be smooth submanifolds of $\Var$. 
In this case, $Z$ does not have a well-defined normal bundle 
and hence we cannot apply the formulas given above directly. 
To generalize them to the unstable case, first note that if $\mc{D}$ is quasi-stable we can rewrite \eqref{eq: def freckle number stratum} as 
$$\PQM(X,Y,\mathcal{D};Z_i)=  \int_{Z_i} c(B_{Z_i}) = \int_{\Var} c(E)c(N_{Z_i})^{-1}\eta_{Z_i},$$
where $\eta_{Z_i}$ denotes a representative of the class Poincar\'e dual to $Z_i$. For the unstable case, $c(N_{Z_i})^{-1}\eta_{Z_i}$ should be replaced with the Segre class $s(Z_i,\Var)$ of $Z_i$ in $\Var$. Recall (see e.g. \cite{Aluffi}) that for a variety $Y$, and a closed embedding $X \subset Y$, there is a Segre class $s^\vee(X,Y) \in A_*(X)$ (the Chow ring of $X$)
uniquely characterized by the properties that 
\begin{enumerate}[i)] 
\item  for regular embeddings $s^\vee(X,Y) = c^\vee(N_XY)^{-1} \cap [X]$ and 
\item for $f \colon Y' \to Y$ a proper, onto, birational morphism of varieties, we have 
\begin{equation} \label{eq: segre under f} s^\vee(X,Y) = (f|_{f^{-1}(X)})_*s^\vee(f^{-1}(X),Y'). \end{equation}
\end{enumerate} 
We will denote without further comment\footnote{Our notation is somewhat opposite to the usual algebraic geometry conventions because we work in the differential form formalism.} by $s(Z,\Var)\in \Omega^\bullet(\Var)$ a representative of a class Poincar\'e dual to the pushforward of the Segre class to the Chow ring of $\Var$. We then define the $\PQM$ number of $Z_i$ as 
\begin{equation}
    \PQM(X,Y,\mathcal{D};Z_i) \defeq \int_{\Var} c(E)s(Z_i,\Var). \label{eq: def freckle number singular stratum}
\end{equation}
Defining the total $\PQM$ number again by \eqref{eq: def freckle number quasi stable}, we see that with this definition Proposition \ref{prop: QMap count stable} extends to the unstable case.  


In general, computation of Segre classes is a hard problem, 
we will defer examples of such computations to the next paper.

\section{Enumerative examples}
In this section we consider various examples for the count of quasimaps and the decomposition of the $\GLSM$ number into holomorphic maps and proper quasimaps. 
We consider a number of examples. In these examples we will meet 1- and 2-dimensional sources. For 2-dimensional sources we also study 1-dimensional cycles in the source. We will meet both freckles and scars. Examples range in complexity. In the simplest cases, it is just computation of the number of freckle configurations. In more complicated cases, we have to apply excess intersection theory. In the simplest (quasi-stable) cases it easily doable. In more complicated cases we have to invoke Segre classes. However, there are two examples where $\PQL$ locus is so peculiar, see Figures \ref{fig:non_triv_intersection_strata_64} and \ref{Fig 12} that we postpone computation to a subsequent paper. 

\subsection{Quasimaps \texorpdfstring{$\CP^1 \qra \CP^2$}{CP1 -> CP2}} 
Let us first consider quasimaps  $\CP^1 \qra \CP^2$ of degree $1$ and $2$, 
so that
\begin{align*}
    \dim \QMap_1(\CP^1,\CP^2) &= 5,\\
    \dim \QMap_2(\CP^1,\CP^2) &= 8.
\end{align*}
\subsubsection{\texorpdfstring{$0+2 = 2$}{0+2=2} - An example with no holomorphic maps} \label{sss: example 0+2=2}
Take $
d=1$. The minimal amount of cycles we need for a stable quasimap counting data  is $l=3$, with $c^X_i$ a point for $i=1,2$, $c_3^X = \CP^1$ and $c_i^Y$ a point for $i=1,2,3$, cf.\ Figure \ref{fig: 0 + 2 = 2}.
\begin{figure}[H]
    \begin{subfigure}[t]{0.3\textwidth}
    \centering
        \begin{tikzpicture}[baseline={([yshift=-.5ex]current bounding box.center)},scale=.6]
            \draw (0,0) circle (3);
        
            \coordinate[point] (c2) at (100:2);
            \coordinate[point] (c1) at (240:1.5);
        
            \coordinate[point] (x) at (0:1.5);
        
            \begin{scope}[every node/.style={transform shape}]
                \draw (x) circle (.5);
                \draw pic at (x) {H={45}{.5}};
                \draw pic at (x) {H={-45}{.5}};
        
                \draw pic at (c1) {H={45}{.5}};
                \draw pic at (c1) {H={-45}{.5}};
        
                \draw pic at (c2) {H={45}{.5}};
                \draw pic at (c2) {H={-45}{.5}};
            \end{scope}
        
            \node[left] at ($(c1) + (-.5,0)$) {$c_1^X$};
            \node[left] at ($(c2) + (-.5,0)$) {$c_2^X$};
        
            \node[below] at ($(x) + (0,-.5)$) {$x$};
        
        \end{tikzpicture} 
    \caption{Pictorial representation of the enumerative problem. }
    \label{fig: 0 + 2 = 2}
    \end{subfigure}
    \begin{subfigure}[t]{0.6\textwidth}
    \centering
    \begin{tikzpicture}[baseline={([yshift=-.5ex]current bounding box.center)},scale=.6]
        \draw (0,0) circle (3);

        \coordinate[point] (c2) at (100:2);
        \coordinate[point] (c1) at (240:1.5);
        \coordinate[point] (x) at (0:1.5);

        \begin{scope}[every node/.style={transform shape}]
            \draw (x) circle (.5); 
            
            \draw pic at (x) {H={45}{.5}};
            \draw pic at (x) {H={-45}{.5}};
            
            \draw pic at (c1) {H={45}{.5}};
            \draw pic at (c1) {H={-45}{.5}};

            \draw pic at (c2) {H={45}{.5}};
            \draw pic at (c2) {H={-45}{.5}}; 
        \end{scope}

        \node[left] at ($(c1)+(-.5,0)$) {$c_1^X$};
        \node[left] at ($(c2)+(-.5,0)$) {$c_2^X$};
        \node[below] at ($(x)+(0,-.5)$) {$x$};
        
        \node[red] (fr) at (0,0) {$*$};

        \node (xx) at (x) {$\phantom{*}$}; 
        \draw[->,thick,red] (fr) to[bend right] ($(c1)+(90:.3)$);
        \draw[->,thick,red] (xx) to[bend left] ($(c1)+(0:.3)$);
        
      \end{tikzpicture}\hfill
      \begin{tikzpicture}[baseline={([yshift=-.5ex]current bounding box.center)},scale=.6]
        \draw (0,0) circle (3);

        \coordinate[point] (c2) at (100:2);
        \coordinate[point] (c1) at (240:1.5);
        \coordinate[point] (x) at (0:1.5);

        \begin{scope}[every node/.style={transform shape}]
            \draw (x) circle (.5); 
            
            \draw pic at (x) {H={45}{.5}};
            \draw pic at (x) {H={-45}{.5}};
            
            \draw pic at (c1) {H={45}{.5}};
            \draw pic at (c1) {H={-45}{.5}};

            \draw pic at (c2) {H={45}{.5}};
            \draw pic at (c2) {H={-45}{.5}}; 
        \end{scope}

        \node[left] at ($(c1)+(-.5,0)$) {$c_1^X$};
        \node[left] at ($(c2)+(-.5,0)$) {$c_2^X$};
        \node[below] at ($(x)+(0,-.5)$) {$x$};
        
        \node[red] (fr) at (0,0) {$*$};

        \node (xx) at (x) {$\phantom{*}$}; 
        
        \draw[->,thick,red] (fr) to ($(c2)+(-90:.3)$);
        \draw[->,thick,red] (xx) to[bend right] ($(c2)+(0:.3)$);
      \end{tikzpicture}
     \caption{The two freckle contributions to the $\QM$ number (the freckle is pictorially denoted by a star).
     }
    \label{fig: 0 + 2 = 2 2}
    \end{subfigure}
        \caption{Configuration of source and target cycles in the $0+2 = 2$ example. We have two fixed points $c_1^X = (0:1) = 0$, $c_2^X = (1:0) =\infty$ in the source and one running point $x$ (the circle around $x$ indicates it is running.) Number of lines $\lambda_i$ through a source point denotes the codimension of corresponding target cycle $c_i^Y$ which is given as intersection of $\lambda_i $ hyperplanes in the target $\CP^2$ (in this case, $\lambda_i \equiv 2$.) }
\end{figure}
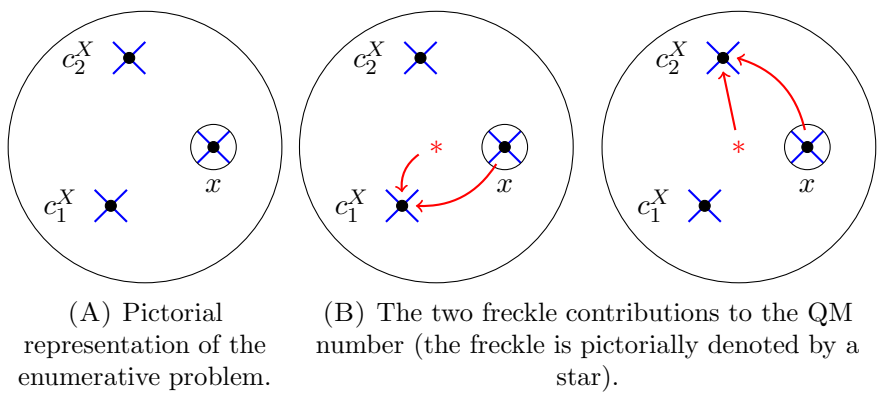

This counting data is stable and 
$$\GLSM(\CP^1,\CP^2,\mathcal{D}) = \int_{\CP^5 \times \CP^1_{(3)}}
H^2 H^2 (H+p_3^* h)^2
=2, $$
with notations $H,h$ as in (\ref{eq: GLSM via Euler class}); the subscript in $\CP^1_{(3)}$ reminds that it is the third source cycle.
However, for $c_i^Y$ in general position there is no degree 1 holomorphic map $\CP^1 \to \CP^2$ passing through all of them, as three points in $\CP^2$ are generically not contained in a line. Hence, there must be exactly two proper quasimaps for this counting data. Indeed, a degree 1 quasimap sending two fixed points in $\CP^1$ to two fixed points in $\CP^2$ is, up to a change of coordinates, of the form 
$$\underline{f}(x^0:x^1) = (ax^0:bx^1:0),$$
where $(a:b) \in \CP^1$ are projective coordinates (put differently, this corresponds to fixing $c_1^X = (1:0), c_2^X = (0:1), c_1^Y = (1:0:0), c_2^Y = (0:1:0))$. In particular, either $a$ or $b$ can be zero, in which the quasimap has a freckle, at the point $(1:0)$ for $a=0$ or at the point $(0:1)$ for $b=0$. Explicitly, the two proper quasimaps are given by 
\begin{align*}
    \underline{f}_1(x^0:x^1)&=(x^0:0:0), \\
    \underline{f}_2(x^0:x^1)&=(0:x^1:0).
\end{align*}
That is, the two proper quasimap solutions to the enumerative problem have a freckle at $c_i^X$, $i=1,2$. 
Note that in this case, the proper quasimap solves the equation at the running point exactly when the running point hits the freckle, i.e.\ sits equally at $c_i^X$.
The situation is schematically depicted in Figure \ref{fig: 0 + 2 = 2 2}.
One might think that those violate the condition $f(c_i^X) = c_i^Y$, but this is not the case: Indeed, as a quasimap $\underline{f}$ maps to $\CC^3$ (rather than $\CP^2$) and we ask that it intersects the \emph{lifts} of the target cycles, evaluating $\underline{f}$ at the freckle yields 0 which lies in the intersection of all lines defining the cycle $c_i^Y$.
In general, one observes the \emph{freckle principle}: 
\begin{center}
at the freckle \emph{all} equations are satisfied.
\end{center}
\begin{remark}
    Note that in this example there is a $1_\CC$-dimensional subgroup of $\PSL(2,\CC)$ preserving the source cycles.
\end{remark}

\subsubsection{\texorpdfstring{$1 + 1 = 2$}{1+1=2}}\label{sec:1+1=2}
For the next example, let again $d=1$. 
We consider $l=4$ source cycles with $c_i^X$ a point for $i=1,2,3$ and $c_4^X = \CP^1$, i.e.\ 3 fixed and 1 running point. For the target cycles we choose hyperplanes (i.e.\ lines) for $c_1^Y,c_2^Y$ and points for $c_3^Y, c_4^Y$. Then, we can compute the quasimap number as follows:
\begin{equation}
    \GLSM(\CP^1,\CP^2,\mathcal{D}) = \int_{\CP^5 \times \CP^1_{(4)}} 
    H\,H\, H^2\, (H+ p_4^* h)^2
    = 2.
\end{equation}
We know that there is a unique holomorphic map sending the points $c_i^X \mapsto c_i^Y$ for $i = 3,4.$ Hence the decomposition of the $\GLSM$ number into holomorphic maps and proper quasimaps is 
\begin{equation}
    \GLSM = 2 = \underbrace{1}_{\KM} + \underbrace{1}_{\PQM}.
\end{equation}
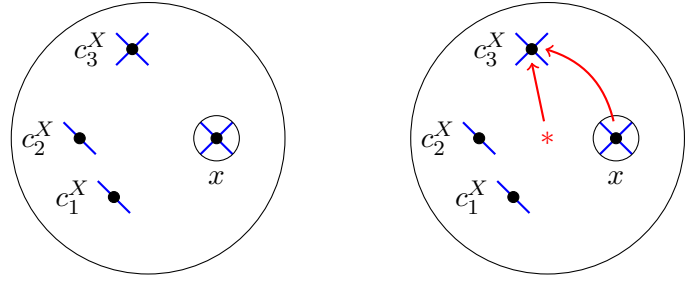
\begin{figure}[H]
\begin{subfigure}[t]{0.4\textwidth}
    \centering
    \begin{tikzpicture}[baseline={([yshift=-.5ex]current bounding box.center)},scale=.6]
    \draw (0,0) circle (3);
    \coordinate[point] (c3) at (100:2);
    \coordinate[point] (c2) at (180:1.5);
    \coordinate[point] (c1) at (240:1.5);
    \coordinate[point] (x) at (0:1.5);

    \begin{scope}[every node/.style={transform shape}]
        \draw (x) circle (.5);

        \draw pic at (x) {H={45}{.5}};
        \draw pic at (x) {H={-45}{.5}};

        \draw pic at (c3) {H={45}{.5}};
        \draw pic at (c3) {H={-45}{.5}};

        \draw pic at (c1) {H={-45}{.5}};
        
        \draw pic at (c2) {H={-45}{.5}};
    \end{scope}

    \node[left] at ($(c1)+(-0.3,0)$) {$c_1^X$};
    \node[left] at ($(c2)+(-0.3,0)$) {$c_2^X$};
    \node[left] at ($(c3)+(-0.3,0)$) {$c_3^X$};

    \node[below] at ($(x)+(0,-0.5)$) {$x$};
  \end{tikzpicture} 
\label{fig: 1 + 1 = 2}
\caption{Pictorial representation of the enumerative problem. }
    \end{subfigure}
    \begin{subfigure}[t]{0.4\textwidth}
            \centering
    \begin{tikzpicture}[baseline={([yshift=-.5ex]current bounding box.center)},scale=.6]
    \draw (0,0) circle (3);
    \coordinate[point] (c3) at (100:2);
    \coordinate[point] (c2) at (180:1.5);
    \coordinate[point] (c1) at (240:1.5);
    \coordinate[point] (x) at (0:1.5);

    \begin{scope}[every node/.style={transform shape}]
        \draw (x) circle (.5);

        \draw pic at (x) {H={45}{.5}};
        \draw pic at (x) {H={-45}{.5}};

        \draw pic at (c3) {H={45}{.5}};
        \draw pic at (c3) {H={-45}{.5}};

        \draw pic at (c1) {H={-45}{.5}};
        
        \draw pic at (c2) {H={-45}{.5}};
    \end{scope}

    \node[left] at ($(c1)+(-0.3,0)$) {$c_1^X$};
    \node[left] at ($(c2)+(-0.3,0)$) {$c_2^X$};
    \node[left] at ($(c3)+(-0.3,0)$) {$c_3^X$};

    \node[below] at ($(x)+(0,-0.5)$) {$x$};

    \node[red] (fr) at (0,0) {$*$};

    \node (xx) at (x) {$\phantom{*}$}; 
    \draw[red,thick,->] (xx) to[bend right] ($(c3) + (0:.3)$);
    \draw[red,thick,->] (fr) to ($(c3) + (-90:.3)$);
  \end{tikzpicture} 
     \caption{The  freckle contributions 
     to the $\GLSM$ number (the freckle is pictorially denoted by a star).}
    \end{subfigure}
        \caption{Configuration of source and target cycles in the $1+1 = 2$ example. We have three fixed points $c_1^X = (0:1) = 0$, $c_2^X = (1:1) =1, c_3^X = (0:1) = \infty$ in the source and one running point $x$. Target cycles are hyperplanes $c_1^Y, c_2^Y$ and points $c_3^Y,c_4^Y$. }
\end{figure}

It is instructive to derive this decomposition directly. Namely,  consider a quasi-map $\CP^1\to \CP^2$ given by three degree 1 homogeneous polynomials 
\begin{equation*}
\underline{f}(x^0:x^1) = (P^0(x^0:x^1),P^1(x^0:x^1),P^2(x^0:x^1)),
\end{equation*}
with $P^i(x^0:x^1) = a^i_0x^0 + a^i_1x^1$. 
We fix the source cycles to be points $$c_1^X = (0:1),\ c_2^X = (1:1),\  c_3^X = (1:0)$$ and the target cycles to be the hyperplanes $$c_1^Y=(0:y^1:y^2),\ c_2^Y = (y^0 : y^1 : y^0 + y^1)$$ and $c_3^Y = (1:0:0)$, $c_4^Y = q$ a generic point in $\CP^2$. 
These quasimaps are parametrized by 
\begin{equation*}
\underline{f}(x^0:x^1) = (ax^0:bx^1:(a+b)x^1). 
\end{equation*}
When $a = 0$, then $\underline{f}$ defines a proper quasimap.
Indeed, the quasimap
\begin{equation*}
    \underline{f}(x^0:x^1) = (0:bx^1:bx^1)
\end{equation*}
vanishes at $c^X_3 = (1:0)$ and sends the complement to $c_1^Y \cap c_1^Y = (0:1:1)$.
We see that in this case $$\PQM(\CP^1,\CP^2,\mathcal{D}) =1.$$

For $a\neq 0$, we can use the $\CC^*$ action on quasimaps to set it to 1 so that 
\begin{equation}
\underline{f}(x^0:x^1) = (x^0:bx^1:(1+b)x^1).
    \label{eq: 1 + 1 = 2 qmap}
\end{equation}
The last condition $f(c_4^X) = q$ fixes the value of $b$ and therefore $$\KM(\CP^1,\CP^2,\mathcal{D}) = 1.$$
\begin{remark}
    This example has a generalization to the case when the dimension $n$ of the target is bigger than 2 (see Section \ref{sec: 1 + N-1 = N}). 
    However, in that case the map is no longer uniquely fixed on the complement of the freckle, because the intersection of two generic hypersurfaces $c_1^Y$ and $c_2^Y$ has positive dimension in $\CP^n$ for $n > 2$. 
\end{remark}

\subsubsection{\texorpdfstring{$1+3 = 4$}{1+3=4}} \label{sec:1+3=4} 
We consider again the case 
$d=1$. 
Consider the counting data described in Example \ref{expl: k1n1d, I}, i.e.\ 
$l = 3d-1 + 3 = 5$, with $c_i^X$ a point for $i = 1,2,3$ and $c_i^X = \CP^1$ for $i=4,5$, $c_i^Y$ a line for $i=1,2,3$ and $c_i^Y$ a point for $i=4,5$, cf.\ Figure \ref{fig: 1 + 3 = 4}.

\begin{figure}[H]
\begin{subfigure}[t]{1\textwidth}
    \centering
    \begin{tikzpicture}[baseline={([yshift=-.5ex]current bounding box.center)},scale=.6]
    \draw (0,0) circle (3);
    \coordinate[point] (c3) at (108:2);
    \coordinate[point] (c2) at (180:1.5);
    \coordinate[point] (c1) at (252:2);

    \coordinate[point] (x) at (36:1.5);
    \coordinate[point] (y) at (-36:1.5);

    \begin{scope}[every node/.style={transform shape}]
        \draw (x) circle (.5);
        \draw (y) circle (.5);
        
        \draw pic at (x) {H={45}{.5}};
        \draw pic at (x) {H={-45}{.5}};

        \draw pic at (y) {H={45}{.5}};
        \draw pic at (y) {H={-45}{.5}};

        \draw pic at (c1) {H={-45}{.5}};
        
        \draw pic at (c2) {H={-45}{.5}};

        \draw pic at (c3) {H={-45}{.5}};
    \end{scope}

    \node[left] at ($(c1)+(-0.3,0)$) {$c_1^X$};
    \node[left] at ($(c2)+(-0.3,0)$) {$c_2^X$};
    \node[left] at ($(c3)+(-0.3,0)$) {$c_3^X$};

    \node[below] at ($(x)+(0,-0.5)$) {$x$};
    \node[below] at ($(y)+(0,-0.5)$) {$y$};

  \end{tikzpicture} 
\caption{Pictorial representation of the enumerative problem. }
\label{fig: 1 + 3 = 4}
    \end{subfigure}\vspace{1cm}
\begin{subfigure}[t]{1\textwidth}
    \centering
    \begin{tikzpicture}[baseline={([yshift=-.5ex]current bounding box.center)},scale=.6]
    \draw (0,0) circle (3);
    \coordinate[point] (c3) at (108:2);
    \coordinate[point] (c2) at (180:1.5);
    \coordinate[point] (c1) at (252:2);

    \coordinate[point] (x) at (36:1.5);
    \coordinate[point] (y) at (-36:1.5);

    \begin{scope}[every node/.style={transform shape}]
        \draw (x) circle (.5);
        \draw (y) circle (.5);
        
        \draw pic at (x) {H={45}{.5}};
        \draw pic at (x) {H={-45}{.5}};

        \draw pic at (y) {H={45}{.5}};
        \draw pic at (y) {H={-45}{.5}};

        \draw pic at (c1) {H={-45}{.5}};
        
        \draw pic at (c2) {H={-45}{.5}};

        \draw pic at (c3) {H={-45}{.5}};
    \end{scope}

    \node[left] at ($(c1)+(-0.3,0)$) {$c_1^X$};
    \node[left] at ($(c2)+(-0.3,0)$) {$c_2^X$};
    \node[left] at ($(c3)+(-0.3,0)$) {$c_3^X$};

    \node[below] at ($(x)+(0,-0.5)$) {$x$};
    \node[below] at ($(y)+(0,-0.5)$) {$y$};

    \node[red] (fr) at (0,0) {$*$};
    \draw[thick,red,->] (fr) to[bend left] ($(c1)+(45:.3)$);

  \end{tikzpicture} \hfill
  \begin{tikzpicture}[baseline={([yshift=-.5ex]current bounding box.center)},scale=.6]
    \draw (0,0) circle (3);
    \coordinate[point] (c3) at (108:2);
    \coordinate[point] (c2) at (180:1.5);
    \coordinate[point] (c1) at (252:2);

    \coordinate[point] (x) at (36:1.5);
    \coordinate[point] (y) at (-36:1.5);

    \begin{scope}[every node/.style={transform shape}]
        \draw (x) circle (.5);
        \draw (y) circle (.5);
        
        \draw pic at (x) {H={45}{.5}};
        \draw pic at (x) {H={-45}{.5}};

        \draw pic at (y) {H={45}{.5}};
        \draw pic at (y) {H={-45}{.5}};

        \draw pic at (c1) {H={-45}{.5}};
        
        \draw pic at (c2) {H={-45}{.5}};

        \draw pic at (c3) {H={-45}{.5}};
    \end{scope}

    \node[left] at ($(c1)+(-0.3,0)$) {$c_1^X$};
    \node[left] at ($(c2)+(-0.3,0)$) {$c_2^X$};
    \node[left] at ($(c3)+(-0.3,0)$) {$c_3^X$};

    \node[below] at ($(x)+(0,-0.5)$) {$x$};
    \node[below] at ($(y)+(0,-0.5)$) {$y$};

    \node[red] (fr) at (0,0) {$*$};
    \draw[thick,red,->] (fr) to ($(c2)+(0:.3)$);

  \end{tikzpicture} \hfill
  \begin{tikzpicture}[baseline={([yshift=-.5ex]current bounding box.center)},scale=.6]
    \draw (0,0) circle (3);
    \coordinate[point] (c3) at (108:2);
    \coordinate[point] (c2) at (180:1.5);
    \coordinate[point] (c1) at (252:2);

    \coordinate[point] (x) at (36:1.5);
    \coordinate[point] (y) at (-36:1.5);

    \begin{scope}[every node/.style={transform shape}]
        \draw (x) circle (.5);
        \draw (y) circle (.5);
        
        \draw pic at (x) {H={45}{.5}};
        \draw pic at (x) {H={-45}{.5}};

        \draw pic at (y) {H={45}{.5}};
        \draw pic at (y) {H={-45}{.5}};

        \draw pic at (c1) {H={-45}{.5}};
        
        \draw pic at (c2) {H={-45}{.5}};

        \draw pic at (c3) {H={-45}{.5}};
    \end{scope}

    \node[left] at ($(c1)+(-0.3,0)$) {$c_1^X$};
    \node[left] at ($(c2)+(-0.3,0)$) {$c_2^X$};
    \node[left] at ($(c3)+(-0.3,0)$) {$c_3^X$};

    \node[below] at ($(x)+(0,-0.5)$) {$x$};
    \node[below] at ($(y)+(0,-0.5)$) {$y$};

    \node[red] (fr) at (0,0) {$*$};
    \draw[thick,red,->] (fr) to[bend left] ($(c3)+(-90:.3)$);

  \end{tikzpicture} 
\caption{ Proper quasimap solutions. 
}
    \end{subfigure}
    \caption{ $4=1+3$ enumerative problem 
    }
\end{figure}
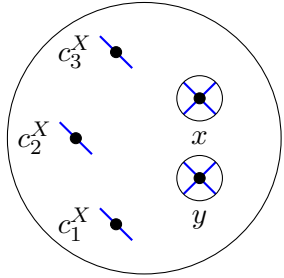
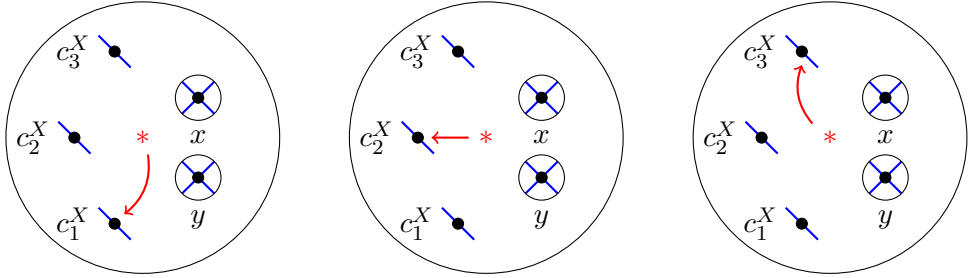

The quasimap number 
for this problem is
\begin{equation}
    \GLSM(\CP^1,\CP^2,\mathcal{D}) = \int_{\CP^5\times\CP^1_{(4)}\times\CP^1_{(5)}}
    H^3 (H+p_4^* h)^2 (H+p_5^* h)^2
    = 4.
\end{equation}
We claim that in this case there are three proper quasimaps and a unique holomorphic map satisfying the conditions of the enumerative problem, so that 
$$\underbrace{4}_{\QM}=\underbrace{1}_{\KM}+\underbrace{3}_{\PQM}.$$ 
Indeed, let us choose $c_i^X = \{0,1,\infty\}$ as above and lines $c_i^Y = \{y^i = 0\}$, so that our quasimap satisfies 
$$P^i(c_i^X) = 0.$$
It is easy to see that such a quasi-map is given by 
\begin{equation}
\underline{f}(x^0:x^1) = (ax^0: b(x^1-x^0):cx^1),
    \label{eq: 1 + 3 = 4 qmap}
\end{equation}
where the parameters $a,b,c$ form a two-dimensional projective space, $(a:b:c) \in \CP^2$. In particular, the quasimap has a freckle if and only if two out of the three parameters vanish, and in this case the freckle will be at one of the fixed points $c_i^X$. If at most one of $a,b,c$ is zero, then $\underline{f}$ defines a holomorphic map $f\colon \CP^1 \to \CP^2$, which is uniquely fixed by the requirement that it passes through the points $c_4^Y$ and $c_5^Y$. 
\begin{remark}
    Maps $\CP^1 \to \CP^2$ of degree 1 describe parametrized lines in $\CP^2$.
    By demanding that the three points $c_i^X$, $i = 1,2,3$ have images in prescribed lines, one fixes the parametrization.   
    Geometrically, we are thus counting 
    lines through two given points in $\CP^2$ and it is clear that there exists a unique such map.
\end{remark}

 \subsection{Higher-dimensional target}
\subsubsection{\texorpdfstring{$2+0 = 2$}{2+0=2} - An example with no freckles}
Let $k=1$, $n=3$, $d=1$ - i.e. we are counting quasimaps $\CP^1 \qra \CP^3$ of degree 1. We have $\dim \QMap_1(\CP^1,\CP^3) = 7.$
Consider $c_i^X$ to be a point for $i=1,2,3$ and $c_4^X = \CP^1$. 
As target cycles, consider $c_i^Y$ to be a line ($\codim c_i^Y =2$).
Then
\begin{equation}
    \GLSM(\CP^1,\CP^3,\mathcal D) = \int_{\CP^7 \times \CP^1_{(4)}} 
    H^2 H^2 H^2 (H+p_4^* h)^2
    = 2.
\end{equation}
Again, we let $c_i^X = \{0,1,\infty\}$ for $i=1,2,3$ and we consider the lines 
\begin{align*}
    c_1^Y &= \{y^1 = y^2 = 0\} ,\\
    c_2^Y &= \{y^1 - y^3 = y^2 - y^4 = 0\} ,\\
    c_3^Y &= \{y^3 = y^4 = 0\}.
\end{align*} 
These lines are generic in the sense that they have trivial pairwise intersections. 
Such quasimaps are of the form 
$$ \underline{f}(x^0:x^1) = (ax^0,bx^0,ax^1,bx^1
) $$ 
and such quasimaps have no freckle. Taking a fourth line in general position with respect to the first three (i.e. it does not intersect any of them) yields a quadratic equation for the quasimap parameter which then determines the running point. 

\subsection{Higher-dimensional source}
Next, we consider some examples with source $\CP^2$. 
\subsubsection{\texorpdfstring{$1+0=1$}{1+0=1}} 
As an easy example, consider the case $k=n=2$, $d=1$, i.e. degree 1 quasimaps $\CP^2 \qra \CP^2$. Here, we have $\dim \QMap_1(\CP^2,\CP^2) = 8$. The simplest configuration is $l = 4$ with source and target cycles being 4 fixed points. The $\GLSM$ number is 1, and indeed there is a unique degree 1 holomorphic map $f \colon \CP^2 \to \CP^2$ mapping four fixed points in $\CP^2$ to other four fixed points in $\CP^2$.\footnote{It is important that both quadruples are \emph{in general positions}. Otherwise if, e.g., one quadruple has a collinear triple in it and the other does not, such $f$ does not exist.} 
\subsubsection{\texorpdfstring{$1+2=3$}{1+2=3}}\label{sssec: 1+2 =3}
We consider degree 1 quasimaps from $X = \CP^2$ to $Y = \CP^3$.
Then
\begin{equation}
    \QMap_1(\CP^2,\CP^3) = \CP^{11}.
\end{equation}
We consider five cycles $c_i^X$ in the source and five cycles $c_i^Y$ in the target. 
Consider the following quasimap data $\mathcal D$: Out of the five cycles in the source we consider four to be fixed and one to be running, i.e.\ $c_i^X = {\rm p_i}$ for $i = 1, \dots, 4$ and $c_5^X = \CP^2$.
Moreover, let three out of the five cycles in the target be points and the remaining two be lines, i.e.\ $c_j^Y = {\rm point}$ for $i = 1,2,3$ and $c_j^Y = \ell_j$ for $j = 4,5$.

Let us first consider quasimaps that pass through a line, i.e.\ that map the running point to a line in the target.
The situation is depicted in Figure \ref{fig:P2 to P3 1 = 1 + 0} 
\begin{figure}[H]
    \centering\hfill
    \begin{subfigure}[t]{.4\textwidth}
    \centering
    \begin{tikzpicture}[baseline={([yshift=-.5ex]current bounding box.center)},scale=.6]
    \draw (0,0) circle (3);

    \coordinate[point] (c3) at (252:2);
    \coordinate[point] (c2) at (180:2);
    \coordinate[point] (c1) at (108:2);
    \coordinate[point] (c4) at (-36:2);

    \coordinate[point] (x) at (36:2);

    \begin{scope}[every node/.style={transform shape}]
        \draw (x) circle (.5);

        \draw pic at (x) {H={45}{.5}};
        \draw pic at (x) {H={-45}{.5}};

        \draw pic at (c1) {H={60}{.5}};
        \draw pic at (c1) {H={-60}{.5}};
        \draw pic at (c1) {H={0}{.5}};
        
        \draw pic at (c2) {H={60}{.5}};
        \draw pic at (c2) {H={-60}{.5}};
        \draw pic at (c2) {H={0}{.5}};

        \draw pic at (c3) {H={60}{.5}};
        \draw pic at (c3) {H={-60}{.5}};
        \draw pic at (c3) {H={0}{.5}};

        \draw pic at (c4) {H={45}{.5}};
        \draw pic at (c4) {H={-45}{.5}};
    \end{scope}

    \node[below right] at (c3) {$c_3^X$};
    \node[below,yshift=-.5em] at (c2) {$c_2^X$};
    \node[below,yshift=-.5em] at (c1) {$c_1^X$};
    \node[below] at ($(c4)+(0,-.2)$) {$c_4^X$};
    \node[below] at ($(x)+(0,-.5)$) {$c_5^X$};
    
  \end{tikzpicture} 
    \caption{Quasimap counting data allowing a unique holomorphic solution.}
    \label{fig:P2 to P3 1 = 1 + 0}
    \end{subfigure}\hfill
    \begin{subfigure}[t]{.4\textwidth}
    \centering
    \begin{tikzpicture}[baseline={([yshift=-.5ex]current bounding box.center)},scale=.6]
    \draw (0,0) circle (3);

    \coordinate[point] (c3) at (252:2);
    \coordinate[point] (c2) at (180:2);
    \coordinate[point] (c1) at (108:2);
    \coordinate[point] (c4) at (-36:2);

    \coordinate[point] (x) at (36:2);

    \begin{scope}[every node/.style={transform shape}]
        \draw (x) circle (.5);

        \draw pic at (x) {H={60}{.5}};
        \draw pic at (x) {H={-60}{.5}};
        \draw pic at (x) {H={0}{.5}};

        \draw pic at (c1) {H={60}{.5}};
        \draw pic at (c1) {H={-60}{.5}};
        \draw pic at (c1) {H={0}{.5}};
        
        \draw pic at (c2) {H={60}{.5}};
        \draw pic at (c2) {H={-60}{.5}};
        \draw pic at (c2) {H={0}{.5}};

        \draw pic at (c3) {H={45}{.5}};
        \draw pic at (c3) {H={-45}{.5}};

        \draw pic at (c4) {H={45}{.5}};
        \draw pic at (c4) {H={-45}{.5}};
    \end{scope}

    \node[below right] at (c3) {$c_3^X$};
    \node[below,yshift=-.5em] at (c2) {$c_2^X$};
    \node[below,yshift=-.5em] at (c1) {$c_1^X$};
    \node[below] at ($(c4)+(0,-.2)$) {$c_4^X$};
    \node[below] at ($(x)+(0,-.5)$) {$c_5^X$};
    
  \end{tikzpicture} 
    \caption{Quasimap counting data allowing proper quasimap solutions.}
    \label{fig:P2 to P3 3 = 1 + 2}
    \end{subfigure}\hfill
    \caption{Quasiplanes in $\CP^3$ passing through three points and two lines.}
\end{figure}
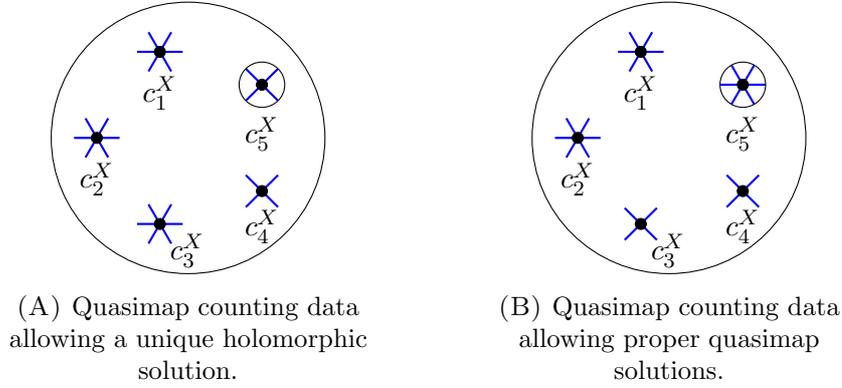

The $\GLSM$ number of this problem is given by 
\begin{equation}
    \GLSM(\CP^2,\CP^3,\mathcal D) = \int_{\CP^{11} \times \CP^2_{(5)}} H^3 H^3 H^3 H^2 (H + p_5^* h)^2 = 1.
\end{equation}
 This corresponds to a unique holomorphic map (i.e. $\KM=1$). Its image is the unique plane in the target through three points $\mr{p}_1,\mr{p}_2,\mr{p}_3$. It automatically passes through the lines $\ell_4,\ell_5$. The condition that the four fixed points on the source have prescribed images fixes the parametrization (i.e. gauge-fixes the source automorphism group $\PSL(3,\CC)$).

Now, let us consider a slightly modified problem, where the running point is mapped not to a line but to a point, cf.\ Figure \ref{fig:P2 to P3 3 = 1 + 2}.
In this case the $\GLSM$ number is 
\begin{equation}
    \GLSM = \int_{\CP^{11} \times \CP^2_{(5)}} H^3 H^3 H^2 H^2 (H + p_5^* h)^3 = 3.
\end{equation}
However, the geometric situation has not changed, i.e.\ there still exists a unique map solving the constraints. 
Therefore we expect that the proper quasimap contribution to the problem is 2. 

To compute the proper quasimap contribution, recall from Section \ref{sec:freckle_stratification} that the one-freckle stratum has complex codimension 
\begin{equation}
    \codim\QMap^1_1(\CP^2,\CP^3) = 3 + 1 - 2 = 2,
\end{equation}
that is the one-freckle stratum is a complex nine dimensional subspace in $\QMap$.
Consider the situation where the freckle and the running point collide with one of the two cycles which get mapped to a point, cf.\ Figure \ref{fig:one_freckle_stratum_P2_P3}.
 \begin{figure}[H]
    \centering
    \begin{tikzpicture}[baseline={([yshift=-.5ex]current bounding box.center)},scale=.8]
    \draw (0,0) circle (3);

    \coordinate[point] (c3) at (252:2);
    \coordinate[point] (c2) at (180:2);
    \coordinate[point] (c1) at (108:2);
    \coordinate[point] (c4) at (-36:2);

    \coordinate[point] (x) at (36:2);

    \begin{scope}[every node/.style={transform shape}]
        \draw (x) circle (.5);

        \draw pic at (x) {H={60}{.5}};
        \draw pic at (x) {H={-60}{.5}};
        \draw pic at (x) {H={0}{.5}};

        \draw pic at (c1) {H={60}{.5}};
        \draw pic at (c1) {H={-60}{.5}};
        \draw pic at (c1) {H={0}{.5}};
        
        \draw pic at (c2) {H={60}{.5}};
        \draw pic at (c2) {H={-60}{.5}};
        \draw pic at (c2) {H={0}{.5}};

        \draw pic at (c3) {H={45}{.5}};
        \draw pic at (c3) {H={-45}{.5}};

        \draw pic at (c4) {H={45}{.5}};
        \draw pic at (c4) {H={-45}{.5}};
    \end{scope}

    \node[below right] at (c3) {$c_3^X$};
    \node[below,yshift=-.5em] at (c2) {$c_2^X$};
    \node[left,xshift=-.5em] at (c1) {$c_1^X$};
    \node[below] at ($(c4)+(0,-.2)$) {$c_4^X$};
    \node[below] at ($(x)+(0,-.5)$) {$c_5^X$};

    \node[red] (fr) at (0,0) {$*$};

    \node (xx) at (x) {$\phantom{*}$}; 
    \draw[thick,red,->] (xx) to[bend right] ($(c1) + (0:.6)$);
    \draw[thick,red,->] (fr) to ($(c1) + (-90:.3)$);
  \end{tikzpicture} 
    \caption{One of the two 1-freckle configurations}
    \label{fig:one_freckle_stratum_P2_P3}
 \end{figure}
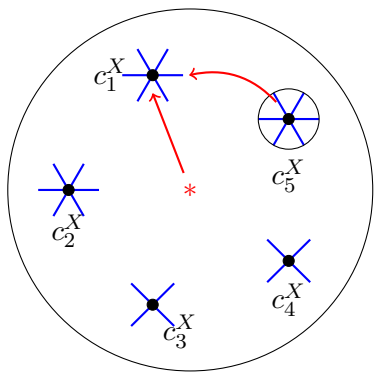
Such a configuration has codimension 2 in the freckle stratum (because we fix the position of the freckle in $\CP^2$) and therefore dimension 7. This is precisely the number of remaining equations, so generically there is a unique quasimap with a freckle at $c_1^X$ that satisfies all equations. Similarly, there is such a quasimap with a freckle at $c_2^X$. However, when the freckle sits at $c_3^X$ or $c_4^X$ there are generically no solutions since we have 8 remaining equations. Therefore, the total $\PQM$ number is 2 as expected.

\subsection{Quasi-stable examples}

In this subsection we consider examples where the naive quasimap count is degenerate, i.e.\ even though we are considering a balanced configuration of cycles, there are positive-dimensional components in the space of quasimaps satisfying all the equations.

Intuitively, we encounter the following situation:  for each pair of source and target cycles $c_i^X,c_i^Y$ we have $\dim c_i^X$ moduli and $\codim c_i^Y$ equations, i.e. by adding such a pair we expect to reduce the dimension of the space of quasimaps by $d_i = \codim c_i^Y -\dim c_i^X$. For \emph{proper} quasimaps, however, the equations become void when the running points on $c_i^X$ hit the freckle.
If in this way, we lose more $d_i$ than the codimension of the corresponding freckle stratum, we obtain a family of such quasimaps. For example, for $k=1,n=2$, the freckle stratum has codimension 2 (common zero locus of three polynomials). A running point (sent to a point) has $d=1$, so if we have 3 running points colliding in a freckle, we obtain a family of positive dimension. 

\subsubsection{A family of degenerate examples: \texorpdfstring{$1+(N-1)= N$}{1+(N-1)=1} }\label{sec: 1 + N-1 = N}
Consider the case $k=1,n=N,d=1$ and $l=4$, i.e. degree 1 quasimaps $\CP^1 \qra \CP^N$, with $l=4$ given cycles in the source and the target. 
We have $\dim\QMap_1(\CP^1,\CP^N)=2N+1$.
The following is a balanced configuration: We fix three points $c_1^X,c_2^X, c_3^X$ and keep one running point $c_4^X$ in the source, while we take $c_1^Y,c_2^Y$ to be hyperplanes in $\CP^N$ and $c_3^Y,c_4^Y$ fixed points, as depicted in Figure \ref{fig: 1 + (N-1) = N}.
\begin{figure}[H]
    \centering
    \begin{tikzpicture}[baseline={([yshift=-.5ex]current bounding box.center)},scale=.8]
    \draw (0,0) circle (3.3);

    \coordinate[point] (c3) at (90:2);
    \coordinate[point] (c2) at (180:2);
    \coordinate[point] (c1) at (270:1.5);

    \coordinate[point] (x) at (0:2);

    \begin{scope}[every node/.style={transform shape}]
        \draw (x) circle (.5);

        \draw pic at (x) {H={60}{.5}};
        \draw pic at (x) {H={-60}{.5}};
        \draw pic at (x) {H={0}{.5}};

        \draw pic at (c3) {H={60}{.5}};
        \draw pic at (c3) {H={-60}{.5}};
        \draw pic at (c3) {H={0}{.5}};

        \draw pic at (c2) {H={60}{.5}};

        \draw pic at (c1) {H={60}{.5}};
    \end{scope}

    \node[left] at ($(c1)+(-.2,0)$) {$c_1^X$};
    \node[left] at ($(c2)+(-.2,0)$) {$c_2^X$};
    \node[left] at ($(c3)+(-.5,0)$) {$c_3^X$};
    \node[below] at ($(x) + (0,-.5)$) {$c_4^X$};
    
    \draw[dotted,thick] ($(c3) + (45:0.6)$) arc (45:0:0.6cm);
    \node at ($(c3) + (30:0.9)$) {$N$};
    \draw[dotted,thick] ($(x) + (45:0.6)$) arc (45:0:0.6cm);
    \node at ($(x) + (30:0.9)$) {$N$};
    
  \end{tikzpicture} 
\caption{Pictorial representation of the enumerative problem. Here $N$ lines run through $c_3^X$ and $c_4^X$ each representing a hyperplane in $\CP^N$.}
\label{fig: 1 + (N-1) = N}
\end{figure}
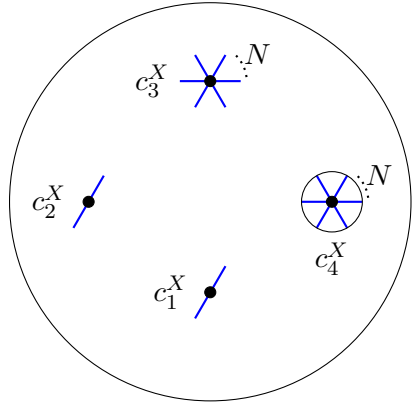

Then the $\GLSM$ number is 
\begin{equation}
    \GLSM(\CP^1,\CP^N,\mathcal D) = \int_{\CP^{2N + 1} \times \CP^1_{(4)}}  
    H\,H\, H^N (H+p_4^* h)^N
    = N.
\end{equation}
However, there is a unique holomorphic map, whose image is the unique line through $c_3^Y$ and $c_4^Y$. 
Therefore we expect that proper quasimaps contribute to the $\GLSM$ number with $N-1$.
However, there exist a family of proper quasimaps: consider the situation where the freckle and the running point collide with $c_3^X$, cf.\ Figure \ref{fig: quasimap 1 + (N-1) = N}.
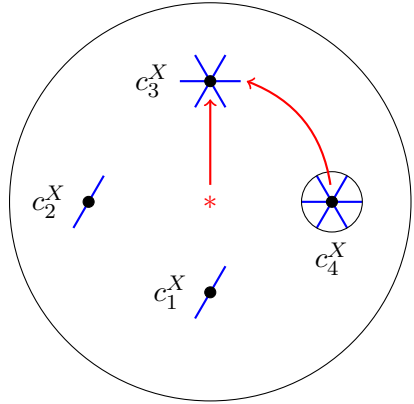
\begin{figure}[H]
    \centering
      \scalebox{1}{
    \begin{tikzpicture}[baseline={([yshift=-.5ex]current bounding box.center)},scale=.8]
    \draw (0,0) circle (3.3);

    \coordinate[point] (c3) at (90:2);
    \coordinate[point] (c2) at (180:2);
    \coordinate[point] (c1) at (270:1.5);

    \coordinate[point] (x) at (0:2);

    \begin{scope}[every node/.style={transform shape}]
        \draw (x) circle (.5);

        \draw pic at (x) {H={60}{.5}};
        \draw pic at (x) {H={-60}{.5}};
        \draw pic at (x) {H={0}{.5}};

        \draw pic at (c3) {H={60}{.5}};
        \draw pic at (c3) {H={-60}{.5}};
        \draw pic at (c3) {H={0}{.5}};

        \draw pic at (c2) {H={60}{.5}};

        \draw pic at (c1) {H={60}{.5}};
    \end{scope}

    \node[left] at ($(c1)+(-.2,0)$) {$c_1^X$};
    \node[left] at ($(c2)+(-.2,0)$) {$c_2^X$};
    \node[left] at ($(c3)+(-.5,0)$) {$c_3^X$};
    \node[below] at ($(x) + (0,-.5)$) {$c_4^X$};

    \node[red] (fr) at (0,0) {$*$};

    \node (xx) at (x) {$\phantom{*}$}; 
    \draw[thick,red,->] (xx) to [bend right] ($(c3) + (0.6,0)$);
    \draw[thick,red,->] (fr) to($(c3) + (0,-0.3)$); 
    
  \end{tikzpicture} 
}
\caption{The unique quasimap stratum is where the freckle and the running point sit at $c_3^X$.}
\label{fig: quasimap 1 + (N-1) = N}
\end{figure}

On the complement of the freckle, the quasimap defines a degree zero, hence a constant map. 
This map is constraint to map $c_1^X$ and $c_2^X$ to the hyperplane $c_1^Y$ and $c_2^Y$ respectively. 
Since the map is constant, it maps both $c_1^X$ and $c_2^X$ to the intersection $c_1^Y \cap c_2^Y$. 
Since two hyperplanes meet in a cycle 
of codimension 2, i.e.\ of dimension $N-2$, if $N > 2$, there exists a moduli for the map and hence for the proper quasimaps. \\

In fact this modulus is a $\CP^{N-2} \subset \CP^{2N+1}$, which is the unique connected component of $$\QMap^\pr(X,Y,\mathcal{D}) = Z = \CP^{N-2} \times \{c_3^X\} \subset \CP^{2N+1} \times \CP^1_{(4)}. $$ 
\begin{remark}
To be fully precise one has 
\[
Z = \CP^{N-2} \times \{c_1^X\}\times \{c_2^X\}\times \{c_3^X\}\times \{c_3^X\} \subset \Var.
\]
Note that any vector bundle over a point is trivial and therefore has total Chern class equal to $1$.
Points in $Z$ therefore contribute trivially, i.e.\ by a factor of 1 at the appropriate place, to $c(B)$, cf.\ \eqref{eq:BZ}. 
Hence, fixing $\{c_i^X\}_{i=1}^3$ to be points, effectively reduces the problem over $\Var$ to a problem over $\CP^{N-2} 
$.
\end{remark}
We can compute its $\PQM$ number by first computing the Chern class of the excess bundle using \eqref{eq:BZ}. 
Let $\zeta$ denote the generator of $H^2(Z) = H^2(\CP^{N-2})$. 
Then
\begin{equation}
\begin{split}
c(B) &= \frac{\Big(1 + p^*_{\CP^{2N+1}}c_1(\OO(1)_\QMap)\Big)^{2N+2}\Big|_Z c(\CP^{N-2})}{c(\CP^{2N + 1} \times \CP^1_{(4)})|_Z}\\
&= \frac{(1+H)^{2N+2}|_Z (1+\zeta)^{N-1}}{(1+H)^{2N+2}|_Z} \\
&= \frac{(1+\zeta)^{2N+2}(1+\zeta)^{N-1}}{(1+\zeta)^{2N+2}} = (N-1)\zeta^{N-2} + \ldots,
\end{split}
\end{equation}
where the dots denote terms of lower than top degree. From this, we get 
\begin{equation}
    \PQM(X,Y,\mathcal{D}) =\PQM(X,Y,\mathcal{D};Z) = \int_{\CP^{N-2}} (N-1)\zeta^{N-2} = N-1.
\end{equation}
\subsubsection{\texorpdfstring{$1+(K-1) = K$}{1+(K-1)=1}}
We can generalize this example by fixing only $K < N$ equations at the running point. In order to obtain a balanced problem, we then fix $N_i$ equations at $c_i^X$, with $N_1 + N_2 + N_3 +K = 2N +2$. See Figure \ref{fig: 1 + (K-1) = K}. 
\begin{figure}[H]
    \centering
    \begin{tikzpicture}[baseline={([yshift=-.5ex]current bounding box.center)},scale=.8]
    \draw (0,0) circle (3.3);

    \coordinate[point] (c3) at (90:2);
    \coordinate[point] (c2) at (180:2);
    \coordinate[point] (c1) at (270:1.5);

    \coordinate[point] (x) at (0:2);

    \begin{scope}[every node/.style={transform shape}]
        \draw (x) circle (.5);

        \draw pic at (x) {H={60}{.5}};
        \draw pic at (x) {H={-60}{.5}};
        \draw pic at (x) {H={0}{.5}};

        \draw pic at (c1) {H={60}{.5}};
        \draw pic at (c1) {H={-60}{.5}};
        \draw pic at (c1) {H={0}{.5}};

        \draw pic at (c2) {H={60}{.5}};
        \draw pic at (c2) {H={-60}{.5}};
        \draw pic at (c2) {H={0}{.5}};

        \draw pic at (c3) {H={60}{.5}};
        \draw pic at (c3) {H={-60}{.5}};
        \draw pic at (c3) {H={0}{.5}};
    \end{scope}

    \node[left] at ($(c1)+(-.2,0)$) {$c_1^X$};
    \node[left] at ($(c2)+(-.2,0)$) {$c_2^X$};
    \node[left] at ($(c3)+(-.5,0)$) {$c_3^X$};
    \node[below] at ($(x) + (0,-.5)$) {$c_4^X$};

    \draw[dotted,thick] ($(x) + (45:0.6)$) arc (45:0:0.6cm);
    \node at ($(x) + (30:0.9)$) {$K$};
    \draw[dotted,thick] ($(c1) + (45:0.6)$) arc (45:0:0.6cm);
    \node at ($(c1) + (30:0.9)$) {$N_1$};
    \draw[dotted,thick] ($(c2) + (45:0.6)$) arc (45:0:0.6cm);
    \node at ($(c2) + (30:0.9)$) {$N_2$};
    \draw[dotted,thick] ($(c3) + (45:0.6)$) arc (45:0:0.6cm);
    \node at ($(c3) + (30:0.9)$) {$N_3$};
    
  \end{tikzpicture} 
\caption{Pictorial representation of the enumerative problem. Here $N_i$ lines run through $c_i^X$ and $K$ lines through $c_4^X$ each representing a hyperplane in $\CP^N$.}
\label{fig: 1 + (K-1) = K}
\end{figure}
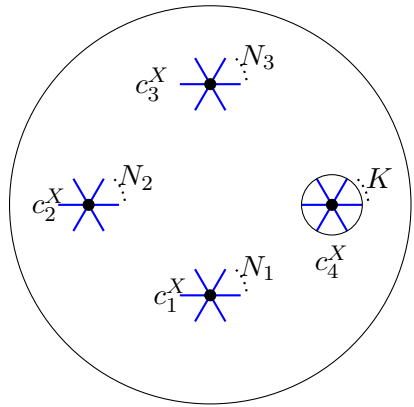
\begin{equation}\label{K}
    \GLSM(\CP^1,\CP^N,\mathcal D) = \int_{\CP^{2N + 1} \times \CP^1_{(4)}}  
    H^{N_1} H^{N_2} H^{N_3} (H+p_4^* h)^K
    = K.
\end{equation}
\begin{remark}
Surprisingly, the QM numbers violate the symmetry that is well-known in Gromov-Witten theory.
Observe that  the $\QM$ number (\ref{K}) depends on $K$.
Naively, this seems to indicate that the QM counting problem is \emph{not} invariant w.r.t. automorphisms of the source: naively, one can apply a M\"obius transformation which is constant at points $c_{1,2}^X$, stops  $c_4^X$ and makes it non-moving, and as a result makes $c_3^X$ a moving point -- but the corresponding QM number is $N_3$, not $K$! The problem with this argument is that it implicitly assumes genericity of the configuration, in particular that the moving point doesn't collide with any of the stationary points (otherwise, one cannot separate points in this non-generic configuration by a M\"obius transformation). On the other hand there are freckle contributions exactly from configurations when the moving point and the freckle collide with one of the stationary points. 


For example, consider 
a special case of the example above, with $N=2$,
cf.\ Figure \ref{fig:Moebius_inequivalent_situations}.
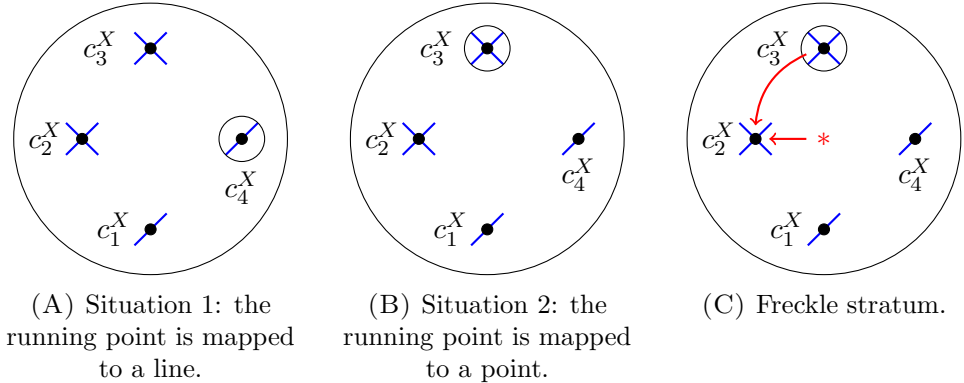
\begin{figure}[H]
    \centering
    \begin{subfigure}[t]{.3\textwidth}
        \centering
        \begin{tikzpicture}[baseline={([yshift=-.5ex]current bounding box.center)},scale=.6]
    
            \draw (0,0) circle (3);
        
            \coordinate[point] (c3) at (90:2);
            \coordinate[point] (c2) at (180:1.5);
            \coordinate[point] (c1) at (270:2);
        
            \coordinate[point] (x) at (0:2);

            \begin{scope}[every node/.style={transform shape}]
                \draw (x) circle (.5);

                \draw pic at (x) {H={45}{.5}};

                \draw pic at (c1) {H={45}{.5}};

                \draw pic at (c2) {H={45}{.5}};
                \draw pic at (c2) {H={-45}{.5}};
                \draw pic at (c3) {H={45}{.5}};
                \draw pic at (c3) {H={-45}{.5}};
                
            \end{scope}
            
            \node[left] at ($(c1)+(-.2,0)$) {$c_1^X$};
            \node[left] at ($(c2)+(-.2,0)$) {$c_2^X$};
            \node[left] at ($(c3)+(-.5,0)$) {$c_3^X$};
            \node[below] at ($(x)+(0,-.5)$) {$c_4^X$};
            
        \end{tikzpicture}
        \caption{Situation 1: the running point is mapped to a line.}
        \label{fig:Moebius_inequivalent_situation1}
    \end{subfigure}\hfill
    \begin{subfigure}[t]{.3\textwidth}
        \centering
        \begin{tikzpicture}[baseline={([yshift=-.5ex]current bounding box.center)},scale=.6]
    
            \draw (0,0) circle (3);
        
            \coordinate[point] (c3) at (90:2);
            \coordinate[point] (c2) at (180:1.5);
            \coordinate[point] (c1) at (270:2);
        
            \coordinate[point] (x) at (0:2);

            \begin{scope}[every node/.style={transform shape}]
                \draw (c3) circle (.5);

                \draw pic at (x) {H={45}{.5}};

                \draw pic at (c1) {H={45}{.5}};

                \draw pic at (c2) {H={45}{.5}};
                \draw pic at (c2) {H={-45}{.5}};
                \draw pic at (c3) {H={45}{.5}};
                \draw pic at (c3) {H={-45}{.5}};
                
            \end{scope}
            
            \node[left] at ($(c1)+(-.2,0)$) {$c_1^X$};
            \node[left] at ($(c2)+(-.2,0)$) {$c_2^X$};
            \node[left] at ($(c3)+(-.5,0)$) {$c_3^X$};
            \node[below] at ($(x)+(0,-.2)$) {$c_4^X$};
            
       \end{tikzpicture}
       \caption{Situation 2: the running point is mapped to a point.}
       \label{fig:Moebius_inequivalent_situation2}
    \end{subfigure}\hfill
    \begin{subfigure}[t]{.3\textwidth}
        \centering
        \begin{tikzpicture}[baseline={([yshift=-.5ex]current bounding box.center)},scale=.6]
    
            \draw (0,0) circle (3);
        
            \coordinate[point] (c3) at (90:2);
            \coordinate[point] (c2) at (180:1.5);
            \coordinate[point] (c1) at (270:2);
        
            \coordinate[point] (x) at (0:2);

            \begin{scope}[every node/.style={transform shape}]
                \draw (c3) circle (.5);

                \draw pic at (x) {H={45}{.5}};

                \draw pic at (c1) {H={45}{.5}};

                \draw pic at (c2) {H={45}{.5}};
                \draw pic at (c2) {H={-45}{.5}};
                \draw pic at (c3) {H={45}{.5}};
                \draw pic at (c3) {H={-45}{.5}};
                
            \end{scope}
            
            \node[left] at ($(c1)+(-.2,0)$) {$c_1^X$};
            \node[left] at ($(c2)+(-.2,0)$) {$c_2^X$};
            \node[left] at ($(c3)+(-.5,0)$) {$c_3^X$};
            \node[below] at ($(x)+(0,-.2)$) {$c_4^X$};

            \node[red] (fr) at (0,0) {$*$};
            \node (cc3) at (c3) {$\phantom{*}$};

            \draw[red,thick,->] (fr) to ($(c2) + (0:.3)$);
            \draw[red,thick,->] (cc3) to[bend right] ($(c2) + (90:.3)$);
            
       \end{tikzpicture}
       \caption{Freckle stratum.}
       \label{fig:Moebius_inequivalent_situation_freckle}
    \end{subfigure}
    \caption{An apparent ``contradiction'' to M\"obius invariance of $\QM$ numbers.}
    \label{fig:Moebius_inequivalent_situations}
\end{figure}
The $\QM$ numbers of the configurations depicted in Figure \ref{fig:Moebius_inequivalent_situation1} and Figure \ref{fig:Moebius_inequivalent_situation2} are $1$ and $2$ respectively, while $\KM=1$ in both cases (a single line through two points in $\CP^2$). 
In Figure \ref{fig:Moebius_inequivalent_situation_freckle} we show the degenerate configuration of Figure \ref{fig:Moebius_inequivalent_situation2} -- a proper quasimap, which contributes the additional $1$ to the $\QM$ number. 
Importantly, the running point (together with the freckle) collapses with $c_2^X$. 
This is clearly a non-generic situation and hence cannot be reached from a generic situation, such as shown in Figure \ref{fig:Moebius_inequivalent_situation1}, by action of a M\"obius transformation.
\end{remark}

Now, let us first assume that $N_2 + N_3  > N$. Then there are two freckle strata $Z_i$, where the freckle and the running point sit at $c_i^X$ for $i=2,3$, and its dimension is $d_i = N - N_j - N_k$, with $\{i,j,k\} = \{1,2,3\}$ (the freckle at $c_1^X$ is prohibited by $N_2 + N_3 > N$). The 
contribution of such a stratum to $\PQM$ is 
\begin{equation}
    \begin{split}
        \int_{Z_i}c(B_{Z_i}) =\int_{Z_i} \frac{(1 + \zeta_i)^{2N+2}(1+\zeta_i)^{d_i+1}}{(1+\zeta_i)^{2N+2}}= d_i + 1,
    \end{split}
\end{equation}
where $\zeta_i$ is the generator of $H^2(Z_i)=H^2(\CP^{d_i})$.
Therefore, we obtain the total $\PQM$ contribution 
$$d_2 + d_3 + 2 = 2N - 2N_1 - N_2 - N_3 + 2 = 2N +1 - (N_1 - N_2 - N_3) - N_1  = K - N_1.$$ 
Here we have used that $N_1 + N_2 + N_3 = 2N + 2 - K$. In particular, the KM number is $N_1$. 

\subsubsection{
Conic through five points: \texorpdfstring{$1 + 6 + 9 = 16$}{1+6+9=16}}\label{sec:1+3+12=16} 
In this example we will recover the fact that there is a unique conic through five points in general position, using quasimap counting machinery.

We set
$k = 1$, $n = 2$ and $d=2$. Remember that $\dim \QMap_2(\CP^1,\CP^2) = 8$, so we can fix $l=5$,  with source cycles $c_i^X$ points for $i=1,2,3$ and $\CP^1$ otherwise, and $c_i^Y$ a collection of points (see Figure \ref{fig: 1 + 3  + 12 = 16 1}). The $\GLSM$ number is then 
\begin{equation}
    \GLSM(\CP^1,\CP^2,\mathcal{D}) = \int_{\CP^8\times \CP^1_{(4)}\times \CP^1_{(5)}} 
    (H^2)^3 \wedge_{i=4}^5(H+2\, p_i^* h)^2 
    = 16.
\end{equation}
Let us fix again $c_i^X = \{0,1,\infty\}$  as before, and take $c_1^Y = (0:0:1), c_2^Y = (0:1:0),c_3^Y = (1:0:0)$. Such degree 2 quasimaps can be parametrized by $(a:b:c) \in \CP^2$ by 
\begin{equation} \underline{f}(x^0:x^1) = (a((x^0)^2 - x^0x^1): bx^0x^1: c(x^0x^1-(x^1)^2)).
\label{eq: P1 -> P2 d=2 qmap}
\end{equation}
We notice that $\underline{f}$ has $k$ freckles if $k$ of the parameters $a,b,c$ are zero. That is, for $a,b,c$ nonzero we obtain a degree 2 holomorphic map $f\colon\CP^1 \to \CP^2$, which is uniquely fixed by the requirement that it passes through the five points $c_i^Y$. If one out of the three parameters is zero, then $\underline{f}$ has a freckle at one of the $c_i^X$, $i=1,2,3$, and both running points sit at this freckle - there are three such configurations, depicted in Figure \ref{fig: 1 + 3 + 12 = 16 2}. For each of these there is actually a $\CP^1$ of solutions $Z_i$, since there are four remaining equations, but the dimension of the space of degree 1 quasimaps is 5. However, if two of the parameters vanish, then $\underline{f}$ has two freckles at two of the $c_i^X$, $i=1,2,3$. There are 3 possible configurations for the 2-freckle locus, however, we have to take into account the 2 running points: They can sit at the two 2 freckles in any configuration, however, if they go to the same freckle, sitting at $c_i^X$, then that is actually part of the stratum $Z_i$.
See Figure \ref{fig: 1 + 3 + 12 = 16 3}.
 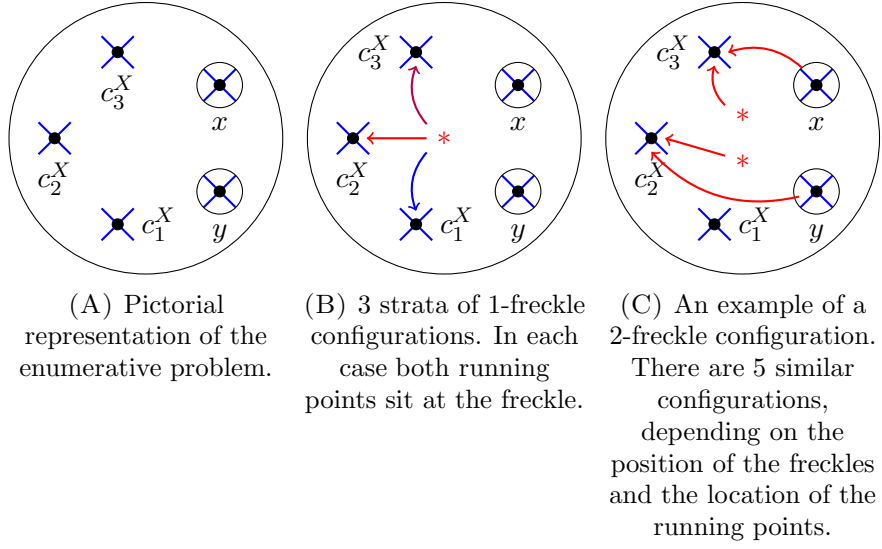
\begin{figure}[H]
    \centering 
\begin{subfigure}[t]{0.3\textwidth}
    \centering
    \begin{tikzpicture}[baseline={([yshift=-.5ex]current bounding box.center)},scale=.6]
    \draw (0,0) circle (3);

    \coordinate[point] (c3) at (108:2);
    \coordinate[point] (c2) at (180:2);
    \coordinate[point] (c1) at (252:2);

    \coordinate[point] (x) at (36:2);
    \coordinate[point] (y) at (-36:2);

    \begin{scope}[every node/.style={transform shape}]
        \draw (x) circle (.5);
        \draw (y) circle (.5);

        \draw pic at (x) {H={45}{.5}};
        \draw pic at (x) {H={-45}{.5}};

        \draw pic at (y) {H={45}{.5}};
        \draw pic at (y) {H={-45}{.5}};

        \draw pic at (c1) {H={45}{.5}};
        \draw pic at (c1) {H={-45}{.5}};

        \draw pic at (c2) {H={45}{.5}};
        \draw pic at (c2) {H={-45}{.5}};

        \draw pic at (c3) {H={45}{.5}};
        \draw pic at (c3) {H={-45}{.5}};
    \end{scope}

    \node[right] at ($(c1)+ (.3,0)$) {$c_1^X$};
    \node[below] at ($(c2)+(0,-.3)$) {$c_2^X$};
    \node[below] at ($(c3)+(0,-.3)$) {$c_3^X$};

    \node[below] at ($(x)+(0,-.5)$) {$x$};
    \node[below] at ($(y)+(0,-.5)$) {$y$};
  \end{tikzpicture} 
\caption{Pictorial representation of the enumerative problem. }
\label{fig: 1 + 3  + 12 = 16 1}

    \end{subfigure}
    \begin{subfigure}[t]{0.3\textwidth}
    \centering
    \begin{tikzpicture}[baseline={([yshift=-.5ex]current bounding box.center)},scale=.6]
    \draw (0,0) circle (3);

    \coordinate[point] (c3) at (108:2);
    \coordinate[point] (c2) at (180:2);
    \coordinate[point] (c1) at (252:2);

    \coordinate[point] (x) at (36:2);
    \coordinate[point] (y) at (-36:2);

    \begin{scope}[every node/.style={transform shape}]
        \draw (x) circle (.5);
        \draw (y) circle (.5);

        \draw pic at (x) {H={45}{.5}};
        \draw pic at (x) {H={-45}{.5}};

        \draw pic at (y) {H={45}{.5}};
        \draw pic at (y) {H={-45}{.5}};

        \draw pic at (c1) {H={45}{.5}};
        \draw pic at (c1) {H={-45}{.5}};

        \draw pic at (c2) {H={45}{.5}};
        \draw pic at (c2) {H={-45}{.5}};

        \draw pic at (c3) {H={45}{.5}};
        \draw pic at (c3) {H={-45}{.5}};
    \end{scope}

    \node[right] at ($(c1)+ (.3,0)$) {$c_1^X$};
    \node[below] at ($(c2)+(0,-.3)$) {$c_2^X$};
    \node[left] at ($(c3)+(-.3,0)$) {$c_3^X$};

    \node[below] at ($(x)+(0,-.5)$) {$x$};
    \node[below] at ($(y)+(0,-.5)$) {$y$};

    \node[red] (fr) at (0,0) {$*$};

    \draw[thick,blue,->] (fr) to[bend right] ($(c1) + (90:0.3)$);
    \draw[thick,red,->] (fr) to ($(c2) + (0:0.3)$);
    \draw[thick,purple,->] (fr) to[bend left] ($(c3) + (-90:.3)$);
    \end{tikzpicture} 
\caption{3 strata of 1-freckle configurations. 
In each case both running points sit at the freckle.}
\label{fig: 1 + 3 + 12 = 16 2}

    \end{subfigure}
        \begin{subfigure}[t]{0.3\textwidth}
    \centering
\begin{tikzpicture}[baseline={([yshift=-.5ex]current bounding box.center)},scale=.6]
    \draw (0,0) circle (3);

    \coordinate[point] (c3) at (108:2);
    \coordinate[point] (c2) at (180:2);
    \coordinate[point] (c1) at (252:2);

    \coordinate[point] (x) at (36:2);
    \coordinate[point] (y) at (-36:2);

    \begin{scope}[every node/.style={transform shape}]
        \draw (x) circle (.5);
        \draw (y) circle (.5);

        \draw pic at (x) {H={45}{.5}};
        \draw pic at (x) {H={-45}{.5}};

        \draw pic at (y) {H={45}{.5}};
        \draw pic at (y) {H={-45}{.5}};

        \draw pic at (c1) {H={45}{.5}};
        \draw pic at (c1) {H={-45}{.5}};

        \draw pic at (c2) {H={45}{.5}};
        \draw pic at (c2) {H={-45}{.5}};

        \draw pic at (c3) {H={45}{.5}};
        \draw pic at (c3) {H={-45}{.5}};
    \end{scope}

    \node[right] at ($(c1)+ (.3,0)$) {$c_1^X$};
    \node[below] at ($(c2)+(0,-.3)$) {$c_2^X$};
    \node[left] at ($(c3)+(-.3,0)$) {$c_3^X$};

    \node[below] at ($(x)+(0,-.5)$) {$x$};
    \node[below] at ($(y)+(0,-.5)$) {$y$};

    \node[red] (fr1) at (0,-.5) {$*$};
    \node[red] (fr2) at (0,.5) {$*$};

    \node (xx) at (x) {$\phantom{*}$};
    \node (yy) at (y) {$\phantom{*}$};
    
    \draw[thick,red,->] (fr1) to ($(c2) + (0:0.3)$);
    \draw[thick,red,->] (yy) to[bend left] ($(c2) + (-90:0.3)$);
    
    \draw[thick,red,->] (fr2) to[bend left] ($(c3) + (-90:.3)$);
    \draw[thick,red,->] (xx) to[bend right] ($(c3) + (0:.3)$);
    \end{tikzpicture} 
\caption{An example of a 2-freckle configuration. There are 5 similar configurations, depending on the position of the freckles and the location of the running points.}
\label{fig: 1 + 3 + 12 = 16 3}

    \end{subfigure}
    \caption{Quasimap count for degree 2 quasimaps $\CP^1 \qra \CP^2$.}
    \label{fig: 1 + 3 + 12 = 16}

    \end{figure}
    
The set of proper quasimaps, $\QMap(X,Y,\mathcal{D})^{\pr}$, admits the following stratification:
\begin{equation}
\QMap(X,Y,\mathcal{D})^{\pr} = \sqcup_{i \neq j =1}^3 Z_{ij} \sqcup_{i=1}^3 Z_i \subset \CP^8 \times \CP^1 \times \CP^1, 
\end{equation}
where the six components $Z_{ij}$ are given by
\begin{equation}
    Z_{ij} = \{ (\underline{f}_{(i,j)},c_i^X,c_j^X)  \}.
\end{equation}
Here $\underline{f}_{(i,j)}$ is the unique quasimap that has freckles at fixed points $c_i^X$ and $c_j^X$ and on the complement is given by $c_k^Y$. 
Since all those components are points, they contribute with $1$ to $\PQM(X,Y,\mathcal D)$ so that the total contribution of $\sqcup_{i\neq j = 1}^3 Z_{ij}$ is 6.

The components $Z_i$ are given by
\begin{equation}
    Z_{i} = \tilde Z_i \times \{ c_i^X \} \times \{ c_i^X \},
\end{equation}
where $\tilde{Z}_i\cong \CP^1$ is the stratum of quasimaps $\underline{f} = P\cdot \underline{f}_1$ with $P$ vanishing at $c_i^X$, and $\underline{f}_1$ the degree 1 quasimap satisfying the equations given at $c_j^X$ for $j \in \{1,2,3\} \setminus \{i\}$. 




Here $\tilde{Z}_i \cong \CP^1$ because there is a $\CP^5$ of degree 1 quasimaps $\underline{f}_1$ on which we impose 4 linear equations. We notice that this counting data is quasi-stable. The $\PQM$ number of the stratum $Z_i = \tilde{Z_i}\times \{c_i^X\} \times \{c_i^X\}$ is computed by the Chern class of the excess bundle according to \eqref{eq:BZ}: 
\begin{align*}
c\left(B_{Z_i}\right) &=  \frac{\big(1 + p^*_{\CP^8}c_1(\OO(1))\big)^{10}|_{Z_i}c(\CP^1)}{c(\CP^8)|_{Z_i}} \\
&=\frac{(1+H)^{10}|_{Z_i}(1+\zeta_i)^2}{(1+H)^9|_{Z_i}} = 1+ 3\zeta_i,
\end{align*}
where again $\zeta_i$ denotes the generator of $H^2(Z_i) = H^2(\CP^1)$. 
Hence, 
\[
\PQM(X,Y,\mathcal{D};\sqcup_{i=1}^3 Z_i) = \sum_{i=1}^3 \PQM(X,Y,\mathcal{D};Z_i) = \sum_{i=1}^3 \int_{Z_i} c\left( B_{Z_i} \right) = 3\cdot 3 = 9.
\]
In total we therefore have $\PQM(X,Y,\mathcal{D}) = 6 + 9 = 15$, and therefore we obtain that 
$$\KM(X,Y,\mathcal{D}) = \GLSM(X,Y,\mathcal{D}) - \PQM(X,Y,\mathcal{D})=1.$$
We remark that in $Z_i$, there are quasimaps with a second freckle (which has to be located at another fixed point $c_j^X$) but these contributions are not counted separately.

\subsection{Unstable examples: the need for Segre class computations.}
To compute the $\PQM$ number in non-quasistable examples (with the zero-locus $Z\subset \Var$ of the section $\sigma$ not given by a disjoint union of smooth submanifolds), one needs to compute Segre classes, see Section \ref{sec: unstable}. We defer such computations to a future paper and restrict ourselves to the analysis of the locus $Z$ in several examples. 

\subsubsection{The simplest non-quasistable example} \label{sss: Segre class computed}
Consider degree $1$ quasimaps $\CP^1\qra \CP^3$ with $l=5$ source/target cycles: 3 fixed points mapping to planes and two running points mapping to lines, see Figure \ref{fig:P1 to P3 non quasi stable}.
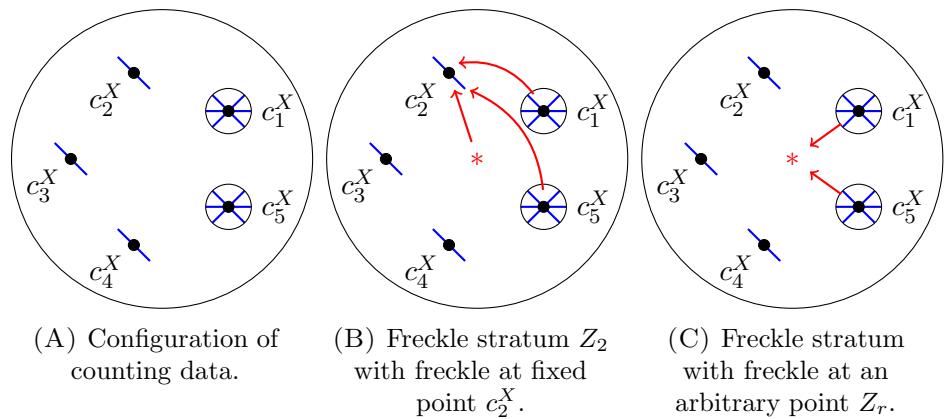
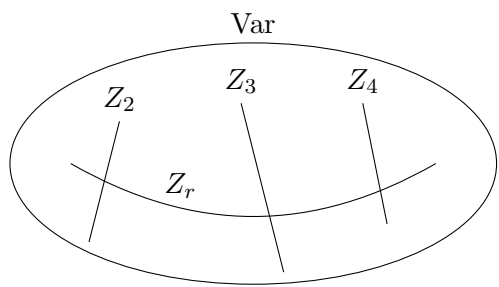
\begin{figure}[H]
    \centering\hfill
    \begin{subfigure}[t]{0.3\textwidth}
    \centering
      \begin{tikzpicture}[baseline={([yshift=-.5ex]current bounding box.center)},scale=.6]
    \draw (0,0) circle (3.3);
    \coordinate[point] (c2) at (108:2);
    \coordinate[point] (c3) at (180:2);
    \coordinate[point] (c4) at (252:2);

    \coordinate[point] (x) at (36:1.8);
    \coordinate[point] (y) at (-36:1.8);

    \begin{scope}[every node/.style={transform shape}]
        \draw (x) circle (.5);
        \draw (y) circle (.5);
        
        \draw pic at (x) {H={45}{.5}};
        \draw pic at (x) {H={-45}{.5}};
        \draw pic at (x) {H={0}{.5}};

        \draw pic at (y) {H={45}{.5}};
        \draw pic at (y) {H={-45}{.5}};
        \draw pic at (y) {H={0}{.5}};

        \draw pic at (c1) {H={-45}{.5}};
        
        \draw pic at (c2) {H={-45}{.5}};

        \draw pic at (c3) {H={-45}{.5}};
    \end{scope}

    \node[below left] at (c2) {$c_2^X$};
    \node[below left] at (c3) {$c_3^X$};
    \node[below left] at (c4) {$c_4^X$};

    \node[right] at ($(x)+(0.5,0)$) {$c_1^X$};
    \node[right] at ($(y)+(0.5,0)$) {$c_5^X$};

  \end{tikzpicture} 
    \caption{Configuration of counting data.}
    \label{fig:P1 to P3 non quasi stable conf}
    \end{subfigure}
    \hfill
    \begin{subfigure}[t]{0.3\textwidth}
    \centering
  \begin{tikzpicture}[baseline={([yshift=-.5ex]current bounding box.center)},scale=.6]
    \draw (0,0) circle (3.3);
    \coordinate[point] (c2) at (108:2);
    \coordinate[point] (c3) at (180:2);
    \coordinate[point] (c4) at (252:2);

    \coordinate[point] (x) at (36:1.8);
    \coordinate[point] (y) at (-36:1.8);

    \begin{scope}[every node/.style={transform shape}]
        \draw (x) circle (.5);
        \draw (y) circle (.5);
        
        \draw pic at (x) {H={45}{.5}};
        \draw pic at (x) {H={-45}{.5}};
        \draw pic at (x) {H={0}{.5}};

        \draw pic at (y) {H={45}{.5}};
        \draw pic at (y) {H={-45}{.5}};
        \draw pic at (y) {H={0}{.5}};

        \draw pic at (c1) {H={-45}{.5}};
        
        \draw pic at (c2) {H={-45}{.5}};

        \draw pic at (c3) {H={-45}{.5}};
    \end{scope}

    \node[below left] at (c2) {$c_2^X$};
    \node[below left] at (c3) {$c_3^X$};
    \node[below left] at (c4) {$c_4^X$};

    \node[right] at ($(x)+(0.5,0)$) {$c_1^X$};
    \node[right] at ($(y)+(0.5,0)$) {$c_5^X$};

    \node[red] (fr) at (0,0) {$*$};

    \node (xx) at (x) {$\phantom{*}$}; 
    \node (yy) at (y) {$\phantom{*}$};
    \node (cc2) at (c2) {$\phantom{*}$};

    \draw[thick,red,->] (xx) to[bend right] ($(cc2) + (45:.3)$);
    \draw[thick,red,->] (fr) to (cc2);
    \draw[thick,red,->] (yy) to[bend right] ($(cc2) + (-40:.6)$);

  \end{tikzpicture}
    \caption{Freckle stratum $Z_2$ with freckle at fixed point $c_2^X$.}
    \label{fig:P1 to P3 non quasi stable 1 fixed fr}
    \end{subfigure}
    \hfill
    \begin{subfigure}[t]{0.3\textwidth}
    \centering
    \begin{tikzpicture}[baseline={([yshift=-.5ex]current bounding box.center)},scale=.6]
    \draw (0,0) circle (3.3);
    \coordinate[point] (c2) at (108:2);
    \coordinate[point] (c3) at (180:2);
    \coordinate[point] (c4) at (252:2);

    \coordinate[point] (x) at (36:1.8);
    \coordinate[point] (y) at (-36:1.8);

    \begin{scope}[every node/.style={transform shape}]
        \draw (x) circle (.5);
        \draw (y) circle (.5);
        
        \draw pic at (x) {H={45}{.5}};
        \draw pic at (x) {H={-45}{.5}};
        \draw pic at (x) {H={0}{.5}};

        \draw pic at (y) {H={45}{.5}};
        \draw pic at (y) {H={-45}{.5}};
        \draw pic at (y) {H={0}{.5}};

        \draw pic at (c1) {H={-45}{.5}};
        
        \draw pic at (c2) {H={-45}{.5}};

        \draw pic at (c3) {H={-45}{.5}};
    \end{scope}

    \node[below left] at (c2) {$c_2^X$};
    \node[below left] at (c3) {$c_3^X$};
    \node[below left] at (c4) {$c_4^X$};

    \node[right] at ($(x)+(0.5,0)$) {$c_1^X$};
    \node[right] at ($(y)+(0.5,0)$) {$c_5^X$};

    \node[red] (fr) at (0,0) {$*$};

    \node (xx) at (x) {$\phantom{*}$}; 
    \node (yy) at (y) {$\phantom{*}$};

    \draw[thick,red,->] (xx) to (fr);
    \draw[thick,red,->] (yy) to (fr);

  \end{tikzpicture}
    \caption{Freckle stratum with freckle at an arbitrary point $Z_r$.}
    \label{fig:P1 to P3 non quasi stable 1 running fr}
    \end{subfigure}
    \hfill
    \vspace{1cm}
    \begin{subfigure}[t]{\textwidth}
    \centering
    \begin{tikzpicture}[scale = 0.8]
        \draw (0,0) ellipse (4cm and 2cm);
        \node[above] at (0,2) {$\Var$};
        \draw (-3,0) to[bend right] node[above, pos=0.3] {$Z_r$} (3,0);
        \draw (-2.7,-1.3) -- (-2.2,0.7);
        \node[above] at (-2.2,0.7) {$Z_2$};
        \draw (0.5,-1.8) -- (-0.2,1);
        \node[above] at (-0.2,1) {$Z_3$}; 
        \draw (2.2,-1) -- (1.8,1);
        \node[above] at (1.8,1) {$Z_4$};
    \end{tikzpicture}
    \caption{Schematic picture of the singular 1-freckle stratum $Z$}\label{fig:P1 to P3 non quasi stable Z.}
    \label{fig:P1 to P3 non quasi stable Z}
    \end{subfigure}
    \caption{An unstable quasimap counting data for $\QMap_1(\CP^1,\CP^3)$.}
    \label{fig:P1 to P3 non quasi stable}
\end{figure}
In this example  $\GLSM = 9$, $\KM =1$. One has a running freckle stratum $Z_r$ 
(Figure \ref{fig:P1 to P3 non quasi stable 1 running fr}) and 3 fixed freckle strata $Z_i$ (Figure \ref{fig:P1 to P3 non quasi stable 1 fixed fr}), each of them is a $\CP^1$. These strata intersect when the running freckle hits a fixed point, see Figure \ref{fig:P1 to P3 non quasi stable Z}, therefore $Z = Z_2 \cup Z_3 \cup Z_4 \cup Z_r$ is singular. 

Computing the contributions of the components independently, without taking intersections into account, we would have obtained
\begin{equation}\label{14 neq 8}
\underbrace{3+3+3}_{Z_2,Z_3,Z_4}+\underbrace{5}_{Z_r} \neq  \underbrace{8}_{\QM-\KM}.
\end{equation} 
This shows that due to the non-trivial intersection of the strata, one cannot treat the $Z_{(I)}$ independently and one is led to the problem to compute the Segre class of $Z = \bigcup_I Z_{(I)}$.

\subsubsection{An 
unstable example with a scar}\label{sss: unstable example with a scar}
Consider degree 1 quasimaps $\CP^2 \qra \CP^N$, with the same configuration as in Figure \ref{fig:P2 to P3 3 = 1 + 2}, but now we put $N$ equations at $c_1^X,c_2^X,c_5^X$. 
The total number of equations is then $3N + 4$, whereas $\dim \QMap_1(\CP^2,\CP^N) = 3N+2$. Hence, $$\dim \Var = \dim \QMap_1(\CP^2,\CP^N) +2,$$ 
so that the counting data is balanced. The $\QM$ number of this configuration is 
$\binom{N}{2}$. 
The 1-freckle stratum now has codimension $N + 1 - 2$, so $\dim \QMap^1_1 = 2N +3$.
In particular, the 1-freckle strata are again given by configurations where the freckle and the running point collide with a fixed point, see Figure \ref{fig:one_freckle_stratum_P2_P3_modified}.
These configurations are strata of dimension $2N + 1 - N -4 = N-3$, in particular, for $N > 3$ they have position dimension and the counting data is not stable. However, in this case there is also a \emph{scar} stratum: we can have two freckles, each at one of the fixed points $c_1^X,c_2^X$. In this case, the line through this two points forms a scar (cf. Section \ref{sec:freckle_stratification}) and the running point can be at any point on the scar, as shown in Figure \ref{fig: scar stratum P2_P3}. This stratum is of the form $\CP^{N-4} \times \CP^1$, where the second factor describes the position of the running point along the scar, and intersects both 1-freckle strata in a $\CP^{N-4}$ when this running point hits either $c_1^X$ or $c_2^X$. 
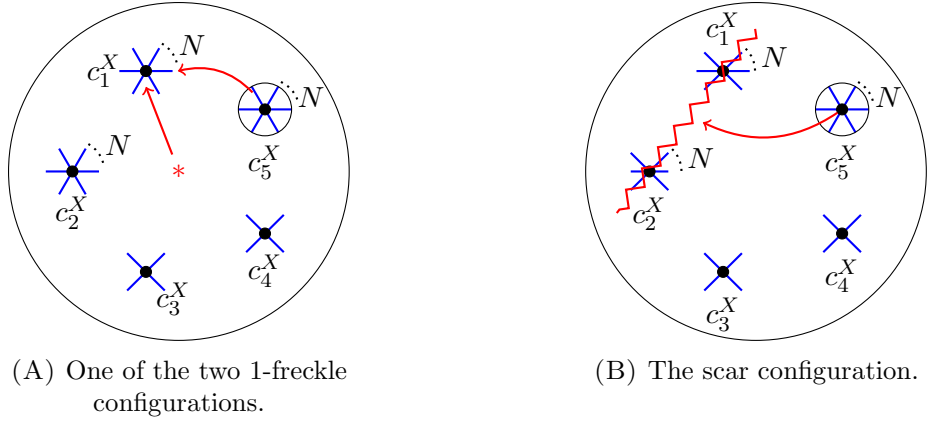
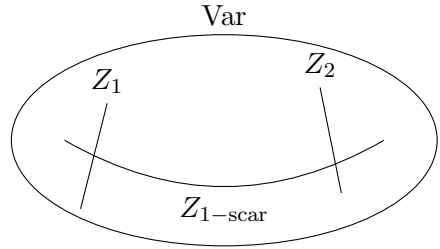
\begin{figure}[H]
    \centering
    \begin{subfigure}[t]{.4\textwidth}
            \centering
    \begin{tikzpicture}[baseline={([yshift=-.5ex]current bounding box.center)},scale=.7]
    \draw (0,0) circle (3.2);

    \coordinate[point] (c3) at (252:2);
    \coordinate[point] (c2) at (180:2);
    \coordinate[point] (c1) at (108:2);
    \coordinate[point] (c4) at (-36:2);

    \coordinate[point] (x) at (36:2);

    \begin{scope}[every node/.style={transform shape}]
        \draw (x) circle (.5);

        \draw pic at (x) {H={60}{.5}};
        \draw pic at (x) {H={-60}{.5}};
        \draw pic at (x) {H={0}{.5}};

        \draw pic at (c1) {H={60}{.5}};
        \draw pic at (c1) {H={-60}{.5}};
        \draw pic at (c1) {H={0}{.5}};
        
        \draw pic at (c2) {H={60}{.5}};
        \draw pic at (c2) {H={-60}{.5}};
        \draw pic at (c2) {H={0}{.5}};

        \draw pic at (c3) {H={45}{.5}};
        \draw pic at (c3) {H={-45}{.5}};

        \draw pic at (c4) {H={45}{.5}};
        \draw pic at (c4) {H={-45}{.5}};
    \end{scope}

    \draw[dotted,thick] ($(x) + (60:0.6)$) arc (65:15:0.6cm);
    \node at ($(x) + (15:.9)$) {$N$};
    \draw[dotted,thick] ($(c1) + (60:0.6)$) arc (60:15:0.6cm);
    \node at ($(c1) + (30:1)$) {$N$};
    \draw[dotted,thick] ($(c2) + (60:0.6)$) arc (60:15:0.6cm);
    \node at ($(c2) + (30:1)$) {$N$};

    \node[below right] at (c3) {$c_3^X$};
    \node[below,yshift=-.5em] at (c2) {$c_2^X$};
    \node[left,xshift=-.5em] at (c1) {$c_1^X$};
    \node[below] at ($(c4)+(0,-.2)$) {$c_4^X$};
    \node[below] at ($(x)+(0,-.5)$) {$c_5^X$};

    \node[red] (fr) at (0,0) {$*$};

    \node (xx) at (x) {$\phantom{*}$}; 
    \draw[thick,red,->] (xx) to[bend right] ($(c1) + (0:.6)$);
    \draw[thick,red,->] (fr) to ($(c1) + (-90:.3)$);
  \end{tikzpicture} 
    \caption{One of the two 1-freckle configurations.}
    \label{fig:one_freckle_stratum_P2_P3_modified}
    \end{subfigure}
    \hfill
    \begin{subfigure}[t]{.4\textwidth}
    \centering
    \begin{tikzpicture}[baseline={([yshift=-.5ex]current bounding box.center)},scale=.7]
    \draw (0,0) circle (3.2);

    \coordinate[point] (c1) at (108:2);
    \coordinate[point] (c2) at (180:2);
    \coordinate[point] (c3) at (252:2);
    \coordinate[point] (c4) at (-36:2);

    \coordinate[point] (x) at (36:2);

    \begin{scope}[every node/.style={transform shape}]
        \draw (x) circle (.5);

        \draw pic at (x) {H={60}{.5}};
        \draw pic at (x) {H={-60}{.5}};
        \draw pic at (x) {H={0}{.5}};

        \draw pic at (c1) {H={45}{.5}};
        \draw pic at (c1) {H={-45}{.5}};
        \draw pic at (c1) {H={0}{.5}};

        \draw pic at (c2) {H={45}{.5}};
        \draw pic at (c2) {H={-45}{.5}};
        \draw pic at (c2) {H={0}{.353}};

        \draw pic at (c3) {H={45}{.5}};
        \draw pic at (c3) {H={-45}{.5}};

        \draw pic at (c4) {H={45}{.5}};
        \draw pic at (c4) {H={-45}{.5}};
    \end{scope}

    \draw[dotted,thick] ($(x) + (60:0.6)$) arc (60:15:0.6cm);
    \node at ($(x) + (15:.9)$) {$N$};
    \draw[dotted,thick] ($(c1) + (45:0.6)$) arc (45:0:0.6cm);
    \node at ($(c1) + (15:1)$) {$N$};
    \draw[dotted,thick] ($(c2) + (45:0.6)$) arc (45:0:0.6cm);
    \node at ($(c2) + (15:1)$) {$N$};

    \node[above] at ($(c1) + (0,.3)$) {$c_1^X$};
    \node[below] at ($(c2) + (0,-.3)$) {$c_2^X$};
    \node[below] at ($(c3) + (0,-.3)$) {$c_3^X$};
    \node[below] at ($(c4) + (0,-.3)$) {$c_4^X$};
    \node[below] at ($(x) + (0,-.5)$) {$c_5^X$};

    \node (xx) at (x) {$\phantom{*}$}; 
    \draw[red,thick,decorate,decoration=zigzag] ($(c1) + (53:1)$) -- ($(c2) + (232:1)$);
    \draw[->,thick,red] (x) to[bend left] ($(c1) + (232:1) + (322:.3)$);
    
    \end{tikzpicture} 
    \caption{The scar configuration.}
    \label{fig: scar stratum P2_P3}
    \end{subfigure}
    \hfill
    \begin{subfigure}[t]{\textwidth}
    \centering
    \begin{tikzpicture}[baseline={([yshift=-.5ex]current bounding box.center)},scale=.7]
        \draw (0,0) ellipse (4cm and 2cm);
        \node[above] at (0,2) {$\Var$};
        \draw (-3,0) to[bend right] node[below, pos=0.5] {$Z_{1-{\rm scar}}$} (3,0);
        \draw (-2.7,-1.3) -- (-2.2,0.7);
        \node[above] at (-2.2,0.7) {$Z_1$}; 
        \draw (2.2,-1) -- (1.8,1);
        \node[above] at (1.8,1) {$Z_2$};
    \end{tikzpicture}
    \caption{Intersection of the 1-freckle strata $Z_1 = \CP^{N-3} \times \{ c_1^X\}$ and $Z_2 = \CP^{N-3} \times \{c_2^X\}$ with the 1-scar stratum $Z_{1-{\rm scar}} = \CP^{N-4} \times \CP^1$.}
    \end{subfigure}
    \caption{Unstable example with a scar.}
    \label{Fig 12}
 \end{figure}


\subsubsection{An 
example with a lot of stuff
}
Consider again the setup of example \ref{sec:1+3+12=16} but with a third running point added. 
We thus have $k=1$, $n=2$, $d=2$ but now $l=6$ and we take source cycles $c_i^X$ to be points for $i=1,2,3$ and $c_i^X=\CP^1$ for $i=4,5,6$. 
To compensate for the extra constraint introduced by the extra running point we take $c_1^Y$ to be a line and $c_i^Y$ to be a point for $i > 1$. 
The configuration is sketched in Figure \ref{fig: 2 + 62 = 64 1}. 
Since a line intersects a conic in two points, the KM number of this configuration is $2$. 
On the other hand, the $\GLSM$ number is 
\[
\int_{\CP^8 \times \CP^1_{(4)} \times \CP^1_{(5)}\times \CP^1_{(6)}}H^5\cdot \wedge_{i=4}^6(H + 2\, p_i^* h)^2 = 64.
\]
Hence we ought to find that the proper quasimap contribution is $\PQM = 62$.
\begin{figure}[H]
    \centering 
    \centering
    \begin{tikzpicture}[baseline={([yshift=-.5ex]current bounding box.center)},scale=.8]
    \draw (0,0) circle (3);

    \coordinate[point] (c3) at (150:2);
    \coordinate[point] (c2) at (210:2);
    \coordinate[point] (c1) at (270:1.5);

    \coordinate[point] (x) at (-30:2);
    \coordinate[point] (y) at (30:2);
    \coordinate[point] (z) at (90:2);

    \begin{scope}[every node/.style={transform shape}]
        \draw (x) circle (.5);
        \draw (y) circle (.5);
        \draw (z) circle (.5);

        \draw pic at (x) {H={45}{.5}};
        \draw pic at (x) {H={-45}{.5}};

        \draw pic at (y) {H={45}{.5}};
        \draw pic at (y) {H={-45}{.5}};

        \draw pic at (z) {H={45}{.5}};
        \draw pic at (z) {H={-45}{.5}};

        \draw pic at (c1) {H={45}{.5}};

        \draw pic at (c2) {H={45}{.5}};
        \draw pic at (c2) {H={-45}{.5}};

        \draw pic at (c3) {H={45}{.5}};
        \draw pic at (c3) {H={-45}{.5}};
    \end{scope}

    \node[below] at ($(c1) + (0,-.2)$) {$c_1^X$};
    \node[below] at ($(c2) + (0,-.2)$) {$c_2^X$};
    \node[below] at ($(c3) + (0,-.2)$)  {$c_3^X$};

    \node[below] at ($(x) + (0,-.5)$) {$x$};
    \node[below] at ($(y) + (0,-.5)$) {$y$};
    \node[below] at ($(z) + (0,-.5)$) {$z$};
  \end{tikzpicture} 
    \caption{Pictorial representation of the enumerative problem. }
    \label{fig: 2 + 62 = 64 1}
\end{figure}
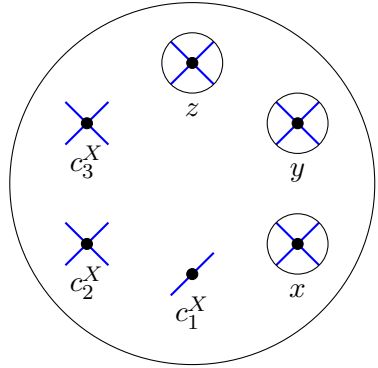

The space of proper quasimaps can be decomposed into the strata dipicted in Figure \ref{fig: 2 + 62 = 64 strata}. 

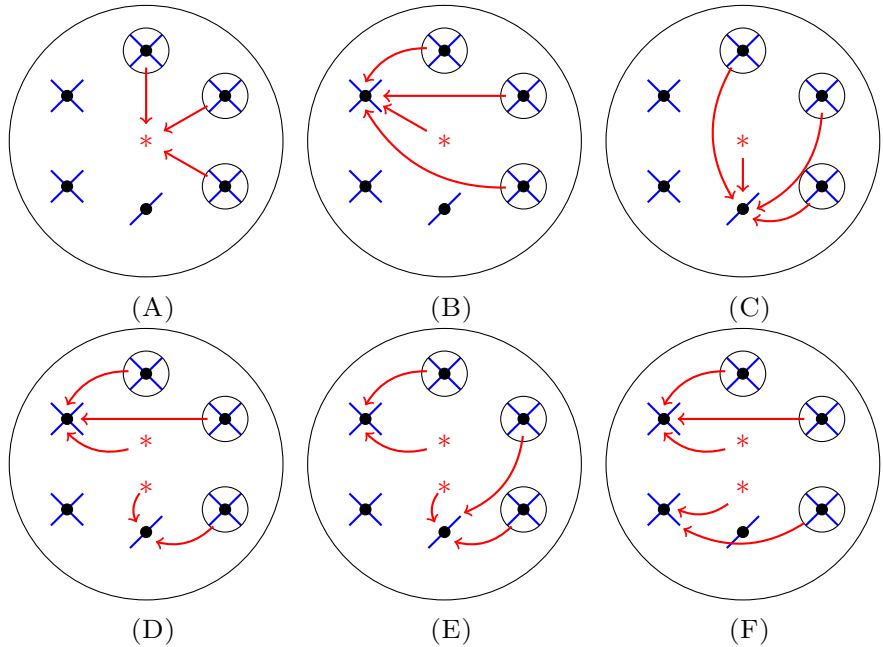
\begin{figure}[H]
    \centering
    \begin{subfigure}[t]{0.3\textwidth}
      \begin{tikzpicture}[baseline={([yshift=-.5ex]current bounding box.center)},scale=.6]
    \draw (0,0) circle (3);

    \coordinate[point] (c3) at (150:2);
    \coordinate[point] (c2) at (210:2);
    \coordinate[point] (c1) at (270:1.5);

    \coordinate[point] (x) at (-30:2);
    \coordinate[point] (y) at (30:2);
    \coordinate[point] (z) at (90:2);

    \begin{scope}[every node/.style={transform shape}]
        \draw (x) circle (.5);
        \draw (y) circle (.5);
        \draw (z) circle (.5);

        \draw pic at (x) {H={45}{.5}};
        \draw pic at (x) {H={-45}{.5}};

        \draw pic at (y) {H={45}{.5}};
        \draw pic at (y) {H={-45}{.5}};

        \draw pic at (z) {H={45}{.5}};
        \draw pic at (z) {H={-45}{.5}};

        \draw pic at (c1) {H={45}{.5}};

        \draw pic at (c2) {H={45}{.5}};
        \draw pic at (c2) {H={-45}{.5}};

        \draw pic at (c3) {H={45}{.5}};
        \draw pic at (c3) {H={-45}{.5}};
    \end{scope}



    \node[red] (fr) at (0,0) {$*$};

    \node (xx) at (x) {$\phantom{*}$}; 
    \node (yy) at (y) {$\phantom{*}$};
    \node (zz) at (z) {$\phantom{*}$};

    \draw[thick,red,->] (xx) to (fr);
    \draw[thick,red,->] (yy) to (fr);
    \draw[thick,red,->] (zz) to (fr);
  \end{tikzpicture}
  \caption{}
    \label{fig:one_freckle_stratum_1}
    \end{subfigure}
    \begin{subfigure}[t]{0.3\textwidth}
    \begin{tikzpicture}[baseline={([yshift=-.5ex]current bounding box.center)},scale=.6]
    \draw (0,0) circle (3);

    \coordinate[point] (c3) at (150:2);
    \coordinate[point] (c2) at (210:2);
    \coordinate[point] (c1) at (270:1.5);

    \coordinate[point] (x) at (-30:2);
    \coordinate[point] (y) at (30:2);
    \coordinate[point] (z) at (90:2);

    \begin{scope}[every node/.style={transform shape}]
        \draw (x) circle (.5);
        \draw (y) circle (.5);
        \draw (z) circle (.5);

        \draw pic at (x) {H={45}{.5}};
        \draw pic at (x) {H={-45}{.5}};

        \draw pic at (y) {H={45}{.5}};
        \draw pic at (y) {H={-45}{.5}};

        \draw pic at (z) {H={45}{.5}};
        \draw pic at (z) {H={-45}{.5}};

        \draw pic at (c1) {H={45}{.5}};

        \draw pic at (c2) {H={45}{.5}};
        \draw pic at (c2) {H={-45}{.5}};

        \draw pic at (c3) {H={45}{.5}};
        \draw pic at (c3) {H={-45}{.5}};
    \end{scope}



    \node[red] (fr) at (0,0) {$*$};

    \node (xx) at (x) {$\phantom{*}$}; 
    \node (yy) at (y) {$\phantom{*}$};
    \node (zz) at (z) {$\phantom{*}$};
    \node (cc3) at (c3){$\phantom{*}$};

    \draw[thick,red,->] (fr) to (cc3);
    \draw[thick,red,->] (xx) to[bend left] ($(cc3)+(-90:.3)$);
    \draw[thick,red,->] (yy) to (cc3);
    \draw[thick,red,->] (zz) to[bend right]  ($(cc3)+(90:.3)$);
  \end{tikzpicture}
  \caption{}
    \label{fig:one_freckle_stratum_2}
    \end{subfigure}
  \begin{subfigure}[t]{0.3\textwidth}
  \begin{tikzpicture}[baseline={([yshift=-.5ex]current bounding box.center)},scale=.6]
    \draw (0,0) circle (3);

    \coordinate[point] (c3) at (150:2);
    \coordinate[point] (c2) at (210:2);
    \coordinate[point] (c1) at (270:1.5);

    \coordinate[point] (x) at (-30:2);
    \coordinate[point] (y) at (30:2);
    \coordinate[point] (z) at (90:2);

    \begin{scope}[every node/.style={transform shape}]
        \draw (x) circle (.5);
        \draw (y) circle (.5);
        \draw (z) circle (.5);

        \draw pic at (x) {H={45}{.5}};
        \draw pic at (x) {H={-45}{.5}};

        \draw pic at (y) {H={45}{.5}};
        \draw pic at (y) {H={-45}{.5}};

        \draw pic at (z) {H={45}{.5}};
        \draw pic at (z) {H={-45}{.5}};

        \draw pic at (c1) {H={45}{.5}};

        \draw pic at (c2) {H={45}{.5}};
        \draw pic at (c2) {H={-45}{.5}};

        \draw pic at (c3) {H={45}{.5}};
        \draw pic at (c3) {H={-45}{.5}};
    \end{scope}



    \node[red] (fr) at (0,0) {$*$};

    \node (xx) at (x) {$\phantom{*}$}; 
    \node (yy) at (y) {$\phantom{*}$};
    \node (zz) at (z) {$\phantom{*}$};
    \node (cc1) at (c1){$\phantom{*}$};

    \draw[thick,red,->] (fr) to (cc1);
    \draw[thick,red,->] (xx) to[bend left] ($(cc1)+(-45:.3)$);
    \draw[thick,red,->] (yy) to[bend left] ($(cc1)+(0:.3)$);
    \draw[thick,red,->] (zz) to[bend right]  ($(cc1)+(135:.3)$);
  \end{tikzpicture}
  \caption{}
    \label{fig:one_freckle_stratum_3}
    \end{subfigure}
    \begin{subfigure}[t]{0.3\textwidth}
  \begin{tikzpicture}[baseline={([yshift=-.5ex]current bounding box.center)},scale=.6]
    \draw (0,0) circle (3);

    \coordinate[point] (c3) at (150:2);
    \coordinate[point] (c2) at (210:2);
    \coordinate[point] (c1) at (270:1.5);

    \coordinate[point] (x) at (-30:2);
    \coordinate[point] (y) at (30:2);
    \coordinate[point] (z) at (90:2);

    \begin{scope}[every node/.style={transform shape}]
        \draw (x) circle (.5);
        \draw (y) circle (.5);
        \draw (z) circle (.5);

        \draw pic at (x) {H={45}{.5}};
        \draw pic at (x) {H={-45}{.5}};

        \draw pic at (y) {H={45}{.5}};
        \draw pic at (y) {H={-45}{.5}};

        \draw pic at (z) {H={45}{.5}};
        \draw pic at (z) {H={-45}{.5}};

        \draw pic at (c1) {H={45}{.5}};

        \draw pic at (c2) {H={45}{.5}};
        \draw pic at (c2) {H={-45}{.5}};

        \draw pic at (c3) {H={45}{.5}};
        \draw pic at (c3) {H={-45}{.5}};
    \end{scope}



    \node[red] (fr1) at (0,-.5) {$*$};
    \node[red] (fr2) at (0,.5) {$*$};

    \node (xx) at (x) {$\phantom{*}$}; 
    \node (yy) at (y) {$\phantom{*}$};
    \node (zz) at (z) {$\phantom{*}$};
    \node (cc1) at (c1){$\phantom{*}$};
    \node (cc3) at (c3){$\phantom{*}$};

    \draw[thick,red,->] (fr2) to[bend left] ($(cc3) + (-90:.3)$);
    \draw[thick,red,->] ($(fr1.center)+ (-135:.2)$) to[bend right] ($(cc1) + (135:.3)$);
    
    \draw[thick,red,->] (xx) to[bend left] ($(cc1)+(-45:.3)$);
    \draw[thick,red,->] (yy) to ($(cc3)+(0:.3)$);
    \draw[thick,red,->] (zz) to[bend right]  ($(cc3)+(90:.3)$);
  \end{tikzpicture}
  \caption{}
    \label{fig:two_freckle_stratum_1}
    \end{subfigure}
    \begin{subfigure}[t]{0.3\textwidth}
  \begin{tikzpicture}[baseline={([yshift=-.5ex]current bounding box.center)},scale=.6]
    \draw (0,0) circle (3);

    \coordinate[point] (c3) at (150:2);
    \coordinate[point] (c2) at (210:2);
    \coordinate[point] (c1) at (270:1.5);

    \coordinate[point] (x) at (-30:2);
    \coordinate[point] (y) at (30:2);
    \coordinate[point] (z) at (90:2);

    \begin{scope}[every node/.style={transform shape}]
        \draw (x) circle (.5);
        \draw (y) circle (.5);
        \draw (z) circle (.5);

        \draw pic at (x) {H={45}{.5}};
        \draw pic at (x) {H={-45}{.5}};

        \draw pic at (y) {H={45}{.5}};
        \draw pic at (y) {H={-45}{.5}};

        \draw pic at (z) {H={45}{.5}};
        \draw pic at (z) {H={-45}{.5}};

        \draw pic at (c1) {H={45}{.5}};

        \draw pic at (c2) {H={45}{.5}};
        \draw pic at (c2) {H={-45}{.5}};

        \draw pic at (c3) {H={45}{.5}};
        \draw pic at (c3) {H={-45}{.5}};
    \end{scope}



    \node[red] (fr1) at (0,-.5) {$*$};
    \node[red] (fr2) at (0,.5) {$*$};

    \node (xx) at (x) {$\phantom{*}$}; 
    \node (yy) at (y) {$\phantom{*}$};
    \node (zz) at (z) {$\phantom{*}$};
    \node (cc1) at (c1){$\phantom{*}$};
    \node (cc3) at (c3){$\phantom{*}$};

    \draw[thick,red,->] (fr2) to[bend left] ($(cc3) + (-90:.3)$);
    \draw[thick,red,->] ($(fr1.center)+ (-135:.2)$) to[bend right] ($(cc1) + (135:.3)$);
    
    \draw[thick,red,->] (xx) to[bend left] ($(cc1)+(-45:.3)$);
    \draw[thick,red,->] (yy) to[bend left] ($(cc1)+(45:.6)$);
    \draw[thick,red,->] (zz) to[bend right]  ($(cc3)+(90:.3)$);
  \end{tikzpicture}
    \caption{}
    \label{fig:two_freckle_stratum_2}
    \end{subfigure}
    \begin{subfigure}[t]{0.3\textwidth}









    
    
  \begin{tikzpicture}[baseline={([yshift=-.5ex]current bounding box.center)},scale=.6]
    \draw (0,0) circle (3);

    \coordinate[point] (c3) at (150:2);
    \coordinate[point] (c2) at (210:2);
    \coordinate[point] (c1) at (270:1.5);

    \coordinate[point] (x) at (-30:2);
    \coordinate[point] (y) at (30:2);
    \coordinate[point] (z) at (90:2);

    \begin{scope}[every node/.style={transform shape}]
        \draw (x) circle (.5);
        \draw (y) circle (.5);
        \draw (z) circle (.5);

        \draw pic at (x) {H={45}{.5}};
        \draw pic at (x) {H={-45}{.5}};

        \draw pic at (y) {H={45}{.5}};
        \draw pic at (y) {H={-45}{.5}};

        \draw pic at (z) {H={45}{.5}};
        \draw pic at (z) {H={-45}{.5}};

        \draw pic at (c1) {H={45}{.5}};

        \draw pic at (c2) {H={45}{.5}};
        \draw pic at (c2) {H={-45}{.5}};

        \draw pic at (c3) {H={45}{.5}};
        \draw pic at (c3) {H={-45}{.5}};
    \end{scope}



    \node[red] (fr1) at (0,-.5) {$*$};
    \node[red] (fr2) at (0,.5) {$*$};

    \node (xx) at (x) {$\phantom{*}$}; 
    \node (yy) at (y) {$\phantom{*}$};
    \node (zz) at (z) {$\phantom{*}$};
    \node (cc1) at (c1){$\phantom{*}$};
    \node (cc2) at (c2){$\phantom{*}$};

    \draw[thick,red,->] (fr2) to[bend left] ($(cc3) + (-90:.3)$);
    \draw[thick,red,->] (fr1) to[bend left] ($(cc2) + (0:.3)$);
    
    \draw[thick,red,->] (xx) to[bend left] ($(cc2)+(-45:.6)$);
    \draw[thick,red,->] (yy) to ($(cc3)+(0:.3)$);
    \draw[thick,red,->] (zz) to[bend right]  ($(cc3)+(90:.3)$);
  \end{tikzpicture}
  \caption{}
    \label{fig:two_freckle_stratum_3}
    \end{subfigure}
    \caption{Graphical description of the strata of $\QMap_2^{\pr}(\CP^1,\CP^2,\mathcal D)$.}
    \label{fig: 2 + 62 = 64 strata}
\end{figure}

There is a unique stratum of the form shown in Figure \ref{fig:one_freckle_stratum_1} and \ref{fig:one_freckle_stratum_3},
two strata of the form depicted in Figure \ref{fig:one_freckle_stratum_2}, six strata of the situation shown in Figure \ref{fig:two_freckle_stratum_1} and \ref{fig:two_freckle_stratum_2} and three strata of the situation shown in Figure \ref{fig:two_freckle_stratum_3}.

We label the strata according to their graphical depiction by $Z_{(A)}$ to $Z_{(F)}$.


In situation $Z_{(A)}$, on the complement of the freckle the quasimap has degree 1.
Since $\dim\QMap_1(\CP^1,\CP^2) = 5$, the quasimap is uniquely fixed by imposing the remaining 5 equations.
The only degree of freedom left is the position of the freckle, hence $Z_{(A)} \cong \CP^1$.
This $\CP^1$ is embedded into $\Var$ diagonally: if $y = (y^0:y^1) \in \CP^1$ denotes the position of the freckle, then  
\begin{equation}
    \begin{split}
        Z_{(A)} &\hookrightarrow \Var = \CP^8 \times \CP^1 \times \CP^1 \times \CP^1 \\
        y &\mapsto ( \underline f,\ y,\ y,\ y )
    \end{split}
\end{equation}
where $\CP^8 = \QMap_2(\CP^1,\CP^2)$ and 
\begin{equation}
    \underline f = (Q(x)P_1(x):Q(x)P_3(x):Q(x)P_3(x)).
\end{equation}
Here, the $P_i(x)$ are homogeneous degree 1 polynomials of $x = (x^0:x^1)$ and $Q(x) = (y^1x^0 - y^0 x^1)$.


In situation $Z_{(B)}$, on the complement of the freckle, one is faced with a degree 1 quasimap subject to three constraints. 
Since $\QMap_1(\CP^1,\CP^2) = \CP^5$, this means that $Z_{(B)} \cong \CP^2  \times \{ c_3^X \} \times \{ c_3^X \} \times \{ c_3^X \} \subset \Var$. 

The stratum $Z_{(C)}$ is analyzed analogously to the stratum $Z_{(B)}$.
However, now one has to impose four constraints 
on the complement of the freckle resulting in $Z_{(C)} \cong \CP^1  \times \{ c_1^X \} \times \{ c_1^X \} \times \{ c_1^X \} \subset \Var$.

For $Z_{(D)}$ the quasimap on the complement of the freckles is constant 
and hence uniquely fixed by imposing the remaining two equations.
For example, with the notation of Figure \ref{fig: 2 + 62 = 64 strata}, if $c_i^X = y_i = (y_i^0:y_i^1)$ then 
\begin{equation}
     Z_{(D)} = \{ (\underline f, c^X_1, c^X_3, c^X_3 )\} \subset \Var, 
\end{equation}
where 
\begin{equation}
    \underline f(x) = ( aQ(y_1)Q(y_3) :  bQ(y_1)Q(y_3) :  cQ(y_1)Q(y_3)),
\end{equation}
with $Q(y) = (y^1x^0 - y^0 x^1)$ and $c_2^Y = (a:b:c) \in \CP^2$.

By the same reasoning, the stratum $Z_{(E)}$ is likewise a point, which can be described analogously as 
\begin{equation}
    Z_{(E)} = \{ (\underline f, c^X_1, c^X_1, c^X_3 )\} \subset \Var,
\end{equation}
where 
\begin{equation}
    \underline f(x) = ( aQ(y_1)Q(y_3) :  bQ(y_1)Q(y_3) :  cQ(y_1)Q(y_3)).
\end{equation} 

Finally, in the situation $Z_{(F)}$, the quasimap on the complement of the freckle is again constant, but is mapped to a line $\CP^1 \subset \CP^2$, 
rather than a point. 
Therefore, $Z_{(F)} \cong \CP^1 \times \{ c_2^X \} \times \{ c_3^X\} \times \{ c_3^X \} \subset \Var$.

It is important to note that the strata $Z_{(I)}$ are not all disjoint.
Indeed $Z_{(A)}$ intersects $Z_{(B)}$ and $Z_{(C)}$ non-trivially, cf.\ Figure \ref{fig:non_triv_intersection_strata_64}.
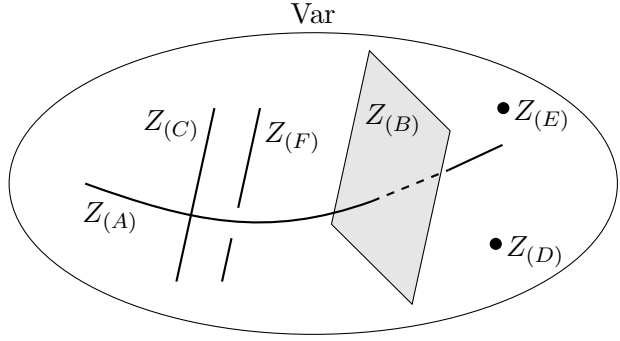
\begin{figure}[H]
    \centering
    \begin{tikzpicture}
        \draw (0,0) ellipse (4cm and 2cm);
        \node[above] at (0,2) {$\Var$};
        
        \path[decoration={
      /tikz/postaction={
          decoration={
              markings,
              mark=at position 0.0 with \coordinate (initial);,
              mark=at position 0.6 with \coordinate (middle1);,
              mark=at position 0.77 with \coordinate (middle2);,
              mark=at position 0.9 with \coordinate (final);
          },decorate
      }
  },decorate]
          (-3,0) to[bend right] node[below, pos=0.05] {$Z_{(A)}$} (3,1);

        \draw[thick] (initial) to[out=-20,in=200] (middle1); 
        \draw[thick,dashed] (middle1) to (middle2); 
        \draw[thick] (middle2) to ($(final) + (0,-0.1)$);
        
        \draw[thick] (-1.8,-1.3) -- (-1.3,1) node[left, pos=.9] {$Z_{(C)}$};

        \coordinate (ZBA) at (1.3,-1.6);
        \coordinate (ZBB) at (1.8,.7);
        \coordinate (ZBC) at ($(1.8,.7)+ (135:1.5)$);
        \coordinate (ZBD) at ($(1.3,-1.6) + (135:1.5)$);
        
        \begin{pgfonlayer}{bg}
        \draw[fill=gray!20] (ZBA) --  (ZBB) -- (ZBC) -- node[right, pos=.4] {$Z_{(B)}$} (ZBD) -- cycle;
        \end{pgfonlayer}

        \coordinate[point] (ZD) at (2.4,-.8);
        \coordinate[point] (ZE) at (2.5,1);
        \node[right, yshift=-3pt] at (ZD) {$Z_{(D)}$};
        \node[right, yshift=-3pt] at (ZE) {$Z_{(E)}$};

        \draw[thick]  (-1.2,-1.3) -- ($(-1.2,-1.3)+0.25*(0.5,2.3)$) ($(-1.2,-1.3) + 0.425*(0.5,2.3)$) -- node[right, pos = .7] {$Z_{(F)}$} ($(-1.2,-1.3)+ (0.5,2.3)$);

    \end{tikzpicture}
    \caption{Intersection of the strata $Z_{(A)} \cong \CP^1$, $Z_{(B)} \cong \CP^2$ and $Z_{(C)} \cong \CP^1$ inside $\QMap \subset \Var$.}
    \label{fig:non_triv_intersection_strata_64}
\end{figure}

If we were to compute the contributions of the strata $Z_{(I)}$ independently (without taking intersections into account) we would find 
\begin{equation}
    \sum_{I \in \{ A, \dots, F\}} \PQM(\CP^1,\CP^2,\mathcal D; Z_{(I)}) = 10 + 20 + 4 + 6 + 6 + 12 = 58 \neq \underbrace{62}_{\QM-\KM}.
\end{equation}
Again, the correct computation should take intersection into account and involve the Segre class.

\subsection{A quasi-stable example with non-trivial source cycles (a computation where scars and freckles work together)} \label{ss: example: semi-moving points}
We consider degree 1 quasimaps from $\CP^2$ to $\CP^3$ with $\QMap_1(\CP^2,\CP^3) = \CP^{11}$.
As quasimap counting data, we consider three points $\{c_i^X\}_{i=1}^3$ and two lines $\{c_i^X\}_{i=4}^5$ in the source and likewise three points $\{c_i^Y\}_{i=1}^3$ and two lines $\{c_i^Y\}_{i=4}^5$ in the target.
We will represent a line in the source by a dashed line and a point moving on that line by a box around the moving point. 
As before, we denote by a solid line a hyperplane in the target.
We then consider the following three situations, cf.\ Figure \ref{fig:counting_data_deg1_P2_to_P3}:
\begin{itemize}
    \item[$\mathcal D_{(1)}$:] each of the three points $\{c_i^X\}_{i=1}^3$ is mapped to a point while each of the two lines $\{c_i^X\}_{i=4}^5$ pass through a line \label{item:P2 -> P3 non-triv cycles situation 1}
    \item[$\mathcal D_{(2)}$:] two of the three points $\{c_i^X\}_{i=1}^3$ are mapped to a point, while the other is mapped to a line; one of the lines $\{c_i^X\}_{i=4}^5$ passes by a point, while the other passes through a line \label{item:P2 -> P3 non-triv cycles situation 2}
    \item[$\mathcal D_{(3)}$:] two of the three points $\{c_i^X\}_{i=1}^3$ are mapped to a line, while the other is mapped to a point; each line $\{c_i^X\}_{i=4}^5$ passes through a point \label{item:P2 -> P3 non-triv cycles situation 3}
\end{itemize}

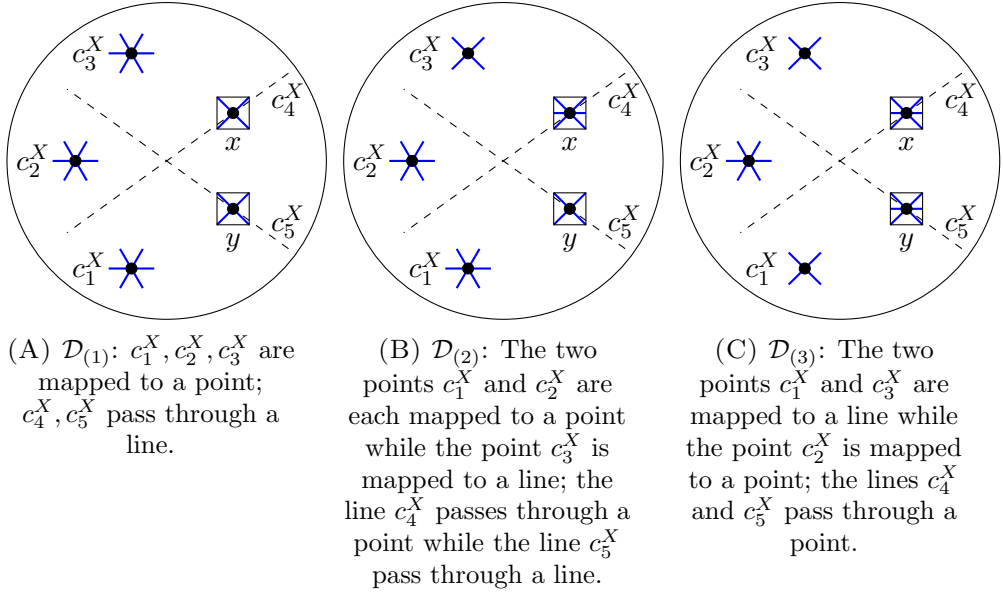
\begin{figure}[H] 
    \centering 
\begin{subfigure}[t]{0.3\textwidth}
    \centering
    \begin{tikzpicture}[baseline={([yshift=-.5ex]current bounding box.center)},scale=.6]
    \draw (0,0) circle (3.5);

    \coordinate[point] (c3) at (108:2.5);
    \coordinate[point] (c2) at (180:2);
    \coordinate[point] (c1) at (252:2.5);

    \coordinate[point] (x) at (36:1.8);
    \coordinate[point] (y) at (324:1.8);
    
    \begin{scope}[every node/.style={transform shape}]
        \node[square] at (x) {};
        \node[square] at (y) {};

        \draw pic at (x) {H={-45}{.5}};
        \draw pic at (x) {H={45}{.5}};
            
        \draw pic at (y) {H={-45}{.5}};
        \draw pic at (y) {H={45}{.5}};

        \draw pic at (c1) {H={-60}{.5}};
        \draw pic at (c1) {H={60}{.5}};
        \draw pic at (c1) {H={0}{.5}};

        \draw pic at (c3) {H={-60}{.5}};
        \draw pic at (c3) {H={60}{.5}};
        \draw pic at (c3) {H={0}{.5}};

        \draw pic at (c2) {H={-60}{.5}};
        \draw pic at (c2) {H={60}{.5}};
        \draw pic at (c2) {H={0}{.5}};
    \end{scope}

    \node[left, xshift=-.5em] at (c1) {$c_1^X$};
    \node[left, xshift=-.5em] at (c2) {$c_2^X$};
    \node[left, xshift=-.5em] at (c3) {$c_3^X$};

    \node[below, yshift=-.5em] at (x) {$x$};
    \node[below, yshift=-.5em] at (y) {$y$};

    \draw[dashed] (324:3.3) to ($(324:1.8) + (144:4.5)$);
    \node[above] at ($(324:1.8) + (324:1.5)$) {$c_5^X$};

    \draw[dashed]  (36:3.3) to ($(36:1.8) + (216:4.5)$);
    \node[below] at ($(36:1.8) + (36:1.5)$) {$c_4^X$};
  \end{tikzpicture} 
\caption{$\mathcal D_{(1)}$: $c^X_1, c_2^X,c_3^X$ are mapped to a point; $c_4^X, c_5^X$ pass through a line.} 
\label{fig:counting_data_deg1_P2_to_P3_1}
    \end{subfigure}\hfill
    \begin{subfigure}[t]{0.3\textwidth}
    \centering
    \begin{tikzpicture}[baseline={([yshift=-.5ex]current bounding box.center)},scale=.6]
    \draw (0,0) circle (3.5);

    \coordinate[point] (c3) at (108:2.5);
    \coordinate[point] (c2) at (180:2);
    \coordinate[point] (c1) at (252:2.5);

    \coordinate[point] (x) at (36:1.8);
    \coordinate[point] (y) at (324:1.8);
    
    \begin{scope}[every node/.style={transform shape}]
        \node[square] at (x) {};
        \node[square] at (y) {};

        \draw pic at (x) {H={45}{.5}};
        \draw pic at (x) {H={-45}{.5}};
        \draw pic at (x) {H={0}{.353}};
            
        \draw pic at (y) {H={-45}{.5}};
        \draw pic at (y) {H={45}{.5}};

        \draw pic at (c1) {H={-60}{.5}};
        \draw pic at (c1) {H={60}{.5}};
        \draw pic at (c1) {H={0}{.5}};

        \draw pic at (c3) {H={-45}{.5}};
        \draw pic at (c3) {H={45}{.5}};

        \draw pic at (c2) {H={-60}{.5}};
        \draw pic at (c2) {H={60}{.5}};
        \draw pic at (c2) {H={0}{.5}};
    \end{scope}

    \node[left, xshift=-.5em] at (c1) {$c_1^X$};
    \node[left, xshift=-.5em] at (c2) {$c_2^X$};
    \node[left, xshift=-.5em] at (c3) {$c_3^X$};

    \node[below, yshift=-.5em] at (x) {$x$};
    \node[below, yshift=-.5em] at (y) {$y$};

    \draw[dashed] (324:3.3) to ($(324:1.8) + (144:4.5)$);
    \node[above] at ($(324:1.8) + (324:1.5)$) {$c_5^X$};

    \draw[dashed]  (36:3.3) to ($(36:1.8) + (216:4.5)$);
    \node[below] at ($(36:1.8) + (36:1.5)$) {$c_4^X$};

  \end{tikzpicture} 
\caption{$\mathcal D_{(2)}$: The two points $c^X_1$ and $c_2^X$ are each mapped to a point while the point $c_3^X$ is mapped to a line; the line $c_4^X$ passes through a point while the line $c_5^X$ pass through a line.}
\label{fig:counting_data_deg1_P2_to_P3_2}
    \end{subfigure}\hfill
    \begin{subfigure}[t]{0.3\textwidth}
    \centering
    \begin{tikzpicture}[baseline={([yshift=-.5ex]current bounding box.center)},scale=.6]
    \draw (0,0) circle (3.5);

    \coordinate[point] (c3) at (108:2.5);
    \coordinate[point] (c2) at (180:2);
    \coordinate[point] (c1) at (252:2.5);

    \coordinate[point] (x) at (36:1.8);
    \coordinate[point] (y) at (324:1.8);

    \begin{scope}[every node/.style={transform shape}]
        \node[square] at (x) {};
        \node[square] at (y) {};

        \draw pic at (x) {H={-45}{.5}};
        \draw pic at (x) {H={45}{.5}};
        \draw pic at (x) {H={0}{.353}};
            
        \draw pic at (y) {H={-45}{.5}};
        \draw pic at (y) {H={45}{.5}};
        \draw pic at (y) {H={0}{.353}};

        \draw pic at (c1) {H={-45}{.5}};
        \draw pic at (c1) {H={45}{.5}};

        \draw pic at (c3) {H={-45}{.5}};
        \draw pic at (c3) {H={45}{.5}};
        
        \draw pic at (c2) {H={-60}{.5}};
        \draw pic at (c2) {H={60}{.5}};
        \draw pic at (c2) {H={0}{.5}};
    \end{scope}

    \node[left, xshift=-.5em] at (c1) {$c_1^X$};
    \node[left, xshift=-.5em] at (c2) {$c_2^X$};
    \node[left, xshift=-.5em] at (c3) {$c_3^X$};

    \node[below, yshift=-.5em] at (x) {$x$};
    \node[below, yshift=-.5em] at (y) {$y$};

    \draw[dashed] (324:3.3) to ($(324:1.8) + (144:4.5)$);
    \node[above] at ($(324:1.8) + (324:1.5)$) {$c_5^X$};

    \draw[dashed]  (36:3.3) to ($(36:1.8) + (216:4.5)$);
    \node[below] at ($(36:1.8) + (36:1.5)$) {$c_4^X$};
  \end{tikzpicture} 
\caption{$\mathcal D_{(3)}$: The two points $c^X_1$ and $c_3^X$ are mapped to a line while the point $c_2^X$ is mapped to a point; the lines $c_4^X$ and $c_5^X$ pass through a point.}
\label{fig:counting_data_deg1_P2_to_P3_3}
    \end{subfigure}
    \caption{Quasimap counting data $\mathcal D$ for $\QMap_1(\CP^2,\CP^3)$. We denote the cycles (lines) $c_4^X, c_5^X \subset \CP^2$ by a dashed line. The box around the points $x \in c_4^X$ and $y \in c_5^X$ indicate that this point can move along the respective line $c_i^X$.}
    \label{fig:counting_data_deg1_P2_to_P3}
\end{figure}

 \textbf{Situation $D_{(1)}$.}
 In this case the $\GLSM$ number is given by
 \begin{equation}
     \GLSM(\CP^2,\CP^3,\mathcal D_{(1)}) = \int_{\CP^{11} \times \CP^1_{(4)} \times \CP^1_{(5)}} H^9 (H + p_4^* h)^2(H + p_5^* h)^2 = 4.
 \end{equation}
 \begin{proposition}
     In the present situation, there exists an unique holomorphic map, i.e.\ 
     \[
        \KM(\CP^2,\CP^3,\mathcal D_{(1)}) = 1.
     \]
 \end{proposition}
 \begin{proof}[Proof by geometry]
     There is a unique plane $\Pi$ through three points $c_1^Y,c_2^Y,c_3^Y$ in $\CP^3$. It passes through the lines $c_4^Y,c_5^Y$. Its parametrization is uniquely fixed by the data $\mc{D}$. More explicitly: fix some map $f_0\colon \CP^2\ra \CP^3$ with image $\Pi$. The preimages of $c_i^Y$, $i=1,\ldots,5$ are five points on $\CP^2$. The problem of finding a $\PSL(3,\CC)$ transformation $g$ moving those points to points $c_1^X,c_2^X,c_3^X$ and lines $c_4^X,c_5^X$ is a linear problem and has a unique solution. Then $f=f_0\circ g^{-1}$ is the desired (unique) holomorphic map.
 \end{proof}
  \begin{proof}[Proof by counting proper quasimaps]
     Note that a quasimap satisfying all conditions cannot have a freckle.
     Indeed, assume there exists a freckle. 
     Recall that the one-freckle stratum $\QMap_1^1(\CP^2.\CP^3)$ has complex codimension $2$, cf.\ \ref{sec:freckle_stratification}. 
     In order to have  a chance to solve all equations, the freckle must sit at the intersection of the lines $c_4^X$ and $c_5^X$, see Figure \ref{fig:counting_data_deg1_P2_to_P3_1_freckle}.
    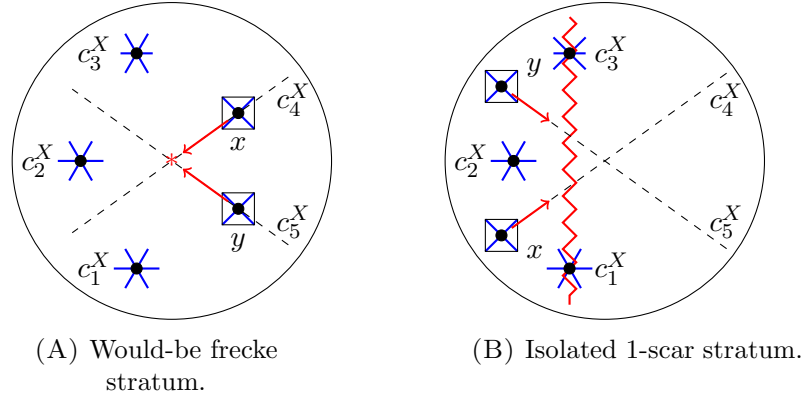
\begin{figure}[H]
        \centering\hfill
        \begin{subfigure}[t]{.3\textwidth}
        \begin{tikzpicture}[baseline={([yshift=-.5ex]current bounding box.center)},scale=.6]
        \draw (0,0) circle (3.5);

        \coordinate[point] (c3) at (108:2.5);
        \coordinate[point] (c2) at (180:2);
        \coordinate[point] (c1) at (252:2.5);
    
        \coordinate[point] (x) at (36:1.8);
        \coordinate[point] (y) at (324:1.8);
        
        \begin{scope}[every node/.style={transform shape}]
            \node[square] at (x) {};
            \node[square] at (y) {};
    
            \draw pic at (x) {H={-45}{.5}};
            \draw pic at (x) {H={45}{.5}};
                
            \draw pic at (y) {H={-45}{.5}};
            \draw pic at (y) {H={45}{.5}};
    
            \draw pic at (c1) {H={-60}{.5}};
            \draw pic at (c1) {H={60}{.5}};
            \draw pic at (c1) {H={0}{.5}};
    
            \draw pic at (c3) {H={-60}{.5}};
            \draw pic at (c3) {H={60}{.5}};
            \draw pic at (c3) {H={0}{.353}};
    
            \draw pic at (c2) {H={-60}{.5}};
            \draw pic at (c2) {H={60}{.5}};
            \draw pic at (c2) {H={0}{.5}};
        \end{scope}

        \node[left, xshift=-.5em] at (c1) {$c_1^X$};
        \node[left, xshift=-.5em] at (c2) {$c_2^X$};
        \node[left, xshift=-.5em] at (c3) {$c_3^X$};
    
        \node[below, yshift=-.5em] at (x) {$x$};
        \node[below, yshift=-.5em] at (y) {$y$};
    
        \draw[dashed] (y) to ($(y) + (144:4.5)$);
        \draw[dashed] (y) to ($(y) + (324:1.5)$);
        \node[above] at ($(y) + (324:1.5)$) {$c_5^X$};
    
        \draw[dashed] (x) to ($(x) + (216:4.5)$);
        \draw[dashed] (x) to ($(x) + (36:1.5)$);
        \node[below] at ($(x) + (36:1.5)$) {$c_4^X$};

        \draw[red,yshift=-0.2em,scale=1.5] node at (0,0) {*};

        \node (xx) at (x) {};
        \node (yy) at (y) {};
        \draw[red,thick,->] (xx) -- ($(0,0) + (36:.3)$);
        \draw[red,thick,->] (yy) -- ($(0,0) + (-36:.3)$);
      \end{tikzpicture} 
      \caption{Would-be frecke stratum.}
        \label{fig:counting_data_deg1_P2_to_P3_1_freckle}
        \end{subfigure}\hfill
        \begin{subfigure}[t]{.4\textwidth}
                \begin{tikzpicture}[baseline={([yshift=-.5ex]current bounding box.center)},scale=.6]
        \draw (0,0) circle (3.5);

        \coordinate[point] (c3) at (108:2.5);
        \coordinate[point] (c2) at (180:2);
        \coordinate[point] (c1) at (252:2.5);
    
        \coordinate[point] (x) at (216:2.8);
        \coordinate[point] (y) at (144:2.8);

        \begin{scope}[every node/.style={transform shape}]
            \node[square] at (x) {};
            \node[square] at (y) {};
    
            \draw pic at (x) {H={-45}{.5}};
            \draw pic at (x) {H={45}{.5}};
                
            \draw pic at (y) {H={-45}{.5}};
            \draw pic at (y) {H={45}{.5}};
    
            \draw pic at (c1) {H={-60}{.5}};
            \draw pic at (c1) {H={60}{.5}};
            \draw pic at (c1) {H={0}{.5}};
    
            \draw pic at (c3) {H={-45}{.5}};
            \draw pic at (c3) {H={45}{.5}};
            \draw pic at (c3) {H={0}{.353}};
    
            \draw pic at (c2) {H={-60}{.5}};
            \draw pic at (c2) {H={60}{.5}};
            \draw pic at (c2) {H={0}{.5}};
        \end{scope}

        \node[right, xshift=.5em] at (c1) {$c_1^X$};
        \node[left, xshift=-.5em] at (c2) {$c_2^X$};
        \node[right, xshift=.5em] at (c3) {$c_3^X$};
    
        \node[below right, xshift=.5em] at (x) {$x$};
        \node[above right, xshift=.5em] at (y) {$y$};
    
        \draw[dashed] (324:3.3) to ($(324:1.8) + (144:4.5)$);
        \node[above] at ($(324:1.8) + (324:1.5)$) {$c_5^X$};
    
        \draw[dashed]  (36:3.3) to ($(36:1.8) + (216:4.5)$);
        \node[below] at ($(36:1.8) + (36:1.5)$) {$c_4^X$};

        \node (xx) at (x) {};
        \node (yy) at (y) {};
        \draw[thick,red,->] (xx) -- (216:1.5);
        \draw[thick,red,->] (yy) -- (144:1.5);

        \begin{pgfonlayer}{bg}  
            \draw[decorate,decoration=zigzag,red,thick]  ($(c3) + (90:.8)$) -- ($(c1)+(-90:.8)$);
        \end{pgfonlayer}
      \end{tikzpicture} 
      \caption{Isolated 1-scar stratum.}
        \label{fig:counting_data_deg1_P2_to_P3_1_scar}
        \end{subfigure}\hfill
        \caption{$\PQL$ for the quasimap counting data $\mathcal D_{(1)}$.}
    \end{figure}
     Its position is hence fixed, which imposes two more equations.
     The space of possible once-freckled quasimaps has therefore dimension $7$. 
     However, we must impose 9 further equations since the remaining three points are each mapped to a point, cf.\ Figure \ref{fig:counting_data_deg1_P2_to_P3_1_freckle}.
     We thus conclude that there cannot be a freckle.

     However, the quasimap can have a scar.
     Suppose that the scar passes through two of the fixed points, say through $c_1^X$ and $c_3^X$.
     Note that the scar intersects the two lines $\{c_i^X\}_{i=4}^5$.
     Hence on the scar all equations but the three equations demanded at $c_2^X$ are satisfied, cf.\ Figure \ref{fig:counting_data_deg1_P2_to_P3_1_scar}. 
     Since away from the scar the quasi map is constant, the remaining three equations fixes the scarred 
     quasimap uniquely.
     
     Since the scar can pass through any two of the three points $\{c_i^X\}_{i=1}^3$, there exists three such one-scar strata and hence three proper quasimaps.
     This allows us to conclude that 
     \[
     \KM = \GLSM - \PQM = 4 - 3\cdot \underbrace{1}_{\mr{scar\, configurations}} = 1.
     \]
 \end{proof}

\textbf{Situation $D_{(2)}$.}
In this case the $\GLSM$ number is given by
 \begin{equation}
     \GLSM(\CP^2,\CP^3,\mathcal D_{(2)}) = \int_{\CP^{11} \times \CP^1_{(4)} \times \CP^1_{(5)}} H^{8} (H + p_4^* h)^3 (H + p_5^* h)^2 = 6.
 \end{equation}
 \begin{proposition}\label{prop:CP2->CP3 non-trivial cycles 2}
     In the present situation, there exists again an unique holomorphic map, i.e.\ 
     \[
        \KM(\CP^2,\CP^3,\mathcal D_{(2)}) = 1.
     \]
 \end{proposition}
 \begin{proof}[Proof by counting proper quasimaps]
 By the same arguments as for the counting quasimap data $\mathcal D_{(1)}$, a proper quasimap cannot have a freckle.
 If it would have a freckle, it must again sit at the intersection of the lines $c_4^X$ and $c_5^X$.
 Such strata has again dimension 7 while we still have to impose 8 equations.

 As before, a proper quasimap can admit a scar.
 There are two situations to consider:\bigskip

 \noindent \textit{Case 1.} Assume that the scar passes by the two points which are mapped to points, say $c_1^X$ and $c_2^X$, cf.\ Figure \ref{fig:counting_data_deg1_P2_to_P3_2_scar}. 
 \begin{figure}[H]
     \centering
     \begin{subfigure}[t]{.3\textwidth}
     \begin{tikzpicture}[baseline={([yshift=-.5ex]current bounding box.center)},scale=.6]
    \draw (0,0) circle (3.5);

    \coordinate[point] (c3) at (108:2.5);
    \coordinate[point] (c2) at (180:2);
    \coordinate[point] (c1) at (252:2.5);

    \coordinate[point] (x) at (216:.8);
    \coordinate[point] (y) at (144:1);
    
    \begin{scope}[every node/.style={transform shape}]
        \node[square] at (x) {};
        \node[square] at (y) {};

        \draw pic at (x) {H={-45}{.5}};
        \draw pic at (x) {H={45}{.5}};
        \draw pic at (x) {H={0}{.353}};

        \draw pic at (y) {H={-45}{.5}};
        \draw pic at (y) {H={45}{.5}};

        \draw pic at (c1) {H={-60}{.5}};
        \draw pic at (c1) {H={60}{.5}};
        \draw pic at (c1) {H={0}{.5}};

        \draw pic at (c3) {H={-45}{.5}};
        \draw pic at (c3) {H={45}{.5}};

        \draw pic at (c2) {H={-60}{.5}};
        \draw pic at (c2) {H={60}{.5}};
        \draw pic at (c2) {H={0}{.5}};
    \end{scope}

    \node[left, xshift=-.5em] at (c1) {$c_1^X$};
    \node[left, xshift=-.5em] at (c2) {$c_2^X$};
    \node[right, xshift=.5em] at (c3) {$c_3^X$};

    \node[below, yshift=-.5em] at (x) {$x$};
    \node[above, yshift=.5em] at (y) {$y$};

    \draw[dashed] (324:3.3) to ($(324:1.8) + (144:4.5)$);
    \node[above] at ($(324:1.8) + (324:1.5)$) {$c_5^X$};

    \draw[dashed] (36:3.3) to ($(36:1.8) + (216:4.5)$);
    \node[below] at ($(36:1.8) + (36:1.5)$) {$c_4^X$};

    \node (xx) at (x) {};
    \node (yy) at (y) {};
    \draw[red,thick,->] (xx) -- (216:2);
    \draw[red,thick,->] (yy) -- (144:2.3);

    \begin{pgfonlayer}{bg}    
        \draw[red,thick,decorate,decoration=zigzag] (c1) to[bend left] (c2);
        \draw[red,thick,decorate,decoration=zigzag] (c1) to[bend left] ($(c1) + (-45:1)$);
        \draw[red,thick,decorate,decoration=zigzag] (c2) to[bend left] ($(c3) + (180:1)$);
    \end{pgfonlayer}
  \end{tikzpicture} 
     \caption{Degenerate 1-scar stratum.}
     \label{fig:counting_data_deg1_P2_to_P3_2_scar}
     \end{subfigure}\hfill
     \begin{subfigure}[t]{.3\textwidth}
              \begin{tikzpicture}[baseline={([yshift=-.5ex]current bounding box.center)},scale=.6]
    \draw (0,0) circle (3.5);

    \coordinate[point] (c3) at (108:2.5);
    \coordinate[point] (c2) at (180:2);
    \coordinate[point] (c1) at (252:2.5);

    \coordinate[point] (x) at (216:2.5);
    \coordinate[point] (y) at (144:2.5);
    
    \begin{scope}[every node/.style={transform shape}]
        \node[square] at (x) {};
        \node[square] at (y) {};

        \draw pic at (x) {H={-45}{.5}};
        \draw pic at (x) {H={45}{.5}};
        \draw pic at (x) {H={0}{.353}};

        \draw pic at (y) {H={-45}{.5}};
        \draw pic at (y) {H={45}{.5}};

        \draw pic at (c1) {H={-60}{.5}};
        \draw pic at (c1) {H={60}{.5}};
        \draw pic at (c1) {H={0}{.5}};

        \draw pic at (c3) {H={-45}{.5}};
        \draw pic at (c3) {H={45}{.5}};

        \draw pic at (c2) {H={-60}{.5}};
        \draw pic at (c2) {H={60}{.5}};
        \draw pic at (c2) {H={0}{.5}};
    \end{scope}
    
    \node[right, xshift=.5em] at (c1) {$c_1^X$};
    \node[left, xshift=-.5em] at (c2) {$c_2^X$};
    \node[right, xshift=.5em] at (c3) {$c_3^X$};

    \node[below, yshift=-.5em] at (x) {$x$};
    \node[above, yshift=.5em] at (y) {$y$};

    \draw[dashed] (324:3.3) to ($(324:1.8) + (144:4.5)$);
    \node[above] at ($(324:1.8) + (324:1.5)$) {$c_5^X$};

    \draw[dashed] (36:3.3) to ($(36:1.8) + (216:4.5)$);
    \node[below] at ($(36:1.8) + (36:1.5)$) {$c_4^X$};

    \node (xx) at (x) {};
    \node (yy) at (y) {};
    \draw[thick,red,->] (xx) -- (216:1.5);
    \draw[thick,red,->] (yy) -- (144:1.5);

    \begin{pgfonlayer}{bg}    
        \draw[red,thick,decorate,decoration=zigzag] ($(c3) + (90:.9)$) -- ($(c1)+(-90:.9)$);
    \end{pgfonlayer}
  \end{tikzpicture} 
     \caption{Isolated 1-scar stratum.}
     \label{fig:counting_data_deg1_P2_to_P3_2_scar2}
     \end{subfigure}\hfill
             \begin{subfigure}[t]{.3\textwidth}
            \begin{tikzpicture}[baseline={([yshift=-.5ex]current bounding box.center)},scale=.6]
        \draw (0,0) circle (3.5);

        \coordinate[point] (c3) at (108:2.5);
        \coordinate[point] (c2) at (180:2);
        \coordinate[point] (c1) at (252:2.5);
    
        \coordinate[point] (x) at (36:1.8);
        \coordinate[point] (y) at (324:1.8);
        
        \begin{scope}[every node/.style={transform shape}]
        \node[square] at (x) {};
        \node[square] at (y) {};

        \draw pic at (x) {H={-45}{.5}};
        \draw pic at (x) {H={45}{.5}};
        \draw pic at (x) {H={0}{.353}};

        \draw pic at (y) {H={-45}{.5}};
        \draw pic at (y) {H={45}{.5}};

        \draw pic at (c1) {H={-60}{.5}};
        \draw pic at (c1) {H={60}{.5}};
        \draw pic at (c1) {H={0}{.5}};

        \draw pic at (c3) {H={-45}{.5}};
        \draw pic at (c3) {H={45}{.5}};

        \draw pic at (c2) {H={-60}{.5}};
        \draw pic at (c2) {H={60}{.5}};
        \draw pic at (c2) {H={0}{.5}};
        \end{scope}

        \node[left, xshift=-.5em] at (c1) {$c_1^X$};
        \node[left, xshift=-.5em] at (c2) {$c_2^X$};
        \node[left, xshift=-.5em] at (c3) {$c_3^X$};
    
        \node[below, yshift=-.5em] at (x) {$x$};
        \node[below, yshift=-.5em] at (y) {$y$};
    
        \draw[dashed] (y) to ($(y) + (144:4.5)$);
        \draw[dashed] (y) to ($(y) + (324:1.5)$);
        \node[above] at ($(y) + (324:1.5)$) {$c_5^X$};
    
        \draw[dashed] (x) to ($(x) + (216:4.5)$);
        \draw[dashed] (x) to ($(x) + (36:1.5)$);
        \node[below] at ($(x) + (36:1.5)$) {$c_4^X$};

        \draw[red,yshift=-0.2em,scale=1.5] node at (0,0) {*};

        \node (xx) at (x) {};
        \node (yy) at (y) {};
        \draw[red,thick,->] (xx) -- ($(0,0) + (36:.3)$);
        \draw[red,thick,->] (yy) -- ($(0,0) + (-36:.3)$);
      \end{tikzpicture} 
        \caption{Would-be freckle stratum.}
        \label{fig:counting_data_deg1_P2_to_P3_2_freckle}
        \end{subfigure}
        \caption{$\PQL$ of the counting data $\mathcal D_{(2)}$.}
 \end{figure}
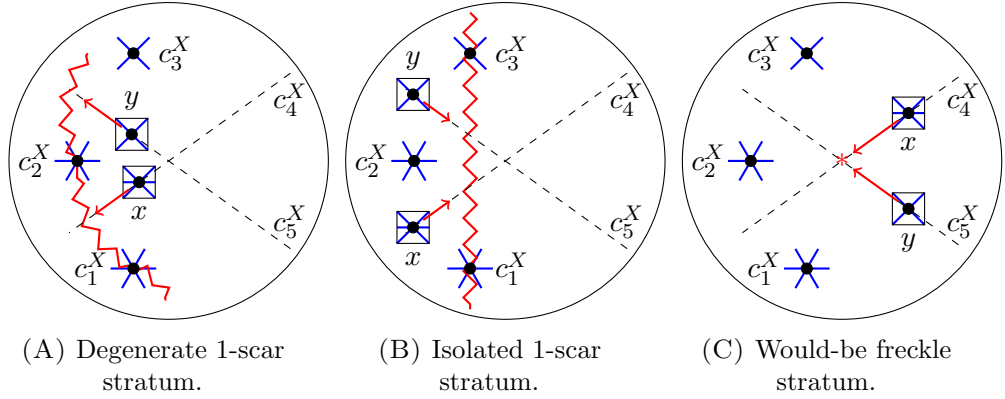
On the complement of the scar the quasimap is a constant map to $\CP^3$ upon which we impose 2 equations. 
This strata is thus a $Z = \CP^1 \subset \QMap_1(\CP^2,\CP^3)$ and its contribution must be calculated by excess intersection theory. 
Since the position of $x$ and $y$ are fixed to be the intersection points of the scar and the line $c_4^X$ resp.\ $c_5^X$, one has $c(E \rvert_Z) = (1 + \zeta)^{13}$ where $\zeta$ denotes the generator of $H^2(Z)$.
Furthermore, it follows that $c(\Var)\rvert_Z = c(\QMap_1(\CP^2,\CP^3)|_Z) = (1 + \zeta)^{12}$.
This allows us to calculate the contribution of $Z$ according to equation \eqref{eq: compute chern class B}:
\begin{equation}
    \int_Z c(B_Z) = \int_{Z} \frac{c(E\rvert_Z)c(Z)}{c(\Var)\rvert_Z} = \int_{\CP^1} \frac{(1+\zeta)^{13}(1+\zeta)^2}{(1+\zeta)^{12}} = 3.
\end{equation}

\noindent\textit{Case 2.} Assume that the scar passes by the point which is mapped to a line and one of the points that is mapped to a point, say through $c_3^X$ and $c_1^X$, cf.\ Figure \ref{fig:counting_data_deg1_P2_to_P3_2_scar2}.

On the complement of the scar the quasimap is uniquely fixed by imposing the three equations at $c_2^X$. 
Such a stratum hence contributes with 1 and there exist two such strata (the scar can pass through $c_3^X$ and either $c_1^X$ or $c_2^X$).

This allows us to conclude 
\begin{equation}
    \KM = \GLSM - \PQM = 6 - \underbrace{3}_{\mr{scar\,through\,}c_1^X,c_2^X} - 2\cdot \underbrace{1}_{\mr{scar\,through\,}c_3^X,c_{1\,\mr{or}\, 2}^X} = 1.
\end{equation}
\end{proof}

\textbf{Situation $D_{(3)}$.}
In this situation, the $\GLSM$ number is given by 
\begin{equation}
    \GLSM(\CP^2,\CP^3,\mathcal D_{(3)}) = \int_{\CP^{11}\times \CP^1_{(4)} \times \CP^1_{(5)}} H^7(H + p_4^* h)^3(H + p_5^* h)^3 = 9.
\end{equation}
\begin{proposition} 
    In the present situation, there exists again an unique holomorphic map, i.e.\
    \[
        \KM(\CP^2,\CP^3,\mathcal D_{(3)}) = 1.
    \]
\end{proposition}
\begin{proof}[Proof by counting proper quasimaps]
Unlike for the quasimap counting data $\mathcal D_{(1)}$ and $\mathcal D_{(2)}$, a proper quasimap now may admit a freckle: As before, the frecke sits at the intersection of the lines $c_4^X$ and $c_5^X$. 
The one-freckle stratum has hence dimension 7. 
But unlike to before, we now only impose exactly 7 equations which determines the once-freckled map uniquely.

On the other hand, the quasimap may admit a scar. 
There are again two cases to consider:\bigskip

\noindent\textit{Case 1.} The scar passes through the two points which are mapped to a line,  $c_1^X$ and $c_3^X$, cf.\ Figure \ref{fig:counting_data_deg1_P2_to_P3_3_scar1}.
Since the quasimap away from the scar is constant, it is uniquely fixed by the remaining equations imposed at $c_2^X$. 
This stratum hence contributes with 1.
     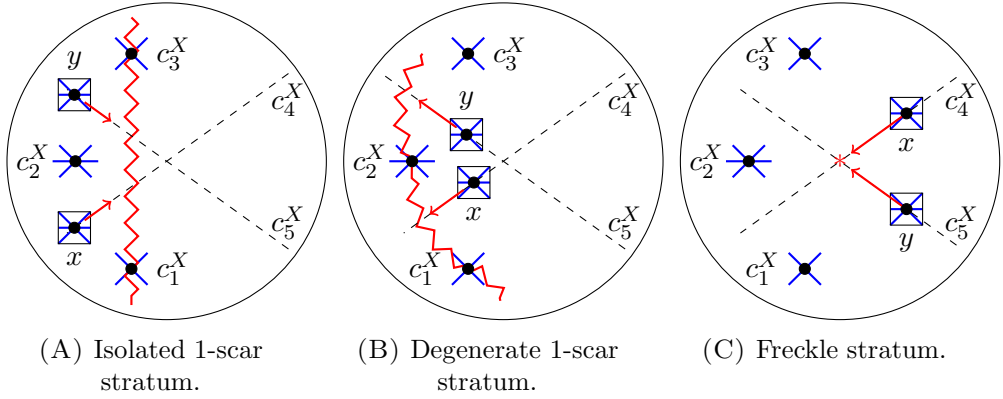
\begin{figure}[H] 
        \centering
        \begin{subfigure}[t]{0.3\textwidth}
        \begin{tikzpicture}[baseline={([yshift=-.5ex]current bounding box.center)},scale=.6]
        \draw (0,0) circle (3.5);

        \coordinate[point] (c3) at (108:2.5);
        \coordinate[point] (c2) at (180:2);
        \coordinate[point] (c1) at (252:2.5);
    
        \coordinate[point] (x) at (216:2.5);
        \coordinate[point] (y) at (144:2.5);

        \begin{scope}[every node/.style={transform shape}]
        \node[square] at (x) {};
        \node[square] at (y) {};

        \draw pic at (x) {H={-45}{.5}};
        \draw pic at (x) {H={45}{.5}};
        \draw pic at (x) {H={0}{.353}};

        \draw pic at (y) {H={-45}{.5}};
        \draw pic at (y) {H={45}{.5}};
        \draw pic at (y) {H={0}{.353}};

        \draw pic at (c1) {H={-45}{.5}};
        \draw pic at (c1) {H={45}{.5}};

        \draw pic at (c3) {H={-45}{.5}};
        \draw pic at (c3) {H={45}{.5}};

        \draw pic at (c2) {H={-45}{.5}};
        \draw pic at (c2) {H={45}{.5}};
        \draw pic at (c2) {H={0}{.5}};
        \end{scope}

        \node[right, xshift=.5em] at (c1) {$c_1^X$};
        \node[left, xshift=-.5em] at (c2) {$c_2^X$};
        \node[right, xshift=.5em] at (c3) {$c_3^X$};
    
        \node[below, yshift=-.5em] at (x) {$x$};
        \node[above, yshift=.5em] at (y) {$y$};
    
        \draw[dashed] (324:3.3) to ($(324:1.8) + (144:4.5)$);
        \node[above] at ($(324:1.8) + (324:1.5)$) {$c_5^X$};
    
        \draw[dashed]  (36:3.3) to ($(36:1.8) + (216:4.5)$);
        \node[below] at ($(36:1.8) + (36:1.5)$) {$c_4^X$};

        \node (xx) at (x) {};
        \node (yy) at (y) {};
        \draw[thick,red,->] (xx) -- (216:1.5);
        \draw[thick,red,->] (yy) -- (144:1.5);

        \begin{pgfonlayer}{bg}  
            \draw[decorate,decoration=zigzag,red,thick]  ($(c3) + (90:.8)$) -- ($(c1)+(-90:.8)$);
        \end{pgfonlayer}
      \end{tikzpicture} 
        \caption{Isolated 1-scar stratum.} 
        \label{fig:counting_data_deg1_P2_to_P3_3_scar1}
        \end{subfigure}\hfill
        \begin{subfigure}[t]{0.3\textwidth}
            \begin{tikzpicture}[baseline={([yshift=-.5ex]current bounding box.center)},scale=.6]
            \draw (0,0) circle (3.5);

            \coordinate[point] (c3) at (108:2.5);
            \coordinate[point] (c2) at (180:2);
            \coordinate[point] (c1) at (252:2.5);
            
            \coordinate[point] (x) at (216:.8);
            \coordinate[point] (y) at (144:1);

            \begin{scope}[every node/.style={transform shape}]
            \node[square] at (x) {};
            \node[square] at (y) {};
        
            \draw pic at (x) {H={-45}{.5}};
            \draw pic at (x) {H={45}{.5}};
            \draw pic at (x) {H={0}{.353}};
        
            \draw pic at (y) {H={-45}{.5}};
            \draw pic at (y) {H={45}{.5}};
            \draw pic at (y) {H={0}{.353}};
        
            \draw pic at (c1) {H={-45}{.5}};
            \draw pic at (c1) {H={45}{.5}};
        
            \draw pic at (c3) {H={-45}{.5}};
            \draw pic at (c3) {H={45}{.5}};
        
            \draw pic at (c2) {H={-45}{.5}};
            \draw pic at (c2) {H={45}{.5}};
            \draw pic at (c2) {H={0}{.5}};
            \end{scope}

            \node[left, xshift=-.5em] at (c1) {$c_1^X$};
            \node[left, xshift=-.5em] at (c2) {$c_2^X$};
            \node[right, xshift=.5em] at (c3) {$c_3^X$};
        
            \node[below, yshift=-.5em] at (x) {$x$};
            \node[above, yshift=.5em] at (y) {$y$};

            \draw[dashed] (324:3.3) to ($(324:1.8) + (144:5)$);
            \node[above] at ($(324:1.8) + (324:1.5)$) {$c_5^X$};
        
            \draw[dashed] (36:3.3) to ($(36:1.8) + (216:4.5)$);
            \node[below] at ($(36:1.8) + (36:1.5)$) {$c_4^X$};

            \node (xx) at (x) {};
            \node (yy) at (y) {};
            \draw[thick,red,->] (xx) -- (216:2);
            \draw[thick,red,->] (yy) -- (144:2.3);

            \begin{pgfonlayer}{bg}    
                \draw[red,thick,decorate,decoration=zigzag] (c1) to[bend left] (c2);
                \draw[red,thick,decorate,decoration=zigzag] (c1) to[bend left] ($(c1) + (-45:1)$);
                \draw[red,thick,decorate,decoration=zigzag] (c2) to[bend left] ($(c3) + (180:1)$);
            \end{pgfonlayer}
     \end{tikzpicture} 
     \caption{Degenerate 1-scar stratum.}
     \label{fig:counting_data_deg1_P2_to_P3_3_scar2}
        \end{subfigure}\hfill
        \begin{subfigure}[t]{.3\textwidth}
            \begin{tikzpicture}[baseline={([yshift=-.5ex]current bounding box.center)},scale=.6]
        \draw (0,0) circle (3.5);

        \coordinate[point] (c3) at (108:2.5);
        \coordinate[point] (c2) at (180:2);
        \coordinate[point] (c1) at (252:2.5);
    
        \coordinate[point] (x) at (36:1.8);
        \coordinate[point] (y) at (324:1.8);
        
        \begin{scope}[every node/.style={transform shape}]
        \node[square] at (x) {};
        \node[square] at (y) {};

        \draw pic at (x) {H={-45}{.5}};
        \draw pic at (x) {H={45}{.5}};
        \draw pic at (x) {H={0}{.353}};

        \draw pic at (y) {H={-45}{.5}};
        \draw pic at (y) {H={45}{.5}};
        \draw pic at (y) {H={0}{.353}};

        \draw pic at (c1) {H={-45}{.5}};
        \draw pic at (c1) {H={45}{.5}};

        \draw pic at (c3) {H={-45}{.5}};
        \draw pic at (c3) {H={45}{.5}};

        \draw pic at (c2) {H={-45}{.5}};
        \draw pic at (c2) {H={45}{.5}};
        \draw pic at (c2) {H={0}{.5}};
        \end{scope}

        \node[left, xshift=-.5em] at (c1) {$c_1^X$};
        \node[left, xshift=-.5em] at (c2) {$c_2^X$};
        \node[left, xshift=-.5em] at (c3) {$c_3^X$};
    
        \node[below, yshift=-.5em] at (x) {$x$};
        \node[below, yshift=-.5em] at (y) {$y$};
    
        \draw[dashed] (y) to ($(y) + (144:4.5)$);
        \draw[dashed] (y) to ($(y) + (324:1.5)$);
        \node[above] at ($(y) + (324:1.5)$) {$c_5^X$};
    
        \draw[dashed] (x) to ($(x) + (216:4.5)$);
        \draw[dashed] (x) to ($(x) + (36:1.5)$);
        \node[below] at ($(x) + (36:1.5)$) {$c_4^X$};

        \draw[red,yshift=-0.2em,scale=1.5] node at (0,0) {*};

        \draw[red,thick,->] (x) -- ($(0,0) + (36:.3)$);
        \draw[red,thick,->] (y) -- ($(0,0) + (-36:.3)$);
      \end{tikzpicture} 
        \caption{Freckle stratum.} 
        \label{fig:counting_data_deg1_P2_to_P3_3_freckle}
        \end{subfigure}
        \caption{$\PQL$ of the counting data $\mathcal D_{(3)}$.}
    \end{figure}

\noindent\textit{Case 2.} The scar passes through the point which is mapped to a point and through one of the points which is mapped to a line, say $c_1^X$, cf. Figure \ref{fig:counting_data_deg1_P2_to_P3_3_scar2}.
In this case, we impose only two equations on the complement of the scar, and hence the stratum is a $\CP^1 \subset \QMap_1(\CP^2,\CP^3,\mathcal D_{(2)})$.
This situation is equivalent to the situation we encountered in the proof of Proposition \ref{prop:CP2->CP3 non-trivial cycles 2}. 
As follows from an analogous computation, the excess contribution of this stratum is given by $3$. 
Note that we have two such strata, namely one when the scar passes by $c_2^X$ and $c_1^X$, and another when the scar passes by $c_2^X$ and $c_3^X$.\bigskip

In total we conclude
\begin{equation}
    \KM = \GLSM - \PQM = 9 - \underbrace{1}_{\mr{freckle}} - \underbrace{1}_{\mr{scar\, through\,}c_1^X,c_3^X} - 2\cdot \underbrace{3}_{\mr{scar\,through\,}c_2^X, c_{1\,\mr{or}\,3}^X} = 1.
\end{equation}

\end{proof}

\section{Smooth conjecture}

\subsection{The conjecture}
Fix cycles $c_i^X$ in the source $X=\CP^k$ and cycles $c_i^Y$ in the target $Y=\CP^n$, $i=1,\ldots,l$. 
We are interested in the solution
\begin{equation}\label{N enum}
\KM(\CP^k,\CP^n;\{c_i^X,c_i^Y\}|d)
\end{equation} 
of Enumerative Problem \ref{enum prob: counting hol maps} for degree $d$ holomorphic maps.

Let
$$ \mr{ev}_i\colon \mr{Maps}_d(\CP^k,\CP^n)\times \CP^k_{1}\times \cdots \times \CP^k_{l}\ra \CP^n$$
be the evaluation of a map at the $i$-th source point, $i=1,\ldots,l$.

\begin{conjecture}[Smooth conjecture]\label{conj:smooth_conjecture}\mbox{}

\begin{itemize}
    \item (Strong version.) For any \emph{smooth} representatives $\alpha_i^X\in \Omega_{cl}(\CP^k), \alpha_i^Y\in \Omega_{cl}(\CP^n)$ of Poincar\'e dual cohomology classes of the homology classes of cycles $c_i^{X,Y}$, the integrals 
    \begin{subequations}
        \begin{align}
            N^\mr{C,S}(\alpha_1^X,\alpha_1^Y;\cdots; \alpha_l^X,\alpha_l^Y|d) &\defeq
        \hspace{-24pt}\int\limits_{\mr{Maps}_d(\CP^k,\CP^n)\times c_1^X\times \cdots \times c_l^X}
        \prod_{i=1}^l \mr{ev}_i^*(\alpha^Y_i), \label{smooth eq3}\\
         N^\mr{S,S}(\alpha_1^X,\alpha_1^Y;\cdots; \alpha_l^X,\alpha_l^Y|d) &\defeq
        \hspace{-24pt}\int\limits_{\mr{Maps}_d(\CP^k,\CP^n)\times \CP^k_1\times \cdots \times \CP^k_l}
        \prod_{i=1}^l (\alpha_i^X \wedge \mr{ev}_i^*(\alpha^Y_i) )
        \end{align}
    \end{subequations}
    are both convergent and equal to the number (\ref{N enum}).\footnote{Superscripts $C,S$ stand for ``cycles on the source, smooth representatives on the target'';  $S,S$ stands for ``smooth representatives on both source and target.''}
    \item (Specialized version.) Assume that cycles $c_i^X,c_i^Y$ are in complex codimension $n_i^X,n_i^Y$ in the source/target and assume that their homology classes are $d_i^{X,Y}$ times the generator of the respective homology group.\footnote{Recall that $H_{2j}(\CP^n,\mathbb{Z})=\mathbb{Z}$ for $j=0,\ldots,n$ and zero otherwise. } Then one has that the integrals 
    \begin{subequations}
        \begin{align}
            &\prod_{i=1}^l d_i^Y\cdot \int_{\mr{Maps}_d(\CP^k,\CP^n)\times c_1^X\times \cdots \times c_l^X}
        \prod_{i=1}^l \mr{ev}_i^*(\omega_{Y}^{\wedge n_i^Y}), \label{smooth eq2}\\
        &\prod_{i=1}^l (d_i^X d_i^Y)\cdot\int_{\mr{Maps}_d(\CP^k,\CP^n)\times \CP^k_1\times \cdots \times \CP^k_l}
        \prod_{i=1}^l (\omega_X^{\wedge n_i^X} \wedge \mr{ev}_i^*(\omega_Y^{\wedge n_i^Y}) )
        \end{align}
    \end{subequations}
    both exist and are equal to the number (\ref{N enum}). Here $\omega_X,\omega_Y$ are the Fubini-Study 2-forms on the source and the target, respectively.
\end{itemize}
\end{conjecture}

\subsection{Numerical evidence for the smooth conjecture}
For simplicity, we will consider a modified version of the integral formula (\ref{smooth eq3}). 
We will consider the case where $c_i^X$ are points for $i=1,\ldots, k+2$, and $c^X_i=\CP^k$ for $i>k+2$.
We reduce the space of maps by demanding that a map $f$ sends the cycles $c_i^X$ (for $i=1,\ldots,k+2$) 
to the cycles $c_i^Y$.  We will call the resulting space of maps the \emph{reduced} space of maps $\Map^{\rm red}_d(\CP^k,\CP^n)$. It is a section of the $\PSL(k+1,\CC)$-action on $\Map_d(\CP^k,\CP^n)$ (cf. footnote \ref{footnote: transitivity of PSL(k+1,C) acting on P^k}). Transitioning to the integral over $\Map^{\rm red}_d(\CP^k,\CP^n)$ 
 corresponds to choosing $\alpha_i^Y=\delta_{c_i^Y}$ to be distributional delta-forms on the cycles $c_i^Y$ in (\ref{smooth eq3}) for $i=1,\ldots,k+2$ and performing the fiber integral in $\Map_d(\CP^k,\CP^n)\ra \Map_d(\CP^k,\CP^n)/\PSL(k+1,\CC)\simeq \Map_d^\mr{red}(\CP^k,\CP^n)$. 

In particular, the integrals in Conjecture \ref{conj:smooth_conjecture} can now be expressed over $\Map^{\rm red}_d(\CP^k,\CP^n)$, e.g. 
\[
N^{C,S}(\alpha_{k+3}^X,\alpha_{k+3}^Y;\dots;\alpha_l^X,\alpha_l^Y \mid d)=\int_{\Map^{\rm red}_d(\CP^k,\CP^n) \times \CP^k_{k+3} \times \dots \times \CP^k_l} \prod_{i = k+3}^l \ev_i^*(\alpha_i^Y).
\]

Let now $k=1$, $n=2$.
Let $(x^0:x^1)$ be homogeneous coordinates on the source 
and  $(y^0:y^1:y^2)$ homogeneous coordinates on the target. 
We denote by $z = x^1/x^0$ and $w = x^0/x^1$ the affine coordinates on the source. 
We will also write $d^2z = \frac{i}{2} dz\wedge d\bar z$ for the real measure on $\CC$.

In the examples below we choose the forms on the target $\alpha_i^Y$ to be appropriate powers of the Fubini-Study form: 
\begin{equation}\label{alpha=omega^codim}
\alpha_i^Y= \omega_Y^{\codim_\CC c_i^Y} \qquad \mr{for}\; i=k+3,\ldots,l. 
\end{equation} 
Finally, we will always consider the Fubini-Study form $\omega \in \Omega^2(\CP^2)$ to be normalized by
\[
\int_{\CP^2} \omega = 1.
\]

\subsubsection{1 or 2}
We start with the example of Section \ref{sec:1+1=2}.
Fix 
\[
c_1^X = (0:1), \quad c_2^X = (1:1),\quad c_3^X = (1:0), \quad c_4^X = \CP^1
\]
and
\[
c_1^Y = (0:y^1:y^2), \quad c_2^Y = (y^0:y^1:y^0+y^1),\quad c_3^Y = (1:0:0), \quad c_4^Y = q
\]
where $q$ is a fixed point in $\CP^2$. By our convention (\ref{alpha=omega^codim}), we set $\alpha^Y_4=\omega^2$.
The map counting data is depicted in Figure \ref{fig:1or2} below.
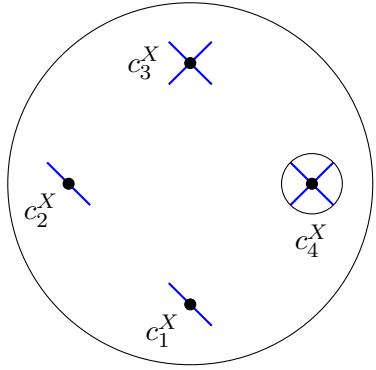
\begin{figure}[H]
    \centering
    \begin{tikzpicture}[baseline={([yshift=-.5ex]current bounding box.center)},scale=.8]
    \draw (0,0) circle (3);

    \coordinate[point] (c3) at (90:2);
    \coordinate[point] (c2) at (180:2);
    \coordinate[point] (c1) at (270:2);

    \coordinate[point] (c4) at (0:2);

    \begin{scope}[every node/.style={transform shape}]
        \draw (c4) circle (.5);

        \draw pic at (c1) {H={-45}{.5}};

        \draw pic at (c2) {H={-45}{.5}};

        \draw pic at (c3) {H={45}{.5}};
        \draw pic at (c3) {H={-45}{.5}};

        \draw pic at (c4) {H={45}{.5}};
        \draw pic at (c4) {H={-45}{.5}};
    \end{scope}

    \node[below left] at (c1) {$c_1^X$};
    \node[below left] at (c2) {$c_2^X$};
    \node[left] at ($(c3) + (180:.3)$) {$c_3^X$};
    \node[below] at ($(c4) + (-90:.5)$) {$c_4^X$};

  \end{tikzpicture} 
    \caption{Map counting data.}
    \label{fig:1or2}
\end{figure}
We impose the condition that $f \colon c_i^X \mapsto c_i^Y$ for $i=1,2,3$, i.e.\ we opt to compute the integral
\begin{equation}
N^{C,S}(1, \omega^2 \mid 1) = \int_{\Map_1^{\rm red}(\CP^1,\CP^2) \times \CP^1} \ev_4^*\omega^2.
\label{eq:NCS_1or2}
\end{equation}
As we have shown in Section \ref{sec:1+1=2}, cf. \eqref{eq: 1 + 1 = 2 qmap}, a quasimap which sends $c_i^X$ to $c_i^Y$, for $i=1,2,3$, is parametrized by 
\[
\underline{f}(x^0:x^1) = (x^0 : b x^1: (1+b) x^1) 
\]
and defines an proper map $f$ for $b \neq 0$.
Thus $\Map_1^{\rm red}(\CP^1,\CP^2) = \CC^*$.

Let us consider the chart $x_0 \neq 0$ and $b\neq 0$. 
In this chart, we can write the map $f$ as 
\begin{equation}
    f_{b}(z) = (b z, (1 + b)z).
\end{equation}
Since the point $b = 0$ has zero measure, we can compute \eqref{eq:NCS_1or2} by 
\begin{equation*}
    N^{C,S}(1 \mid \omega^2) = \int_{\CC^2} \ev(f_{b},z)^* \omega^2,
\end{equation*}
where now
\begin{equation}
    \ev(f_{b},z)^*\omega^2 = \frac{2}{\pi^2} \frac{\lvert z \rvert^2}{(1 + \lvert b z \rvert^2 + \lvert (1 + b)z \rvert^2 )^3}\ d^2b \wedge d^2z.
\end{equation}
This integral evaluates exactly to $1$, which is the $\KM$ number of the problem.

\subsubsection{1 or 4}
Let us revisit the example of Section \ref{sec:1+3=4}. 
Fix 
\[
c_1^X = (0:1), \quad c_2^X = (1:1),\quad c_3^X = (1:0), \quad c_4^X = c_5^X = \CP^1, 
\]
as well as
\[
 c_i^Y = \{ y^i = 0\}, \quad c_4^Y = p,\quad  c_5^Y = q,
\]
where $i=1,2,3$ and $p,q \in \CP^2$ are two fixed points. By convention (\ref{alpha=omega^codim}), we set $\alpha_4^Y=\alpha_5^Y=\omega^2$.

The quasimap counting data is depicted in Figure \ref{fig:1or4} below. 
\begin{figure}[H]
    \centering
    \begin{tikzpicture}[baseline={([yshift=-.5ex]current bounding box.center)},scale=.8]
    \draw (0,0) circle (3.3);

    \coordinate[point] (c3) at (108:2);
    \coordinate[point] (c2) at (180:2);
    \coordinate[point] (c1) at (252:2);

    \coordinate[point] (x) at (36:2);
    \coordinate[point] (y) at (-36:2);

    \begin{scope}[every node/.style={transform shape}]
        \draw (x) circle (.5);
        \draw (y) circle (.5);

        \draw pic at (x) {H={45}{.5}};
        \draw pic at (x) {H={-45}{.5}};

        \draw pic at (y) {H={45}{.5}};
        \draw pic at (y) {H={-45}{.5}};

        \draw pic at (c1) {H={-45}{.5}};

        \draw pic at (c2) {H={-45}{.5}};

        \draw pic at (c3) {H={-45}{.5}};
    \end{scope}

    \node[below left] at ($(c1)+ (.3,0)$) {$c_1^X$};
    \node[below left] at ($(c2)+(0,-.3)$) {$c_2^X$};
    \node[below left] at ($(c3)+(0,-.3)$) {$c_3^X$};

    \node[right] at ($(x)+(0:.5)$) {$c_4^X$};
    \node[right] at ($(y)+(0:.5)$) {$c_5^X$};
  \end{tikzpicture} 
    \caption{Map counting data.}
    \label{fig:1or4}
\end{figure}
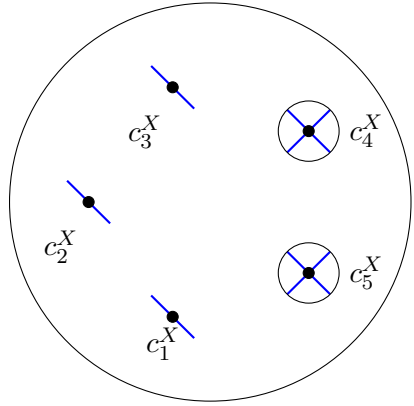
We again impose the condition that $f \colon c_i^X \mapsto c_i^Y$ for $i=1,2,3$, so that 
\begin{equation}
 N^{C,S}(1, \omega^2; 1,\omega^2 \mid 1) = \int_{\Map_1^{\rm red}(\CP^1,\CP^2) \times \CP^1} \ev_4^*\omega^2 \wedge \ev_5^* \omega^2.
\label{eq:NCS_1or4}
\end{equation}
A quasimap $\underline{f} \in \QMap_1(\CP^1,\CP^2)$ that maps $c_i^X$ to $c_i^Y$ for $i=1,2,3$, can be described by a point $(a:b:c) \in \CP^2$, cf.\ \eqref{eq: 1 + 3 = 4 qmap}
\[
\underline{f}(x^0:x^1) = (ax^0, b(x^1 - x^0), cx^1),
\]
which defines a proper map if at most one of the coefficients $a,b,c$ vanishes. 
The space of reduced maps is therefore $\Map_1^{\rm red}(\CP^1,\CP^2) = (\CC^*)^2$.

In the charts $x^1 \neq 0$ of the source $\CP^1$ and $c\neq 0$ of $\CP^2$, the map takes the form 
\begin{equation}
    f_{\alpha,\beta}(w) = (\alpha w, \beta (1-w)),
\end{equation}
where $\alpha = a/c$ and $\beta = b/c$.
One then finds that
\begin{equation}
    \ev(f_{\alpha,\beta},w_1)^*\omega^2 \wedge \ev(f_{\alpha,\beta},w_2)^*\omega^2 = \frac{4}{\pi^4} \frac{\lvert \alpha \rvert^2 \lvert \beta \rvert^2 \lvert w_1 - w_2 \rvert^2}{(F(w_1)F(w_2))^3} d\mu_{\CC^4},
\end{equation}
where 
\begin{align*}
    d\mu_{\CC^4} &= d^2w_1 \wedge d^2w_2 \wedge d^2\alpha \wedge d^2\beta
    \intertext{and}
    F(w) &= 1 + \lvert \alpha \rvert^2 \lvert w \rvert^2 + \lvert \beta \rvert^2 \lvert 1 - w \rvert^2.
\end{align*}
Recall that in order to omit quasimaps, we need to have $\alpha\cdot\beta \neq 0$.

It follows that 
\begin{equation}
    \begin{split}
        N^{C,S}(1, \omega^2; 1, \omega^2 \mid 1)= \int_{\CC^4} \ev(f_{\alpha,\beta},w_1)^*\omega^2 \wedge \ev(f_{\alpha,\beta},w_2)^*\omega^2 
    \end{split}
\end{equation}
which evaluates numerically\footnote{For instance, Monte Carlo method with $10^8$ sample points yields the value $0.995016$. 
}  
 to 
$ N^{C,S}\approx 1$
with $1$ being the $\KM$ number of the problem.

\subsubsection{1 or 16}
We now turn to the example of degree 2 maps from $\CP^1$ to $\CP^2$.
As in Section \ref{sec:1+3+12=16}, we fix 
\[
c_1^X = (0:1), \quad c_2^X = (1:1),\quad c_3^X = (1:0), \quad c_4^X = c_5^X = \CP^1, 
\]
and
\[
c_1^Y = (1:0:0), \quad c_2^Y = (0:1:0) \quad c_3^Y=(0:0:1), \quad c_4^Y = p, \quad c_5^Y = q,
\]
where $p$ and $q$ are again two points in $\CP^2$. By convention (\ref{alpha=omega^codim}), we again set $\alpha_4^Y=\alpha_5^Y=\omega^2$.
Consider the problem of maps which send $c_i^X$ to $c_i^Y$ as shown in Figure \ref{fig:1or16} below: 
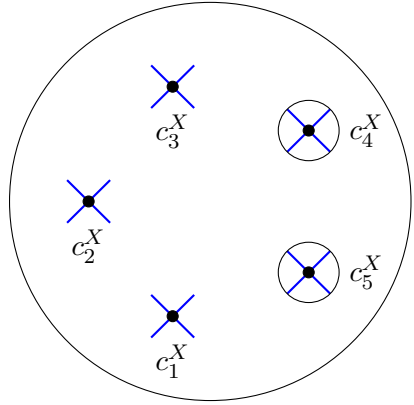
\begin{figure}[H]
    \centering
    \begin{tikzpicture}[baseline={([yshift=-.5ex]current bounding box.center)},scale=.8]
    \draw (0,0) circle (3.3);

    \coordinate[point] (c3) at (108:2);
    \coordinate[point] (c2) at (180:2);
    \coordinate[point] (c1) at (252:2);

    \coordinate[point] (x) at (36:2);
    \coordinate[point] (y) at (-36:2);

    \begin{scope}[every node/.style={transform shape}]
        \draw (x) circle (.5);
        \draw (y) circle (.5);

        \draw pic at (x) {H={45}{.5}};
        \draw pic at (x) {H={-45}{.5}};

        \draw pic at (y) {H={45}{.5}};
        \draw pic at (y) {H={-45}{.5}};

        \draw pic at (c1) {H={45}{.5}};
        \draw pic at (c1) {H={-45}{.5}};

        \draw pic at (c2) {H={45}{.5}};
        \draw pic at (c2) {H={-45}{.5}};

        \draw pic at (c3) {H={45}{.5}};
        \draw pic at (c3) {H={-45}{.5}};
    \end{scope}

    \node[below] at ($(c1)+ (0,-.3)$) {$c_1^X$};
    \node[below] at ($(c2)+(0,-.3)$) {$c_2^X$};
    \node[below] at ($(c3)+(0,-.3)$) {$c_3^X$};

    \node[right] at ($(x)+(0:.5)$) {$c_4^X$};
    \node[right] at ($(y)+(0:.5)$) {$c_5^X$};
  \end{tikzpicture}
    \caption{Map counting data.}
    \label{fig:1or16}
\end{figure}
Any such map can be parametrized by a point $(a:b:c)\in \CP^2$, cf.\ \eqref{eq: P1 -> P2 d=2 qmap}
\[
\underline{f}(x^0:x^1) = (ax^0 (x^0 - x^1): bx^0 x^1: c x^1 (x^0 - x^1)).
\]
In the charts $x^0 \neq 0$, $b \neq 0$, 
the map is given by 
\begin{equation}
  f_{\alpha,\beta}(z) = (\alpha(1 - 1/z), \beta (1 - z)), \qquad \alpha \neq 0, 
\end{equation}
where $\alpha = a/b$ and $\beta = c/b$.
The space of reduced maps is therefore $\Map^{\rm red}_2(\CP^1,\CP^2) = \CC^* \times \CC$ so that
\[
N^{C,S}(1,\omega^2; 1,\omega^2 \mid 2) = \int_{\Map^{\rm red}_2(\CP^1,\CP^2) \times \CP^1 \times \CP^1} \ev_4^*\omega^2 \wedge \ev_5^*\omega^2.
\]
One then finds 
\begin{multline}
    \ev(f_{\alpha,\beta},z_1)^*\omega^2 \wedge \ev(f_{\alpha,\beta},z_2)^*\omega^2 = \\ 
    =\frac{4}{\pi^4} \frac{\lvert \alpha \rvert^2\lvert \beta \rvert^2\lvert z_1 \rvert^2\lvert z_2 \rvert^2\lvert 1-z_1 \rvert^2\lvert 1-z_2 \rvert^2\lvert z_1-z_2\rvert^2}{(F(z_1)F(z_2))^3} d\mu_{\CC^4}
\end{multline}
where 
\[
F(z) = \lvert z \rvert^2 + \lvert \alpha \rvert^2 \lvert 1-z \rvert^2 + \lvert \beta \rvert^2 \lvert z \rvert^2 \lvert 1-z \rvert^2.
\]
Now
\begin{equation}
    \begin{split}
    N^{C,S}(1,\omega^2; 1,\omega^2 \mid 2) = \int_{\Map^{\rm red}_2(\CP^1,\CP^2) \times \CP^1 \times \CP^1} \ev_4^*\omega^2 \wedge \ev_5^*\omega^2.
    \end{split}
\end{equation}
It evaluates numerically\footnote{ Monte Carlo method with $10^8$ sample points gives the value $1.00522$.} to 
$N^{C,S}\approx 1$
with $1$  being again the KM number.

\section*{Concluding remarks}
(1) There is another approach to the enumerative numbers studied in this text based on higher-dimensional generalization of Morse-Bott-Floer theory. In particular, holomorphic maps of $k$-dimensional toric manifolds may be considered as $k$-Morse theory on the space of $k$-loops.
This approach may lead to a higher-dimensional generalization of the WDVV equation where the important intermediate result is the relation between 2-Morse theory and the algebra of the infrared \cite{GMW, KKS, Sukhanov}. We are going to explore it in a subsequent paper.

(2) We gave arguments showing that theories with higher-dimensional source may be considered similarly to theories with 1-dimensional source. Therefore, the natural question is tropicalization of such theories, similar to tropicalization of Gromov-Witten invariants \cite{Mikhalkin,Mikhalkin2,LL,Gross, GM}.

(3) The examples showed in this paper can be seen as showing the phenomena with 0-defects (freckles) and $2_\RR$-dimensional defects (scars) in 4-dimensional holomorphic gauged linear model.
Quantum field theory treatment of this problem will appear elsewhere.

\end{document}